\documentclass[12pt, english]{article}
\usepackage{amsthm, amstext, amssymb, amsmath}
\allowdisplaybreaks
\usepackage[T1]{fontenc}\usepackage[libertine]{newtxmath}\usepackage{libertine}
\usepackage{setspace}
\usepackage[letterpaper]{geometry}
\usepackage{url}
\makeatletter
\g@addto@macro{\UrlBreaks}{\UrlOrds}
\makeatother
\newtheorem{theorem}{Theorem}
\newtheorem{proposition}{Proposition}

\newtheorem{remark}{Remark}

\newtheorem{lemma}{Lemma}

\usepackage[authoryear]{natbib}
\usepackage{graphicx, graphics, caption, subcaption}
\usepackage{comment}
\usepackage[colorinlistoftodos]{todonotes}
\usepackage[shortlabels]{enumitem}
\renewcommand{\le}{\leqslant}
\renewcommand{\ge}{\geqslant}
\renewcommand{\leq}{\leqslant}
\renewcommand{\geq}{\geqslant}
\renewcommand{\epsilon}{\varepsilon}
 \usepackage[compact]{titlesec}
    \titlespacing{\section}{3pt}{3pt}{3pt}
    \titlespacing{\subsection}{3pt}{3pt}{3pt}
    \titlespacing{\subsubsection}{3pt}{3pt}{3pt}
\newcommand*{\threeemdashes}{%
  \makebox[1.5\fontcharwd\font`m]{% or \makebox[3em]
    \rlap{\textemdash}\kern.3em
    \cleaders\hbox{\kern-.2em\textemdash\kern-.2em}\hfill\kern.3em
    \llap{\textemdash}}%---
}
\usepackage[bottom]{footmisc}%footnote at the bottom of the page
 
\newif\ifCE
\ifCE
\usepackage[hang,flushmargin]{footmisc}
\renewcommand{\footnotesize}{\small}
\let\footnotesize\small
\usepackage{etoolbox}
\patchcmd{\thanks}{#1}{\protect\setstretch{1.5}{\small#1}}{}{}
\fi

\newif\ifsub

\ifsub
\usepackage{mathptmx}
\usepackage[moderate]{savetrees}
\usepackage[letterpaper]{geometry}
\fi

\usepackage[colorlinks=true, allcolors=blue]{hyperref}
\defcitealias{AbreuGul2000Ecta}{AG}
\newcommand{\AG}{\citetalias{AbreuGul2000Ecta}}

\begin{document}
\sloppy
\title{Reputational Bargaining with Ultimatum Opportunities}
\author{Mehmet Ekmekci\thanks{Department of Economics, Boston College, \href{mailto:ekmekci@bc.edu}{ekmekci@bc.edu}.} \and Hanzhe Zhang\thanks{Department of Economics, Michigan State University, \href{mailto:hanzhe@msu.edu}{hanzhe@msu.edu}.}}
\date{\today\footnote{We thank Martin Dufwenberg, Jon Eguia, Sel{\c{c}}uk {\"O}zyurt, Harry Di Pei, Phil Reny, and audience at various conferences and seminars for valuable suggestions. Zhang acknowledges the support of the National Science Foundation.}}
\maketitle
\begin{abstract}
\noindent We study two-sided reputational bargaining with opportunities to issue an ultimatum\threeemdashes{}threats to force dispute resolution. Each player is either a justified type, who never concedes and issues an ultimatum whenever an opportunity arrives, or an unjustified type, who can concede, wait, or bluff with an ultimatum. In equilibrium, the presence of ultimatum opportunities can harm or benefit a player by decelerating or accelerating reputation building. When only one player can issue an ultimatum, equilibrium play is unique. The hazard rate of dispute resolution is discontinuous and piecewise monotonic in time. As the probabilities of being justified vanish, agreement is immediate and efficient, and if the set of justifiable demands is rich, payoffs modify \citet{AbreuGul2000Ecta}, with the discount rate replaced by the ultimatum opportunity arrival rate if the former is smaller. When both players' ultimatum opportunities arrive sufficiently fast, there may exist multiple equilibria in which their reputations do not build up and negotiation lasts forever.

\noindent {\bf Keywords}: reputational bargaining, ultimatum, arbitration

\end{abstract}

\thispagestyle{empty}
\newpage

\setcounter{page}{1}
\section{Introduction}
In many bilateral negotiations, as a last resort to resolve their conflict, bargainers can threaten to ask for a final binding decision made by a neutral third party, e.g., a judge or arbitrator.\footnote{For example, 15 of the 20 largest US credit card issuers and seven of the eight largest cell phone companies include arbitration clauses in their contracts with consumers \citep{CFPB2015}, and Major League Baseball (MLB) and the National Hockey League (NHL) have used arbitration to resolve salary conflicts since the 1970s and 1990s, respectively.} However, most disputes are resolved before the impeding binding decision, so resolution through the final means is frequently leveraged as a strategic posture in the form of an ultimatum.\footnote{For example, 98\% of criminal cases and 97\% of civil lawsuits have been resolved before trial, and more than 80\% of financial arbitration cases and more than 95\% of NHL salary arbitration cases are settled before their scheduled hearings \citep{Gramlich2019, FINRA2020, NHLPA2020}.} To this end, we study the strategic and welfare implications of the presence of opportunities to \emph{ultimate/challenge}\threeemdashes{}i.e., to threaten a forced dispute resolution.\footnote{``Ultimate'' is less commonly used as a verb, so we use ``challenge'' synonymously.}

Our model incorporates ultimatum opportunities into the continuous-time war-of-attrition bargaining model of \citet{AbreuGul2000Ecta} ({\AG} henceforth). Players 1 (``he'') and 2 (``she'') negotiate to divide a unit pie. Each player is either justified or unjustified. A justified player demands a fixed share of the pie and never gives in to an offer smaller than their demand (corresponding to a behavioral type in \AG), and an unjustified player can demand any share and give in to any demand (corresponding to a rational type in \AG). Players announce their demands sequentially at the beginning of the game. Afterward, each player can continue the negotiation by holding on to the announced demand, or end the negotiation by giving in to the opposing demand or by challenging via an ultimatum. Upon being challenged, the opponent must respond either by giving in to the challenger's demand, or seeing the ultimatum, which leads to a nonstrategic third party to cast a decision. When the opponent sees the challenge, the third party observes the players' types and executes a division of the surplus that rules in favor of a justified challenger and against an unjustified one, and executes a predetermined compromise division if both players are unjustified.\footnote{The division rule if both players are justified is not strategically relevant, so we omit it.}

Our model can capture negotiations between parties who may threaten to resolve the dispute via a third party. For example, two parties claim conflicting terms for the salary of an athlete or the division of a dissolved company's remaining assets. At each subsequent instance, a party can insist on their demand, or end the negotiation by conceding to the opponent's demand or threatening to resolve the conflict in court. A justified party insists on a demand that can be supported by verifiable evidence, but collecting the evidence and bringing the case to court requires time (captured by the frictional arrival of ultimatum opportunities in the model) and physical or monetary costs (captured by the cost of going to court in the model). An unjustified party does not have any evidence but can threaten to take the case to court nonetheless. Whether a party could gather evidence is their private information. A justified party submits the case to the court once the needed evidence is collected, and an unjustified party can bluff at any time by submitting the case to the court. The opposing party must respond to a submitted case, either by agreeing to the plaintiff's demand out of court or by paying a cost to go to court. The court rules in favor of a plaintiff with evidence and against one without evidence.\footnote{Depending on the interpretation of the third party, applications include negotiation in the shadow of the law, in the presence of auditing possibilities or costly procurable evidence, and with the threat of outside option or war.}

In the model in which neither player has the opportunity to ultimate\threeemdashes{}the {\AG} model\threeemdashes{}the bargaining and reputation dynamics in the unique equilibrium are relatively simple. After players announce their demands, at most one player concedes with a positive probability at time 0. Afterward, both players concede at constant rates and their reputations\threeemdashes{}the opponent's beliefs about a player's being justified\threeemdashes{}increase exponentially at the respective constant concession rates until both reputations reach 1 at the same time, at which point no unjustified player is left in the game and justified players continue to hold on to their demands. When the initial probabilities of being justified tend to zero, typically outcomes will be efficient\threeemdashes{}and if, in addition, the set of possible justified demands is rich, the payoffs only depend on players' impatience, in the style of the alternating-offers bargaining game of \citet{Rubinstein1982Ecta}.

We start our analysis with the case in which only one player\threeemdashes{}player 1\threeemdashes{}has the opportunity to ultimate.\footnote{For example, in MLB and the NHL, essentially, only players can elect to have salary arbitration hearings; in civil lawsuits, usually only one side has the incentive to sue the other side; in price negotiations, typically either buyer or seller\threeemdashes{}but not both\threeemdashes{}waits for outside options; and in international conflicts, one side may consider aggression.} This case is a building block for the case in which both players have opportunities to ultimate, and most of the new economic forces from ultimatum opportunities are present and more transparent in this case. We start with the setting in which each player has a single justified demand (Section \ref{sec:equilibrium}). In the unique equilibrium (Theorem \ref{T:existenceandproperties}), as in the {\AG} equilibrium, at most one player concedes with a positive probability at time 0, both players concede at constant {\AG} concession rates, and both players' reputations increase to 1 at the same time. An unjustified player 1 ultimates with a positive and increasing hazard rate as long as player 2's reputation is not too high, and does not ultimate at all after player 2's reputation increases past a threshold. This creates an \emph{interdependence} in players' reputation-building processes that differs from {\AG}'s: The rate of change of player 1's reputation depends on player 2's reputation at each instance.

The overall hazard rate of dispute resolution is \emph{discontinuous} and \emph{piecewise monotonic} in time (Proposition \ref{prop:hazardrates}). Because player 1 does not challenge after player 2's reputation passes a threshold, there is a discontinuous drop in the equilibrium hazard rate of ultimatum usage by an unjustified player, and consequently a discontinuous drop in the total hazard rate of ultimatum usage. This leads us to a main qualitative finding of our model, which is testable and differs from existing reputational bargaining models. These features are indeed observed in the MLB and NHL salary arbitration cases we collected and analyzed.

Two forces in our model determine the speed and dynamics of reputation building. The first is reputation building by not conceding (as in {\AG}): Persisting longer in the negotiation increases a player's reputation. The second, which is new in our model, is reputation building or loss by not challenging. The combination of these two forces implies that reputations evolve according to Bernoulli differential equations, which include the exponential growth in the ultimatum-free {\AG} model as the special case.

What is the net equilibrium impact of ultimatum opportunities? Player 1's equilibrium payoff may be higher or lower with the presence of ultimatum opportunities (Proposition \ref{prop:gamma1}). Their presence can hurt player 1 by slowing reputation building, when an unjustified player 1 is expected to challenge at a lower rate than a justified player 1. This is because not challenging is evidence against his being justified (bad news). On the other hand, the presence of ultimatum opportunities can benefit player 1 by speeding up reputation building, when an unjustified player 1 is expected to challenge at a higher rate than a justified player 1 (good news). We characterize the set of prior beliefs for which the ultimatum opportunity benefits player 1. When this set is nonempty, it is an intermediate range of the prior reputations of player 1. When the ultimatum opportunities arrive sufficiently fast, the bad news effect dominates the good news effect, resulting in immediate agreement at the terms of player 2.

When the initial reputations approach zero, the equilibrium outcome is efficient with one of the players yielding to the opponent's demand at time zero with a very high probability (Proposition \ref{prop:11limit}). The identity of the loser\threeemdashes{}the player who concedes with a positive probability at time zero\threeemdashes{}and the division of the pie are determined by the discount rates, demands, and ultimatum opportunity arrival rate via a simple formula. The set of parameters for which player 1 loses expands with the ultimatum opportunity arrival rate; hence, ultimatum opportunities always hurt player 1 in the limit case of rationality. We also show that players' equilibrium payoffs in this limit case of rationality do not depend on the details of the decision rule employed by the court, while they do depend on these details in the cases away from the limit of rationality (as the unambiguous comparative statics results in Proposition \ref{prop:compstat1} show).

We then analyze the scenario in which a strategic player can choose to mimic one of multiple justified demands (Section \ref{sec:multiple}). There is still a unique equilibrium outcome (Theorem \ref{thm:1multiple}). Moreover, the presence of ultimatum opportunities affects players' bargaining power in a remarkably simple way. As the initial probability of being justified converges to zero, and as the set of justified demands gets larger and finer, the players' equilibrium payoffs converge to a unique vector and the outcome is efficient (Proposition \ref{prop:1limit}). Player 1's equilibrium payoff is the {\AG} payoff if the ultimatum opportunity arrival rate is smaller than his discount rate, and is equal to what his {\AG}  payoff would be if his discount rate were replaced by the ultimatum opportunity arrival rate if the latter is larger than his discount rate. In the former case, players tend to compromise; in the latter case, player 2 chooses the greediest demand to leverage ultimatums.

We then consider the setting in which both players have the opportunity to challenge and each player has a single justified demand (Section \ref{sec:twosided}). If at least one player's exogenous ultimatum opportunity arrival rate is lower than the {\AG} equilibrium concession rate, there exists a unique equilibrium outcome that is similar to the one in the setting with one-sided ultimatum (Theorem \ref{T:existenceandproperties2}). Otherwise, perhaps interestingly, there may be multiple\threeemdashes{}possibly a continuum of\threeemdashes{}equilibria in which a player's reputation (i) may decline over time, approaching zero but never reaching it, or (ii) may reach a positive level and stay there while they challenge and concede at rates that result in steady reputations.  Inefficient delays persist in equilibrium even in the limit case of rationality due to an overabundant availability of access to the court. This result suggests that more convenient access to the court may be counterproductive and socially inefficiency for negotiation.

The rest of the paper proceeds as follows. Section \ref{sec:model} describes the basic model with one-sided ultimatum and single justified demands. Section \ref{sec:equilibrium} characterizes its equilibrium. Section \ref{sec:implications} discusses the model's implications, including discontinuous and piecewise monotonic rates of challenging, benefits and costs of ultimatum opportunities for players, payoffs in the limit case of rationality, and comparative statics. Section \ref{sec:multiple} discusses the case with multiple justifiable and the limit payoffs in the case of rationality and rich type spaces. Section \ref{sec:twosided} discusses two-sided ultimatum. Section \ref{sec:literature} discusses relation to literature, and Section \ref{sec:conclusion} concludes. Appendix \ref{sec:proofs} collects omitted proofs.

\section{Model}\label{sec:model}
Players 1 (``he'') and 2 (``she'') decide on how to split a unit surplus. Each player is either (i) \emph{justified} in demanding a fixed share of the pie, or (ii) \emph{unjustified} in demanding any fixed share.\footnote{We use ``justified type'' and ``commitment type'' interchangeably and ``unjustified type'' and ``strategic type'' interchangeably.} 
We start by assuming that each player can be of a single justified type: With probability $z_1$ player 1 is justified in demanding $a_1\in (0,1)$, and with probability $z_2$ player 2 is justified in demanding $a_2>1-a_1$. Let $D:= a_i-(1-a_j)$ denote the amount of disagreement between the two players.

Time is continuous and the horizon is infinite. At each instant, each player can either concede to their opponent or not concede. We assume that each justified player never concedes. When an unjustified player $i$ concedes, $i$ gets a payoff of $1-a_j$ and player $j$ gets a payoff of $a_j$. In addition, we start by assuming a one-sided challenge model: Player 1 has an opportunity to challenge player 2 with an ultimatum. The game ends upon a concession, and moves to the challenge phase if a player challenges.

\paragraph{Challenge phase.}
A justified player 1 challenges according to a Poisson process with an exogenous arrival rate $\gamma_1\geq 0$. An unjustified player 1 can challenge at any time, so he can time his challenge strategically.
It costs $c_1D$ for player 1 to challenge.

\paragraph{Response to a challenge.}
Player 2 can respond to a challenge either by yielding or by seeing it. A justified player 2 always sees a challenge, and an unjustified player 2 chooses between the two actions. If player 2 yields, she gets $1-a_1$, and player 1 gets a payoff of $a_1$. It costs $k_2D$ for player 2 to see a challenge, and in this case the division of the pie is determined by a third party, such as a court, who observes the players' types.

\paragraph{Payoff determination in court.}
If an unjustified player meets a justified player, the justified player always wins, so an unjustified player $i$ receives $1-a_{j}$ against a justified player $j$. If two unjustified players meet, the challenging player 1 wins and gets $a_1$ with probability $w_1$, or loses and gets $1-a_2$ with probability $1-w_1$. Therefore, his expected share is $1-a_2+w_1D$ and the defending player 2's expected share is $1-a_1+(1-w_1)D$. The players' payoffs are linear in the share of the surplus they receive, so we could equivalently interpret that the third party decides on a deterministic division that gives each player their respective expected share. We do not specify the court's decision if both players are justified, since this does not play any role in the strategic decisions of the unjustified players.

\paragraph{Assumptions.}
We assume that if player 2 always sees the challenge, player 1 prefers conceding to challenging: $w_1<c_1<1$; and that player 2 prefers seeing a challenge from an unjustified player 1 to yielding to it: $0<k_2<1-w_1$. If $w_1=0$\threeemdashes{}i.e., the court never rules in favor of an unjustified plaintiff\threeemdashes{}we are simply assuming that $c_1$ and $k_2$ are strictly between $0$ and $1$.

In summary, a bargaining game $B=\left(a_1,a_2,z_1,z_2,r_1,r_2,\gamma_1,c_1,k_2,w_1\right)$ with ultimatum opportunities for one player and single demand types for both players is described by players' justified demands $a_1$ and $a_2$, prior probabilities $z_1$ and $z_2$ of being justified, discount rates $r_1$ and $r_2$, challenge opportunity arrival rate $\gamma_1$ for a justified player 1, challenge cost $c_1$ and seeing cost $k_2$ as proportions of the conflicting difference, and an unjustified player 1's winning probability $w_1$ against an unjustified opponent.

\begin{remark}[Discussions of modeling choices]

One interpretation of the justified player is that they can justify their demand with verifiable evidence, which is obtainable with some frictions, such as time delay, but an unjustified player cannot provide any evidence, but nevertheless bluffs with challenges. We assume that the unjustified player can challenge at any time while the justified player challenges only when the opportunity to challenge arrives. If the unjustified players could choose to challenge only when the challenge opportunity arrives (i.e., an unjustified player also takes time to find a lawyer), like the justified player, all our qualitative results except Proposition \ref{prop:gamma1} would hold and quantitatively only minor changes would be needed. In addition, we assume that challenge opportunities arrive according to a Poisson process, implying a constant arrival rate. This assumption eases some of the calculation and exposition of our results. However, our analyses do not rely on the exact arrival process to be Poisson. Moreover, assuming a stationary arrival process helps tease out the sources of nonmonotonicity and discontinuity of dispute resolution in Proposition \ref{prop:hazardrates}. Finally, our results continue to hold if player 1 pays the court cost only when player 2 sees the challenge.

\end{remark}

The formal description of players' strategies and payoffs is as follows. Since only unjustified players can choose their strategies, we drop the qualifier ``unjustified'' or ``strategic'' whenever no confusion can arise. An unjustified player 1's strategy is described by $\Sigma_1=(F_1,G_1)$, where $F_1$ and $G_1$, the probabilities of conceding and challenging by time (including) $t$, respectively, are right-continuous and increasing functions with $F_1(t)+G_1(t)\leq 1$ for every $t\geq 0$. A strategic player 2's strategy is described by $\Sigma_2=(F_2,q_2)$, where $F_2$, the probability of conceding by time $t$, is a right-continuous and increasing function with $F_2(t)\leq 1$ for every $t\geq 0$, and $q_2(t)\in [0,1]$, her probability of yielding to a challenge at time $t$, is a measurable function. Each strategy profile induces a distribution over action profiles, which we refer to as \emph{equilibrium play}.

A strategic player 1's (time-zero) expected utility from conceding at time $t$ is\footnote{We assume an equal split when two players concede at the same time. It is inconsequential for our results, because simultaneous concession occurs with probability 0 in equilibrium.}
\begin{eqnarray}\label{eq:u1t}
U_1(t,\Sigma_2)&=&(1-z_2)\int_{0}^t a_1 e^{-r_1s}dF_2(s)+\Big[1-(1-z_2)F_2(t)\Big]e^{-r_1t}(1-a_2)
\nonumber\\ &&\quad\quad
+(1-z_2)\Big[F_2(t)-F_2(t^-)\Big]\frac{a_1+1-a_2}{2},
\end{eqnarray}
where $F_2(t^-):=\lim_{s\uparrow t}F_2(s)$. His expected utility from challenging at time $t$ is\footnote{We assume that whenever concession and challenge occur simultaneously, the outcome is determined by the concession. This is an innocuous assumption, because simultaneous concession and challenge occur with probability 0 in equilibrium.}
\begin{eqnarray*}
V_1(t,\Sigma_2)&=&(1-z_2)\int_{0}^t a_1 e^{-r_1s}dF_2(s) + \Big[1-(1-z_2)F_2(t)\Big] e^{-r_1t}(1-a_2-c_1D) +\\
&&\quad (1-z_2)[1-F_2(t)]e^{-r_1t}[(1-q_2(t))w_1+q_2(t)]D.
\end{eqnarray*}
His expected utility from strategy $\Sigma_1$ is
$$
u_1(\Sigma_1,\Sigma_2)=\int_{0}^\infty U_1(s,\Sigma_2) dF_1(s) + \int_{0}^\infty V_1(s,\Sigma_2) dG_1(s).
$$

A strategic player 2's expected utility from conceding at time $t$ and yielding according to $q_2(\cdot)$ when facing a challenge is
\begin{eqnarray}\label{eq:u2t}
U_2(t,q_2(\cdot),\Sigma_1)&=&(1-z_1)\int_{0}^t a_2 e^{-r_2s}dF_1(s)
+z_1\int_{0}^t \big[1-a_1-(1-q_2(s))k_2D\big] e^{-r_2s}\gamma_1e^{-\gamma_1s}ds\nonumber\\
&&\quad+(1-z_1)\int_{0}^t \Big\{1-a_1+\big[1-q_2(s)\big]\big[1-w_1-k_2\big]D\Big\}e^{-r_2s}dG_1(s)\nonumber\\
&&\quad+e^{-r_2t}(1-a_1)\Big[1-(1-z_1)F_1(t)-(1-z_1)G_1(t^-)-z_1\left(1-e^{-\gamma_1t}\right)\Big]\nonumber\\
&&\quad+(1-z_1)\Big[F_1(t)-F_1(t^-)\Big]\frac{a_2+1-a_1}{2},
\end{eqnarray}
where $F_1(t^-):=\lim_{s\uparrow t}F_1(s)$. Her expected utility from strategy $\Sigma_2$ is 
$$u_2(\Sigma_2,\Sigma_1)=\int_0^\infty U_2(s,q_2,\Sigma_1)dF_2(s).$$

We study the Bayesian Nash equilibria of this game. Because the game is dynamic, it is natural to define public beliefs about players' types, i.e., \emph{reputation processes}, throughout the game. We define the reputation process $\mu_i(t)$ in the natural way, as the posterior belief that player $i$ is justified conditional on the game not ending by time $t$. Bayes' rule gives us this process explicitly as 
\[\mu_1(t)=\frac{z_1\left[1-\int_0^t \gamma_1e^{-\gamma_1 s}ds\right]}{z_1\left[1-\int_0^t \gamma_1e^{-\gamma_1 s}ds\right]+(1-z_1)\Big[1-F_1(t^-)-G_1(t^-)\Big]},\]
and
\[\mu_2(t)=\frac{z_2}{z_2+(1-z_2)\left[1-F_2(t^-)\right]}.\]
Finally, let $\nu_1(t)$ be player 2's posterior belief that player 1 is justified conditional on player 1 challenging at time $t$. Namely, $\nu_1(t)=0$ at any $t\geq 0$ where $G_1$ has an atom, and at any $t\geq 0$ where $G_1$ is differentiable,
\begin{equation}
\nu_1(t)=\frac{\mu_1\gamma_1}{\mu_1\gamma_1+[1-\mu_1(t)]\chi_1(t)},
    \label{eq:updating}
\end{equation}
where 
\[
\chi_1(t)=\frac{G_1'(t)}{1-F_1(t^-)-G_1(t^-)}
\]
is an unjustified player 1's hazard rate of challenging.\footnote{The function $G_1$ is differentiable almost everywhere, because it is right-continuous and monotone. Moreover, the posterior beliefs are well defined at the jump points of $G_1$, and hence, they are well defined almost everywhere in both the $G_1$ measure and Lebesgue measure.}

\begin{remark}[Connection to continuous-discrete-time model]
We model the negotiation process directly as a concession game in the style of a war of attrition with the additional ultimatums. We could alternatively model the negotiations as a continuous-discrete-time model in  which  a  player  can  change  his  demand at  any  positive integer time, but can concede to an outstanding demand (or challenge in our case) at any time $t\in[0,\infty)$. This formulation was introduced by \citet{AbreuPearce2007Ecta} in a repeated games with contracts setting and adopted by \citet{AbreuPearceStacchetti2015TE} in a bargaining context. In that formulation, without ultimatums, whenever a player makes a demand different from a commitment (justifiable) type she reveals her rationality, and there is a unique equilibrium continuation payoff vector, which coincides with the payoff vector from concession. With ultimatums, however, when player 2 reveals rationality, there are multiple equilibria with different continuation payoffs. For example, there is an equilibrium in which player 2 chooses a fixed demand, players concede to each other at constant hazard rates, player 1 challenges at a constant rate, and player 1's reputation stays constant. However, when player 1 reveals rationality, there is a unique equilibrium continuation payoff vector, which coincides with the payoff vector from concession. In particular, all of the equilibria we identify in our model have an analogous equilibrium in the continuous-discrete-time bargaining model that yields identical behavior.
\end{remark}

\section{Equilibrium characterization}\label{sec:equilibrium}
In this section, we solve and characterize the equilibrium strategies and reputations. The bargaining game entails a unique equilibrium play, which satisfies the following four properties.

\begin{theorem}\label{T:existenceandproperties}
Consider a bargaining game $B=(a_1,a_2,z_1,z_2,r_1,r_2,\gamma_1,c_1,k_2,w_1)$ with one-sided ultimatum and single demand types. There exist finite times $T$ and $T_1\in [0,T)$ such that every equilibrium strategy profile $(\widehat F_1, \widehat G_1, \widehat F_2, \widehat q_2)$ satisfies the following properties.
\begin{enumerate}
    \item $\widehat{F}_1$ and $\widehat{F}_2$ are strictly increasing in $(0,T)$ and constant for $t\geq T$;\label{property1}
    \item $\widehat{F}_1$ and $\widehat{F}_2$ are atomless in $(0,T]$ and at most one of the two has an atom at $t=0$;\label{property2}
    \item\label{property3} \begin{enumerate}
        \item $\widehat{G}_1$ is atomless in $[0,T]$, strictly increasing in $[0,T_1]$, and constant for $t\geq T_1$;\label{property3a}
        \item For almost every $t\in [0,T]$, $\widehat q_2(t)\in (0,1)$ if $t\in [0,T_1]$ and $\widehat q_2(t)=1$ if $t\in (T_1,T]$;\label{property3b}
    \end{enumerate} 
    \item $\widehat{F}_1(T)+\widehat{G}_1(T_1)=1$ and $\widehat{F}_2(T)=1$.\label{property4}
\end{enumerate}
Moreover, $\widehat{F}_1$, $\widehat{F}_2$, and $\widehat{G}_1$ are unique, and $\widehat q_2$ is unique almost everywhere for $t\leq T$.
\end{theorem}

Properties \ref{property1} and \ref{property2}, and the part of Property \ref{property4} regarding concessions coincide with the three properties in {\AG}. The first property states that there is a finite time $T>0$ such that players yield to each other with a strictly positive probability in every subinterval of $(0,T]$, and never yield after $T$. The second property states that distributions of concession are atomless except at $t=0$, and there can be an atom in at most one of these distributions. The fourth property modifies {\AG}, and states that unjustified players have either yielded, challenged, or been yielded to before time $T$.

Theorem \ref{T:existenceandproperties} extends {\AG}'s equilibrium characterization when there are ultimatum opportunities. There are difficulties, however, due to players' larger strategy spaces: In addition to the timing of concession, player 1 chooses the timing of challenge and player 2 chooses how to respond to a potential challenge at each instant.  A priori, players' incentives to concede may change\threeemdashes{}for better or for worse\threeemdashes{}due to the opportunity/anticipation of challenge opportunities at each instant. We first show that in every equilibrium, player 2 does not benefit from challenges, i.e., at each instant she weakly prefers conceding to seeing a challenge. Second, we show that $\widehat G_1$ is atomless. These findings allow us to show that players' concession distributions are strictly increasing and atomless in an interval $(0,T)$. This implies that player 2's reputation is increasing, which allows us to show the novelty of our characterization, namely Property \ref{property3}.

Property \ref{property3a} asserts that player 1 challenges his opponent with an atomless distribution until some time $T_1<T$, and never challenges afterward. Property \ref{property3b} asserts that player 2 responds to a challenge by both seeing the challenge and conceding to it with positive probabilities until $T_1$, and concedes to it afterward. Because this is a new property, let us provide an intuition for why this property must hold. Property \ref{property1} implies that at any time $t\in (0,T)$, player $i$'s continuation payoff at time $t$ is equal to $1-a_j$. If $\widehat{G}_1$ is constant in some interval, then after observing a challenge in that time interval, player 2's posterior belief that player 1 is justified is one, and player 2 optimally yields. However, if player 2's reputation is smaller than $\mu_2^*$, then challenging gives player 1 a payoff that strictly exceeds $1-a_2$, which yields a contradiction. Similarly, if $\widehat{G}_1$ had an atom at some $t$, then after observing a challenge at $t$, player 1's reputation would be 0, and player 2 would optimally see the challenge. However, then player 1 would receive a payoff strictly lower than $1-a_2$, leading again to a contradiction. Finally, as we will argue in the next section, player 2's reputation increases over time, and at some time $T_1<T$ reaches $\mu_2^*$. After this time, player 1 never challenges. Finally, for $t<T_1$, player 1 is indifferent between conceding and challenging, and player 2's reputation is smaller than $\mu_2^*$. Therefore, $\widehat q_2(t)\in(0,1)$ for $t<T_1$.

We now use the four properties to derive the closed-form solutions of equilibrium strategies $\widehat F$ and $\widehat G$. In the next subsection, we derive the equilibrium concession rates at $t>0$, player 1's challenge rate, and player 2's challenge response. We then derive a reputation coevolution diagram based on these rates, which allows us to compute the probabilities of concession at $t=0$.

\subsection{Concession rates, challenge rate, and challenge response}

\subsubsection{Player 2's concession rate}
Property \ref{property1} says that player 1 yields with a positive probability in every subinterval of $(0,T)$, so player 1's continuation payoff at every time $t$ is equal to $1-a_2$, and he is indifferent between yielding at any time in $(0,T)$. Hence, player 2 (the unjustified type) concedes at the rate $\kappa_2$ in the interval $(0,T)$ that sustains this indifference:
$$1-a_2=a_1 (1-\mu_2) \kappa_2 dt + e^{-r_1dt} (1-a_2)\Big[1-(1-\mu_2)\kappa_2dt\Big] \Longrightarrow \kappa_2 = \frac{r_1 (1-a_2)/D}{1-\mu_2}.$$
This implies, from player 1's perspective, that the hazard rate of player 2 yielding to player 1 is $\lambda_2=r_1(1-a_2)/D$, as in {\AG}. An immediate implication is that player 2's reputation conditional on negotiation continuing at time $t<T$, $\mu_2(t)$, is an increasing function.

\subsubsection{Player 1's challenge rate and player 2's response to challenge}
Property \ref{property2} implies that player 1 is indifferent between challenging and yielding at any time $t\in(0,T_1)$. Recall that at any such $t$, player 2's reputation is $\mu_2(t)$ and $q_2(t)$ is the probability that player 2 yields if a challenge comes at time $t$. Player 1's payoff from challenging at time $t$ is equal to $1-a_2+\big[1-\mu_2(t)\big]\big[q_2(t)+w_1(1-q_2(t))\big]D-c_1D.$ The indifference condition for player 1 implies that his payoff from challenging must be equal to $1-a_2$. Hence, we obtain that 
\begin{equation}
q_2(t)=\frac{1}{1-w_1}\left[\frac{c_1}{1-\mu_2(t)}-w_1\right].\label{eq:q2t}
\end{equation}
Hence, we obtain that at any time $t\leq T_1$, player 1's reputation conditional on challenging player 2 is equal to 
$$\nu_1^*:= 1-\frac{k_2}{1-w_1}.$$
This implies, by Bayes' rule, that player 1's challenge rate seen as a function of player 1's reputation is calculated from Equation \eqref{eq:updating}:

\begin{equation}
\chi_1(\mu_1):=\frac{1-\nu_1^*}{\nu_1^*}\frac{\mu_1}{1-\mu_1}\gamma_1.\label{eq:chi1}
\end{equation}
To summarize, Equation \eqref{eq:q2t} holds almost everywhere for $t\leq T$, because actions after time $T$ are off equilibrium path for a strategic player 2, and Equation \eqref{eq:chi1} holds almost everywhere for $t\leq T_1$ and $\chi_1(t)=0$ almost everywhere for $t \in (T_1,T]$.\footnote{
We assume $w_1<c_1<1$ and $0<k_2<1-w_1$ to ensure positive probabilities of using challenge opportunities and responding to challenges. If $k_2\geq 1-w_1$, a strategic player 2 always yields when challenged, as $\nu_1^*\leq 0$. Given player 2's strategy, a strategic player 1 challenges with probability $1$ when $\mu_2<\mu_2^*$, because he is strictly better off challenging than not challenging, and does not challenge when $\mu_2>\mu_2^*$ because he is strictly worse off challenging.}

\subsubsection{Player 1's concession rate}
Property \ref{property1} says that player 2 yields with a positive probability in every subinterval of $(0,T)$, as is the case for player 1. However, from player 2's perspective, in any time interval, player 1 may yield to or challenge player 2. As we have already argued, because player 2 sees the challenge with an interior probability, her continuation payoff when she is challenged is equal to $1-a_1$\threeemdashes{}i.e., her payoff from yielding to player 1. Hence, the indifference condition for player 2 in yielding across all times $t\in(0,T)$ implies the concession rate  $\kappa_1 = {[r_2 (1-a_1)]}/{[D(1-\mu_1)]},$ which results in the overall hazard rate of player 1 yielding to player 2 as $\lambda_1=r_2 (1-a_1)/D$, as in {\AG}. To summarize, each player $i$, $i=1,2$, concedes at the overall rate of
\begin{equation}\label{eq:lambda}
\lambda_i:=\frac{r_j(1-a_i)}{a_1+a_2-1}.
\end{equation}

\subsection{Reputation dynamics and reputation coevolution}\label{sec:reputation}

We now characterize the evolution of the players' reputations. To do so, we use the concession rates and the challenge rate of player 1 found in the previous section. We start with player 2's reputation building for $t\in (0,T]$. Player 1's reputation dynamics depend on both his concession rate and challenge rate. We start with the \emph{no-challenge phase}, $t\in(T_1,T]$, and then characterize the \emph{challenge phase}, $t\in(0,T_1]$.

Note that Property \ref{property4} implies that $\mu_i(T)=1$ for $i=1,2$. Using this property and the reputation dynamics we derive, we characterize the \emph{reputation coevolution curve}. This curve shows the locus of the reputation vectors at times $t>0$. The curve will determine the identity of the player who yields with a positive probability at time 0 and the magnitude of that atom. This will complete the characterization of the unique equilibrium.

\subsubsection{Player 2's reputation}

Player 2's reputation building is only affected by her overall concession rate $\lambda_2$. Following the Martingale property $\mu_i(t)=\mathbb E_t \mu_i(t+dt)$, we have
$$
\mu_2(t)=\lambda_2 dt\cdot 0 + (1-\lambda_2 dt) \cdot \mu_2(t+dt).
$$
Rearranging, dividing both sides by $dt$, and taking $dt\rightarrow 0$, we get
\begin{equation}\label{eq:mu2t}
\mu_2'(t)=\lambda_2 \mu_2(t).
\end{equation}

\subsubsection{Player 1's reputation in the no-challenge phase}

In this phase, player 1 concedes with an overall rate of $\lambda_1$, and the justified player 1 challenges with a rate of $\gamma_1$. Following the Martingale property, we have
$$
\mu_1(t)=\mu_1(t)\gamma_1dt\cdot 1+\lambda_1dt\cdot 0+\big[1-\mu_1(t)\gamma_1dt-\lambda_1dt\big] \cdot \mu_1(t+dt).
$$
Rearranging and 
taking $dt\rightarrow 0$, we get that player 1's reputation follows a Bernoulli differential equation
\begin{equation}\label{eq:mu1tN}
\mu_1'(t)=(\lambda_1-\gamma_1)\mu_1(t)+\gamma_1\mu_1^2(t),
\end{equation}
which can be rearranged and decomposed as
$$
\frac{\mu_1'(t)}{\mu_1(t)}=\lambda_1-\left[1-\mu_1(t)\right]\gamma_1.
$$
Note that the reputation strictly increases, i.e., $\mu_1'(t)>0$, if $$\mu_1(t)>1-\frac{\lambda_1}{\gamma_1}=:\phi_1^*.$$

\subsubsection{Player 1's reputation in the challenge phase}

Recall from Equation \eqref{eq:chi1} that in this phase, an unjustified player 1 challenges at a reputation-dependent rate. Hence, again, using the Martingale property of beliefs, we have
\begin{align*}
\mu_1(t)=&\mu_1(t)\gamma_1dt\cdot 1+
\big\{1-\mu_1(t)\gamma_1dt-[1-\mu_1(t)]\chi_1(t)dt-\lambda_1dt\big\}\mu_1(t+dt).\footnotemark
\end{align*}
\footnotetext{This equality is obtained as if player 1's type is revealed after the challenge. One can verify that it is equivalent to using the reputation $\nu_1^*$ after a challenge.}Rearranging the equation 
and taking $dt\rightarrow 0$, we get that player 1's reputation follows the following Bernoulli differential equation:
\begin{equation}\label{eq:mu1tC}
\mu_1'(t)=(\lambda_1-\gamma_1)\mu_1(t) +\frac{\gamma_1}{\nu_1^*}\mu_1^2(t),
\end{equation}
which can be rearranged and decomposed as
\begin{equation}\label{eq:reputation-challenge}
\frac{\mu_1'(t)}{\mu_1(t)}=\lambda_1-\big[1-\mu_1(t)\big]\gamma_1+\big[1-\mu_1(t)\big]\chi_1(t).
\end{equation}
Note that $\mu_1'(t)>0$ when
\begin{equation}\label{eq:cutoff-challenge}
\mu_1(t)>\left(1-\frac{\lambda_1}{\gamma_1}\right)\nu_1^*=\phi_1^*\nu_1^*.
\end{equation}

\subsubsection{Bad news and good news effects}\label{remark:goodnewsbadnews}

Two forces shape the evolution of player 1's reputation. First, ``no concession is good news'': With player 1 conceding at rate $\lambda_1$, player 1's reputation increases exponentially at rate $\lambda_1$. Observe that when $\gamma_1=0$, Equations \eqref{eq:mu1tN} and \eqref{eq:mu1tC} boil down to the exponential growth reputation dynamics in {\AG}. The second force, which is new, comes from the equilibrium challenges.

This second force can \emph{decelerate} or \emph{accelerate} reputation building. In the no-challenge phase, for example, ``no ultimatum is bad news'': With only the justified player 1 challenging at rate $\gamma_1$ and an unjustified player not challenging at all, player 1's reputation declines at rate $[1-\mu_1(t)]\gamma_1$. In the challenge phase, however, the unjustified player 1 also challenges at a positive rate. Hence, the ``bad news'' effect of no challenge is less severe in this phase compared to the no-challenge phase. This is captured by the third term in Equation \eqref{eq:reputation-challenge}. In fact, when $\chi_1(t)>\gamma_1$, player 1's reputation building ``accelerates'' with no challenge, and no challenge becomes ``good news.'' Given that $\chi_1(t)=\gamma_1\frac{\mu_1(t)}{1-\mu_1(t)}/\frac{\nu_1^*}{1-\nu_1^*}$, player 1's reputation builds faster when $\mu_1(t)>\nu_1^*$, while player 2's reputation is not too high, $\mu_2(t)<\mu_2^*$. This effect provides a benefit from ultimatum opportunities for an unjustified player 1 who has an intermediate range of reputations. This range may not exist in equilibrium. We characterize the range of initial reputations for ultimatum opportunities to be beneficial for an unjustified player 1 in Section \ref{sec:whobenefits}.

\subsubsection{Reputation coevolution diagram and initial concession}

 Both players' reputation dynamics in each phase given by Equations \eqref{eq:mu2t}, \eqref{eq:mu1tN} and \eqref{eq:mu1tC} follow the Bernoulli differential equation, which is one of the few special cases of ordinary differential equations with exact solutions summarized in Lemma \ref{lem:ODE} in Appendix \ref{sec:details}, and includes the exponential growth of {\AG} as the special case when ultimatum opportunities are absent. Hence, it is feasible to combine the reputation-building dynamics at different phases of the game to find the evolution of both players' reputations in equilibrium. To do so, we ``run'' the Bernoulli differential equations that describe players' reputation dynamics backward, starting from $T$.

Recall that the finiteness of $T$ in Property \ref{property4} of Theorem \ref{T:existenceandproperties} implies that $\mu_1(T)=\mu_2(T)=1$. Moreover, $\mu_2(T_1)=\mu_2^*$. Hence, $T-T_1$ can be found using player 2's reputation dynamics given by Equation \eqref{eq:mu2t}. Then we can use player 1's reputation dynamics in the no-challenge phase, Equation \eqref{eq:mu1tN}, to find $\mu_1(T_1)$. Then we let $T_1^*$ be the time it takes for player 1 to build a reputation from $z_1$ to $\mu_1(T_1)$ using the dynamics in Equation \eqref{eq:mu1tC}, and $T_2^*$ the time it takes for player 2 to build a reputation from $z_2$ to $\mu_2^*$ using the dynamics in Equation \eqref{eq:mu2t}. Finally, we let $T_2:=\min\{T_1^*,T_2^*\}$, and conclude that if $T_i^*>T_2$, then player $i$ concedes at time 0 with a strictly positive probability.

Alternatively, we can trace out a parametric \emph{reputation coevolution curve} $(\mu_1(t), \mu_2(t))$ in the belief plane, which represents the locus of players' reputations for any initial reputations at any time $t>0$. Because both reputations are characterized analytically, we can represent the graph of the coevolution curve as $\widetilde \mu_1(\mu_2)$ for $\mu_2\in (0,1]$, or equivalently, its inverse $\widetilde \mu_2(\mu_1)$ for $\mu_1\in \big(\max\big\{0,\phi_1^*\nu_1^*\big\},1\big]$.
The coevolution curve is characterized by
$$
\widetilde\mu_{1}(\mu_{2})=\begin{cases}
\frac{\lambda_{1}-\gamma_{1}}{\lambda_{1}(\mu_{2})^{\frac{\gamma_{1}-\lambda_{1}}{\lambda_{2}}}-\gamma_{1}} & \text{if }\mu_{2}^{*}<\mu_{2}\leq 1,\\
\frac{\lambda_{1}-\gamma_{1}}{\lambda_{1}(\mu_{2})^{\frac{\gamma_{1}-\lambda_{1}}{\lambda_{2}}}+\left(\frac{\gamma_{1}}{\nu_{1}^*}-\gamma_{1}\right)\left(\frac{\mu_{2}}{\mu_{2}^{*}}\right)^{\frac{\gamma_{1}-\lambda_{1}}{\lambda_{2}}}-\frac{\gamma_{1}}{\nu_{1}^*}} & \text{if }0<\mu_{2}\leq \mu_{2}^{*},
\end{cases}
$$
when $\gamma_1\neq\lambda_1$. When $\gamma_1=\lambda_1$, this curve is obtained directly from $\mu_1(t)$ and $\mu_2(t)$ or by applying L'Hospital's rule to the above formula, and is explicitly given in Appendix \ref{sec:explicitresuls}. We can obtain the reputation $\mu_1^N=\widetilde \mu_1(\mu_2^*)$ of player 1 when player 2's reputation is $\mu_2^*$.

\begin{figure}[t!]
{
\centering
\begin{subfigure}{0.48\textwidth}
\centering
\includegraphics{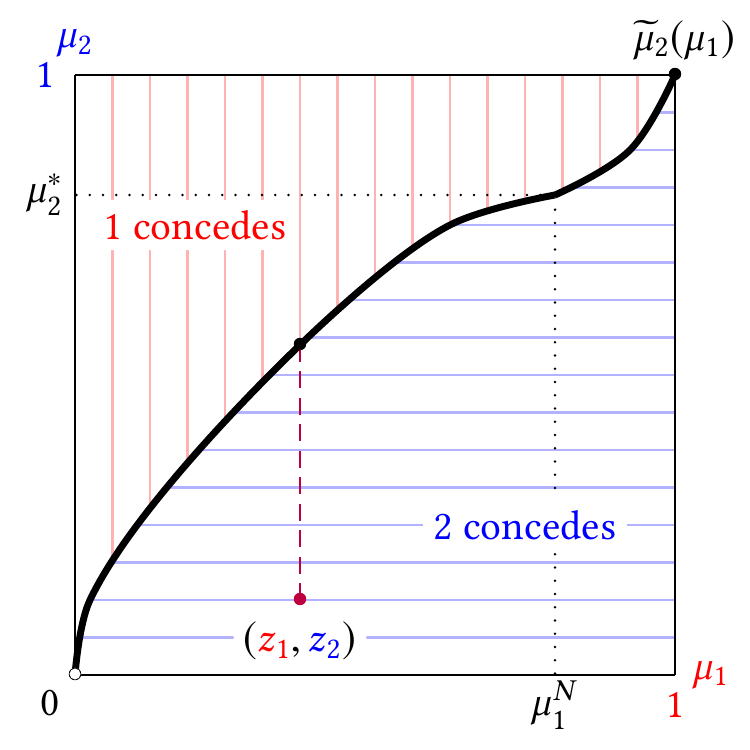}
\caption{\label{subfig:coevolution-11}$\gamma_1\leq \lambda_1$.}
\end{subfigure}
\begin{subfigure}{0.48\textwidth}
\centering
\includegraphics{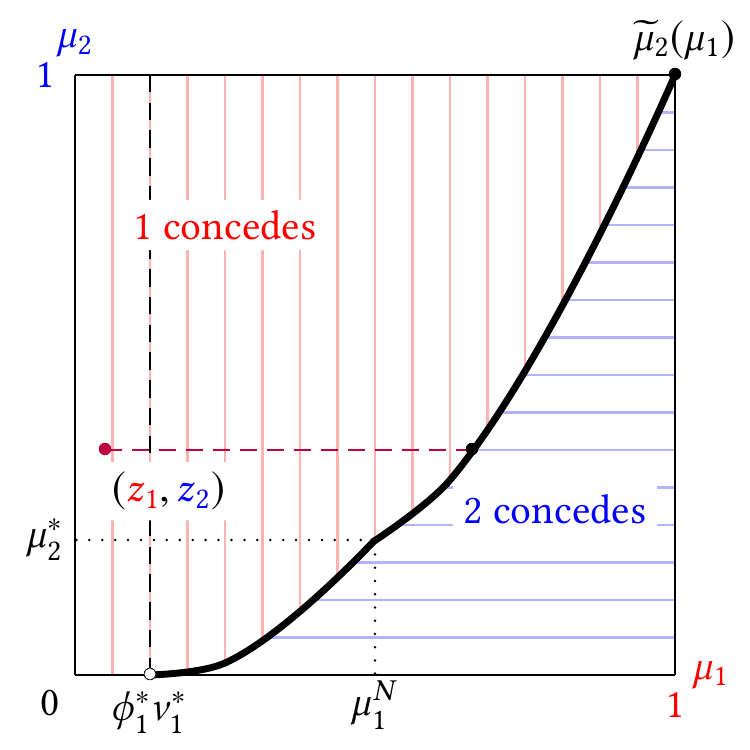}
\caption{\label{subfig:coevolution-12}$\gamma_1>\lambda_1$.}
\end{subfigure}
\caption{\label{fig:coevolution-1}Reputation coevolution and initial concession in games with one-sided ultimatum.}
}
{\footnotesize The solid line in each panel depicts the reputation coevolution curve $\widetilde\mu_2(\mu_1)$. Player 1 concedes with a positive probability at time $0$ when $(z_1,z_2)$ is strictly to the left of the curve, player 2 concedes with a positive probability at time $0$ when $(z_1,z_2)$ is strictly to the right of the curve, and neither player concedes with a positive probability at time $0$ when $(z_1,z_2)$ is on the curve. The probability of initial concession ensures that the posterior reputation vector after initial concession lies on the curve. The reputations coevolve to $(1,1)$ according to the curve. When player 2's reputation reaches $\mu_2^*$, player 1 stops challenging, and player 1's reputation $\mu_1^N$ at the time is derived from the reputation coevolution curve.}
\end{figure}

Figure \ref{fig:coevolution-1} provides examples of the reputation coevolution curve. When $\gamma_1\leq \lambda_1$, the curve tends toward $(0,0)$ (Figure \ref{subfig:coevolution-11}), and when $\gamma_1>\lambda_1$, since player 1's reputation is decreasing for reputation lower than $\phi_1^*\nu_1^*$ in the challenge phase (recall Equation \eqref{eq:cutoff-challenge}), the curve tends toward $(\phi_1^*\nu_1^*,0)$ (Figure \ref{subfig:coevolution-12}). When $(z_1,z_2)$ is on the coevolution curve, their reputations situate on the equilibrium path to $(1, 1)$, so neither player concedes at time 0 with a strictly positive probability. When $(z_1,z_2)$ is to the left of the curve, that is, $\widetilde\mu_2(z_1)<z_2$, or equivalently, $\widetilde\mu_1(z_2)>z_1$, player 1 will be the player who concedes with a positive probability at time 0. He must  concede with a probability $Q_1$ such that the pair of his posterior reputation and player 2's initial reputation $z_2$ exactly falls on the curve:
\begin{equation}\label{eq:Q}
\frac{z_1}{z_1+(1-z_1)(1-Q_1)}=\widetilde\mu_1(z_2)\Longrightarrow Q_1=1-\frac{z_1}{1-z_1}\bigg/\frac{\widetilde\mu_1(z_2)}{1-\widetilde\mu_1(z_2)}.
\end{equation}
When $(z_1,z_2)$ is to the right of the reputation coevolution curve, player 2 will be the one who concedes with a positive probability at time 0, which raises her reputation if she does not concede at time 0 to lie on the coevolution curve.

This completes our equilibrium characterization. We summarize the resulting equilibrium strategies and beliefs explicitly in Theorem \ref{thm:1single} in Appendix \ref{sec:details}.

\section{Implications}\label{sec:implications}
Subsequently, we discuss (i) the discontinuity and piecewise monotonicity of the hazard rates of ultimatum usage and conflict resolution, (ii) the welfare implications of the introduction of ultimatum opportunities, (iii) comparative statics, and (iv) players' strategies and payoffs in the limit case of rationality.

\subsection{Discontinuous and piecewise monotonic rates of ultimatum usage and dispute resolution}
While distributions of challenging and dispute resolution depend on parameters such as prior reputations and ultimatum opportunity arrival rate, the qualitative features of equilibrium hazard rates do not depend on the fine details of the model. For a strategic player 1, the hazard rate of challenging is
$\widehat \chi_1(t)=\frac{1-\nu_1^*}{\nu_1^*}\frac{\widehat\mu_1(t)}{1-\widehat\mu_1(t)}\gamma_1$ if $t<T_1$ and zero otherwise; and the rate of conceding is $\widehat\kappa_1(t)=\lambda_1/[1-\widehat\mu_1(t)]$ if $t<T$.\footnote{These rates are unique almost everywhere with respect to the $F_2$ measure and Lebesgue measure.} Figure \ref{fig:EquilibriumRates} illustrates a strategic player 1's equilibrium hazard rates of ultimatum and concession. Building on these rates, we can derive the overall hazard rates\threeemdashes{}that is, the aggregate rates by justified and unjustified players\threeemdashes{}of sending an ultimatum and of terminating the game. Proposition \ref{prop:hazardrates} summarizes and Figure \ref{fig:EmpiricalRates} illustrates these rates.

\begin{figure}[t!]
\begin{subfigure}[t]{.48\textwidth}
\centering
\includegraphics[width=\textwidth]{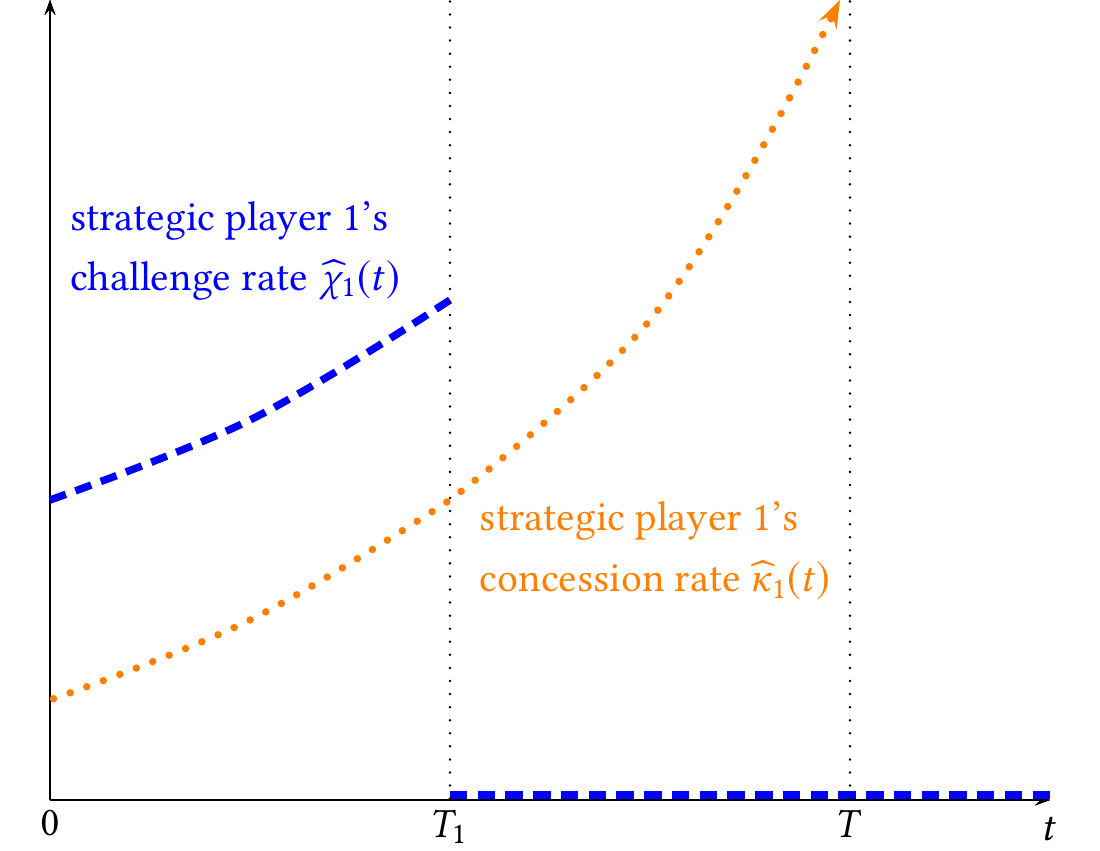}
\caption{A strategic player 1's equilibrium hazard rates of ultimatum (dashed) and concession (dotted).\label{fig:EquilibriumRates}}
\end{subfigure}
~
\begin{subfigure}[t]{.48\textwidth}
\centering
\includegraphics[width=\textwidth]{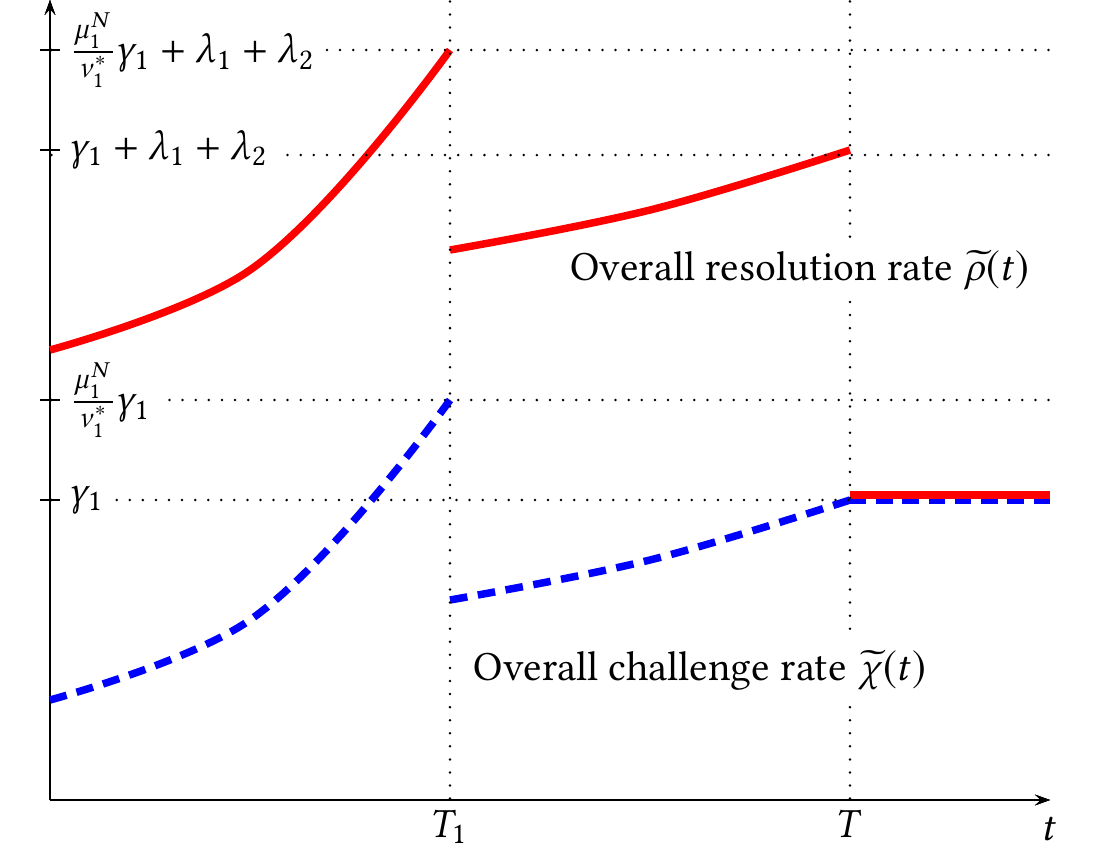}
\caption{Overall hazard rates of ultimatum (dashed) and resolution (solid).\label{fig:EmpiricalRates}}
\end{subfigure}
\caption{Hazard rates of ultimatum usage and dispute resolution}
{\footnotesize
As illustrated in Figure \ref{fig:EquilibriumRates}, a strategic player 1's equilibrium hazard rate of ultimatum usage increases between time $0$ and time $T_1$ and drops to zero afterward, and a strategic player 1's equilibrium concession rate increases between time $0$ and time $T$. As illustrated in Figure \ref{fig:EmpiricalRates}, the overall hazard rate of ultimatum usage increases between time $0$ and time $T_1$, drops from $\frac{\mu_1^N}{\nu_1^*}\gamma_1$\threeemdashes{}which might be above or below $\gamma_1$\threeemdashes{} to a rate below $\gamma_1$, and increases to $\gamma_1$ between time $T_1$ and time $T$. The overall hazard rate of dispute resolution adds the concession rate $\lambda_1+\lambda_2$ to the challenge rate before time $T$, and hence exhibits discontinuities at both  times $T_1$ and $T$.
}
\end{figure}

\begin{proposition}\label{prop:hazardrates}
Consider a bargaining game $B=(a_1,a_2,z_1,z_2,r_1,r_2,\gamma_1,c_1,k_2,w_1)$ with one-sided ultimatum and single demand types. The overall hazard rate of concession for player $i=1,2$ stays constant at $\lambda_i$ from time $0$ to time $T$. The overall hazard rates of challenging and resolution are
$$
\widetilde \chi(t)=
\begin{cases}
\frac{\widehat \mu_1(t)}{\nu_1^*} \gamma_1 & \text{if }t<T_1,\\
\widehat \mu_1(t)\gamma_1 & \text{if } T_1< t< T, \\
\gamma_1 & \text{if } t\ge T,
\end{cases}
\text{ and }
\widetilde \rho(t)=
\begin{cases}
\frac{\widehat \mu_1(t)}{\nu_1^*} \gamma_1+\lambda_1+\lambda_2 & \text{if } t<T_1, \\
\widehat \mu_1(t)\gamma_1+\lambda_1+\lambda_2 & \text{if } T_1< t<T, \\
\gamma_1 & \text{if } t\ge T,
\end{cases}
$$
respectively.
\end{proposition}

A testable prediction of the model is that the hazard rate of resolution in negotiations, which we can observe in many settings, experiences (local) peaks and subsequent discontinuities in three instances: (i) the onset of negotiation, (ii) the moment when a strategic player stops challenging, and (iii) the moment players stop conceding. The first peak arises when the agreement is reached at the onset of the negotiation, the second peak arises when player 2's reputation approaches the level beyond which player 1 has no incentive to challenge, and the last peak arises when both players' reputations approach 1, beyond which neither player has an incentive to continue the negotiation. In Appendix \ref{sec:data}, we present evidence suggesting that there is also an outburst in agreement after the beginning of the negotiation and before the deadline in MLB and the NHL salary arbitration cases, in addition to agreements at the onset of the game (predicted by \citet{AbreuGul2000Ecta} and \citet{Fanning2016Ecta}) and before the deadline (predicted by \citet{Fanning2016Ecta}; \citet{SimsekYildiz2016}; and \citet{VassermanYildiz2019RAND}).\footnote{We do not explicitly add a (stochastic) deadline to the model, but if we do, the discontinuity in the hazard rates of challenge and resolution in the middle of the negotiation remains, and there will be a mass of deals near the deadline.}

\subsection{Who benefits from ultimatum opportunities?}\label{sec:whobenefits}

Our discussion of the reputation dynamics in Section \ref{remark:goodnewsbadnews} decomposes the effects of the introduction of the ultimatum opportunity on reputation building, and shows that reputation building may be faster when a strategic player who challenges at a rate higher than the justified player's rate $\gamma_1$ (when the ``no ultimatum is good news'' effect dominates the ``no ultimatum is bad news'' effect). If we restrict that a strategic player cannot challenge at a rate higher than $\gamma_1$ (as would be the case if both types received a challenge opportunity at Poisson rate $\gamma_1$), then player 1 would never benefit from having the challenge opportunity. Hence, a necessary condition for player 1 to benefit from having an ultimatum is $\mu_1^N>\nu_1^*$. If $\mu_1^N\leq \nu_1^*$, unjustified player 1 never challenges at a higher rate than a justified player (easily seen from the equilibrium challenge rate $\widehat \chi_1(t)=\frac{1-\nu_1^*}{\nu_1^*}\frac{\widehat\mu_1(t)}{1-\widehat\mu_1(t)}\gamma_1$ if $\mu_1(t)<\mu_1^N$ in Proposition \ref{prop:hazardrates}). However, this condition is not sufficient for player 1 to benefit from the introduction of the ultimatum opportunity. The sufficient condition for player 1 to benefit from having the challenge opportunity is that it takes a longer time to build a reputation in the current setting than in {\AG}. Figure \ref{fig:WhoBenefits} illustrates who benefits from the challenge opportunity in the belief plane when $\mu_1^N\leq \nu_1^*$. There is an intermediate range of initial reputations of player 1 in which he benefits from the introduction of the ultimatum opportunity; we can show that this is always a connected interval bounded away from 0 and 1. We state the exact condition in the proposition below.

\begin{figure}[t!]
\begin{subfigure}{.48\textwidth}
\centering
\includegraphics{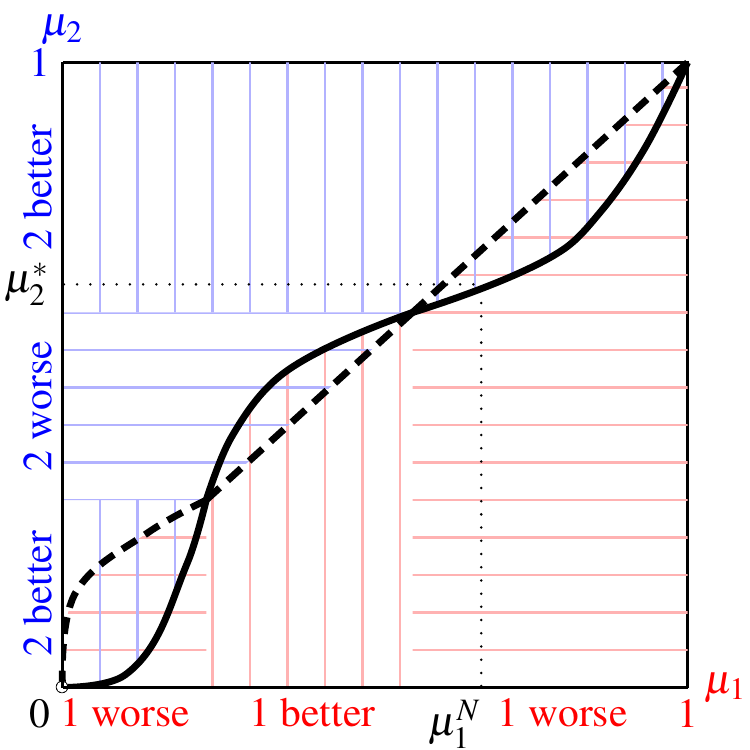}
\caption{Condition \eqref{condA} satisfied: 1 may be better off.}
\end{subfigure}
\begin{subfigure}{.48\textwidth}
\centering
\includegraphics{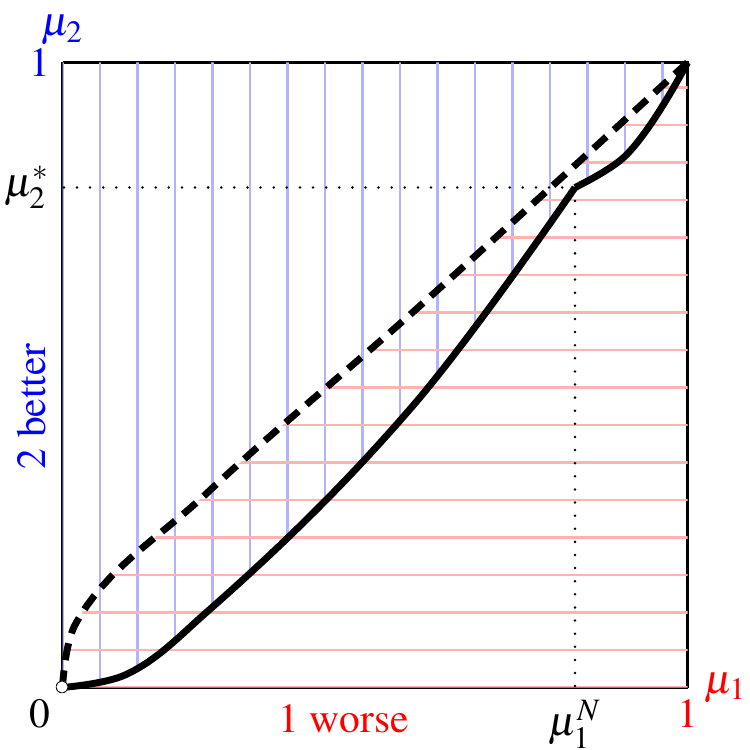}
\caption{Condition \eqref{condA} violated: 1 is never better off.}
\end{subfigure}
\caption{\label{fig:WhoBenefits}Who benefits from ultimatum opportunities?}
{\footnotesize The solid line represents the reputation coevolution curve with a challenge opportunity arriving at rate $\gamma_1>0$, and the dashed line is the reputation coevolution curve with $\gamma_1=0$, as in {\AG}. With the introduction of a challenge opportunity, player 1 (2) is strictly worse off if the pair of initial reputations is in the region filled with red (blue) horizontal lines, and is strictly better off if the pair of initial reputations is in the region filled with red (blue) vertical lines.}
\end{figure}

\begin{proposition}\label{prop:gamma1}
Consider a bargaining game $B=(a_1,a_2,z_1,z_2,r_1,r_2,\gamma_1,c_1,k_2,w_1)$ with one-sided ultimatum and single demand types. Let $\widetilde\mu_2(\mu_1|\gamma_1)$ denote the equilibrium reputation coevolution curve and $t_1(\mu_1|\gamma_1)$ the time it takes for player 1's reputation to increase from $\mu_1$ to $1$ in equilibrium when the challenge rate is $\gamma_1\ge 0$. If
\[
\mu_1^N>\nu_1^* \text{ and } t_1(\nu_1^*|\gamma_1)>t_1(\nu_1^*|0),\quad \tag{A}\label{condA}
\]
an unjustified player 1 strictly benefits from the introduction of the challenge opportunity if and only if $\underline\mu_1<z_1<\overline\mu_1$ and $\widetilde\mu_2(\underline\mu_1|\gamma_1)<\widetilde\mu_2(z_1|\gamma_1)<\widetilde\mu_2(\overline\mu_1|\gamma_1)$, where $\underline\mu_1$ and $\overline\mu_1$ are the two solutions to $\widetilde\mu_2(\mu_1|\gamma_1)=\widetilde\mu_2(\mu_1|0)$. If Condition \eqref{condA} is not satisfied, an unjustified player 1 cannot benefit from the introduction of the challenge opportunity.
\end{proposition}

Finally, as the ultimatum opportunities arrive very frequently (i.e., as $\gamma_1\to\infty$), $\phi_1^*v_1^*\to 1$, and for any given prior, players' payoffs converge to $(1-a_2,a_2)$. In other words, frequent ultimatum opportunities for player 1 cancel out player 1's reputation effects, resulting in a one-sided reputation payoff for player 2.

\subsection{Limit case of rationality}

We now look at the case in which the prior probability that each player is justified is small. This case captures situations in which being justified is a rare event and ultimatum is prominently used for strategic posturing. Generically (precisely, when $\lambda_1\neq \gamma_1+\lambda_2$), in the limit of complete rationality, players divide the surplus efficiently, with one player immediately conceding at time $0$ in equilibrium.
\begin{proposition}\label{prop:11limit}
Let $\{B^n\}_n$ be a sequence of games in which for each $n\in \mathbb{N}$,  $B^n=(a_1$, $a_2$, $z_1^n$, $z_2^n$, $r_1$, $r_2$, $\gamma_1$, $c_1$, $k_2$, $w_1)$ is a bargaining game with one-sided ultimatum and single demand types. If $\displaystyle \lim_{n\rightarrow \infty} z_1^n= \lim_{n\rightarrow \infty} z_2^n =0$, and $u_i^n$ is the equilibrium payoff for player $i$ in the game $B^n$, then 
$$\left(\lim_{n\rightarrow \infty} u_1^n, \lim_{n\rightarrow \infty} u_2^n\right)
=\begin{cases}
(1-a_2,a_2) & \text{if } \lambda_1<\gamma_1, \text{or} \\
 & \text{if } \gamma_1\leq \lambda_1<\gamma_1+\lambda_2 \text{ and } \displaystyle\lim_{n\rightarrow \infty} z_1^n/z_2^n \in (0,\infty),\\
(a_1,1-a_1) & \text{if } \lambda_1>\gamma_1+\lambda_2 \text{ and } \displaystyle\lim_{n\rightarrow \infty} z_1^n/z_2^n \in (0,\infty).\\
\end{cases}
$$
\end{proposition}

If $\lambda_1<\gamma_1$, the reputation coevolution curve approaches the x-axis at the belief $\phi_1^*\nu_1^*=(1-\lambda_1/\gamma_1)\nu_1^*$ (Figure \ref{subfig:coevolution-12}). Hence, for small $z_1$ and $z_2$, player 1 concedes at time 0 with a large probability such that conditional on no concession, player 1's reputation jumps above $\phi_1^*\nu_1^*$; we can verify this from Equation \eqref{eq:Q}.

If $\lambda_1\geq \gamma_1$, the reputation coevolution curve approaches the x-axis at the belief $0$. In this case, when the prior probability of being justified goes to zero on the same order for the two players, agreement is efficient, is on the terms of player 1 if $\lambda_1-\gamma_1>\lambda_2$, and is on the terms of player 2 if $\lambda_1-\gamma_1<\lambda_2$. To see this, note that the derivative of the reputation coevolution curve, $\widetilde\mu_2'(\mu_1)$, as $\mu_2$ goes to 0, tends to $\infty$ if $\lambda_1-\gamma_1>\lambda_2$ and tends to 0 if $\lambda_1-\gamma_1<\lambda_2$. Hence, as $z_1$ and $z_2$ go to $0$ on the same order, in the former case, player 2, and in the latter case, player 1 concedes at time 0 with a probability that approaches $1$.

Note that the limit payoffs are independent of the details of the arbitration, the cost $c_1$ of challenging, the cost $k_2$ of seeing the challenge, and the probability $w_1$ of winning the challenge. The discount rates $r_1$ and $r_2$ and the ultimatum opportunity arrival rate $\gamma_1$ do not affect efficiency, although they determine who is the winner (the player who is conceded to immediately) and the loser (the player who concedes immediately) in the game. In particular, the higher the ultimatum opportunity arrival rate $\gamma_1$, the more likely player 1 the loser. Hence, unlike the general case in which the ultimatum opportunity may benefit or harm a strategic player 1, in the limit case of rationality, the ultimatum opportunity is always detrimental to a strategic player 1.

The intuition for this ``independence from the details of arbitration'' finding can be gained from the reputation dynamics, given by Equations \eqref{eq:mu2t} and \eqref{eq:mu1tC}. When $z_1$ and $z_2$ are small, negotiation may last for a long time\threeemdashes{}i.e., $T$ is long. Moreover, reputation building for player 1, given by Equation \eqref{eq:mu1tC}, spends most of its time when $\mu_1(t)$ is small. Hence, player 1's reputation increases approximately exponentially, and at the rate $\lambda_1-\gamma_1$. In other words, it is as if the bad news effect of not challenging slows the rate of reputation building exactly by $\gamma_1$. In light of our discussion in Section \ref{remark:goodnewsbadnews}, this result shows that the good news effect of challenging disappears and the bad news effect persists for player 1 in the limit case of rationality.

Finally, the player who builds reputation with the higher rate is the ``winner,'' i.e., their opponent concedes at time $0$ with a positive probability. Because reputations grow exponentially (approximately for player 1), the initial concession probability converges to $1$ as $z_1$ and $z_2$ approach $0$ on the same order. This final part of our analysis is similar to the analysis of \citet{AbreuGul2000Ecta} and \citet{Kambe1999GEB}.

\subsection{Comparative statics}

\begin{proposition}\label{prop:compstat1}
Start with a bargaining game $B=(a_1,a_2,z_1,z_2,r_1,r_2,\gamma_1,c_1,k_2,w_1)$ with one-sided ultimatum and single demand types. When parameters change, equilibrium payoffs stay constant unless stated below. When $z_i$ increases or $r_i$ decreases, $u_i$ increases if $z_i\geq \widetilde\mu_i(z_j)$ and $u_j$ decreases if $z_i\leq \widetilde\mu_i(z_j)$. When either $c_1$ increases, $k_2$ increases, or $w_1$ decreases, $u_1$ decreases if $\widetilde\mu_1(z_2)\leq z_1< \mu_1^N$ and $u_2$ increases if $\widetilde \mu_2(z_1)\leq z_2< \mu_2^*$.
\end{proposition}

It is unambiguous and relatively straightforward that an unjustified player's payoff strictly decreases  when their initial reputation declines or they become more impatient.\footnote{The first is due to the strict monotonicity of the reputation coevolution curve, and the second is because of the monotone shift of the curve with respect to a player's concession rate.} However, it may not be straightforward to see the unambiguity in the effects of changes in $c_1$, $k_2$, and $w_1$\threeemdashes{}what we call the details of the court. For example, when $c_1$ increases, there are two opposite effects. On one hand, a strategic player 1 is less likely to challenge because it is more costly. On the other hand, because a strategic player 1 is less likely to challenge, when facing a challenge player 2 is less likely to face a strategic player 1, and hence is more likely to yield, which may increase the value of a challenge and hence player 1's payoff. However, in equilibrium, this second effect is moot, because in equilibrium the value of a challenge is taken away by player 2's adjustment of her strategy to render player 1 indifferent between challenging and not challenging. While similar logic applies to changes in $k_2$ and $w_1$, obtaining the unambiguous results requires accounting for various shifts in both the speed of reputation building and the threshold beliefs that divide the challenge and no-challenge phases. Recall that in contrast, these details of the court will not affect players' payoffs in the limit case of rationality\threeemdashes{}i.e., when $z_1$ and $z_2$ approach $0$.

As a result of the indeterminacy in the benefit in the introduction of the ultimatum opportunity, the effects of a local increase in the ultimatum opportunity arrival rate on strategic players' payoffs are also ambiguous.\footnote{Players' payoffs are also nonmonotonic in demands $a_1$ and $a_2$, a result that is similar to \citet{AbreuGul2000Ecta} and used by \citet{Sanktjohanser2020}.}
\begin{remark}[{\bf Effects of a change in the arrival rate of ultimatum opportunities}]
\label{remark:gamma1}
Start with a bargaining game $B=(a_1,a_2,z_1,z_2,r_1,r_2,\gamma_1,c_1,k_2,w_1)$ with one-sided ultimatum and single demand types. When $\gamma_1$ increases, player 1's payoff strictly decreases if $z_1\geq \widetilde\mu_1(z_2)$ and $z_1\geq \mu_1^N$; may increase or decrease with the same sign as
\[
\frac{\lambda_{2}}{\gamma_{1}-\lambda_{1}}\left\{ -\frac{1}{\lambda_{2}}\log\left[\widetilde{\mu}_{2}(z_{1})\right]+\frac{\frac{1}{\nu_{1}^{*}}-\frac{1}{z_{1}}}{\frac{1}{z_{1}}+\gamma_{1}\left(\frac{1}{\nu_{1}^{*}}-\frac{1}{z_{1}}\right)}+\frac{1-\nu_{1}^{*}}{\nu_{1}^{*}}\frac{1}{\lambda_{1}}(\mu_{2}^{*})^{\frac{\lambda_{1}-\gamma_{1}}{\lambda_{2}}}\left[\frac{\gamma_{1}}{\lambda_{2}}\log(\mu_{2}^{*})-1\right]\right\}
\]
if $\widetilde\mu_1(z_2)\leq z_1\leq \mu_1^N$; and stays constant otherwise.
\end{remark}

\section{One-sided ultimatum and multiple demand types}\label{sec:multiple}
In this section, we consider the case in which there are multiple justifiable demands for both players. Formally, player 1 announces a demand $a_1\in A_1$ first, and upon observing player 1's announcement, player 2 either accepts the demand or rejects the demand and announces her own demand $a_2\in A_2$. Assume $A_1$ and $A_2$ are finite; to make the problem nontrivial, we create some conflicts between the demands by assuming that player $i$'s maximal demand is incompatible with all demands of player $j$: $\max A_i+\min A_j>1$. The prior conditional probability distribution $\pi_i$ of demands by a justified player $i$ is commonly known. The game then proceeds as in the previous case with one-sided ultimatum and single demand types for both players. Hence, a game with one-sided ultimatum and multiple demands is described by ($\pi_1$, $\pi_2$, $z_1$, $z_2$, $r_1$, $r_2$, $\gamma_1$, $c_1$, $k_2$, $w_1$). In addition to choosing their subsequent challenge, concession, and response to challenges, strategic players choose initial demands to mimic. Let $\sigma_1\in \Delta(A_1)$ denote a strategic player 1's mimicking strategy at the beginning of the game, and $\sigma_2(\cdot|a_1)$ a strategic player 2's mimicking strategy upon observing player 1's announced demand $a_1$, where the argument can be any $a_2\in A_2$ or $\{0\}$, which means accepting player 1's demand $a_1$.

\subsection{Equilibrium}
\begin{theorem}\label{thm:1multiple}
For any bargaining game $\left(\pi_1,\pi_2,z_1,z_2,r_1,r_2,\gamma_1,c_1,k_2,w_1\right)$ with one-sided ultimatum and multiple demand types for both players, all equilibria yield the same distribution over outcomes.
\end{theorem}

The proof is fairly similar to the proof in {\AG}. We include a brief discussion here and the details in Appendix \ref{sec:1multiple} for the sake of completeness. The key property\threeemdashes{}that players' payoffs are monotonic in $z_i$\threeemdashes{}is preserved in the current setting, as we have shown in the comparative statics exercises. In the proof, we will first consider the intermediate case in which there is only one justified type of player 1 but there are several justified types of player 2. In this case, a unique equilibrium exists. 

Then we look at the general case in which player 1 first chooses which type $a_1\in A_1$ to mimic, and seeing this, player 2 responds with a type $a_2\in A_2$ to mimic. In this case, we show that the distribution of equilibrium outcomes is unique, which we complete in the appendix.

Note that the equilibrium outcome does depend on the order of the move. If player 2 announces the demand before player 1, then the distribution of equilibrium outcomes is still unique but potentially different from that when player 1 announces first. However, these orders will be irrelevant in the limit case of rationality and rich type space, as we show in the next subsection.

\subsection{Limit case of rationality and rich type space}\label{sec:limit}

We investigate the limit case of rationality when the set of available demand types for each player is fine. The purpose of the analysis is to investigate which types stand out as the ones that are mimicked most often.

For $K\in \mathbb{Z}_{> 0}$, let $A^K:=\{2/K,3/K,...,(K-1)/K\}$ be a set of demands. Each element of $A^K$ corresponds to a commitment type whose demand coincides with that element. Suppose that $\pi_i\in \Delta(A^K)$ with full support, i.e., the prior distribution of player $i$'s type conditional on player $i$ being justified has full support on $A^K$. Finally, let $z_i^n$ be the probability that player $i$ is a justified type. Hence,  $z_i^n\pi_i(k/K)$ is the probability that player $i$ is a justified type who demands $k/K$, for $k=2,..,K-1$.

In what follows, we fix $K$ and analyze the equilibrium sequence of a sequence of bargaining games in which the probabilities of each player being justified go to zero on the same order for the two players.

\begin{proposition}\label{prop:1limit}
Let $\{B^n\}_n$ be a sequence of games in which for each $n\in \mathbb{N}$,  $B^n$ $=$ ($\pi_1$, $\pi_2$, $z_1^n$, $z_2^n$, $r_1$, $r_2$, $\gamma_1$, $c_1$, $k_2$, $w_1$) is a bargaining game with one-sided ultimatum and rich type spaces. If $\displaystyle\lim_{n\rightarrow\infty} z_1^n= \lim_{n\rightarrow\infty} z_2^n =0$, $\displaystyle\lim_{n\rightarrow\infty} z_1^n/z_2^n \in (0,\infty)$, and $u_i^n$ is the equilibrium payoff for player $i$ in the $n^{\text{th}}$ game of the sequence, then
$$\lim\inf u_1^n > \frac{r_2}{\max\{r_1,\gamma_1\}+r_2} -1/K,$$
$$\lim\inf u_2^n > \frac{\max\{r_1,\gamma_1\}}{\max\{r_1,\gamma_1\}+r_2} -1/K.$$
\end{proposition}

\begin{remark}

Proposition \ref{prop:1limit} implies that $\lim\sup u_i^n\leq 1-\lim\inf u_{-i}^n$, because the size of the pie is 1. Therefore, as $K$ grows without bound, player 1's limit equilibrium payoff converges to $\frac{r_2}{\max\{r_1,\gamma_1\}+r_2}$, and player 2's limit equilibrium payoff converges to $\frac{\max\{r_1,\gamma_1\}}{\max\{r_1,\gamma_1\}+r_2}$.
\end{remark}

Proposition \ref{prop:1limit} illustrates how the bargaining power depends on the arrival of ultimatum opportunities in a remarkably simple way. The details of the court's decision rule do not affect players' payoffs. Moreover, ultimatums have no impact if their arrival rate is smaller than the discount rate, while their arrival rate takes the role of the discount rate otherwise. Finally, when ultimatum opportunities are arbitrarily frequent, i.e., as $\gamma_1\to \infty$, player 2 guarantees herself the highest justifiable demand.

Proposition \ref{prop:11limit} shows that the limit equilibrium outcome when each side has a single type is (generically) efficient, i.e., agreement is immediate. Moreover, player 1 wins if $\lambda_1-\gamma_1>\lambda_2$, and player 2 wins if $\lambda_1-\gamma_1<\lambda_2$. Writing this comparison in terms of the primitives of the model, we have that player 1 wins if
$$
r_2(1-a_1)>r_1(1-a_2)+\gamma_1(a_1+a_2-1),
$$
and player 2 wins if the strict inequality sign is flipped. Note that in {\AG}, the comparison is between $r_2(1-a_1)$ and $r_1(1-a_2)$\threeemdashes{}two terms that resemble the marginal costs of waiting that involve only demands and discount rates\threeemdashes{}to determine the winner. The comparison in our model is complicated by an additional term involving the ultimatum opportunity arrival rate $\gamma_1$ and the amount of disagreement $D$. The addition of the ultimatum opportunity cannot simply be thought of as a discount rate. Player $i$'s problem is to maximize $a_i$ subject to being the winner.

In the case of $\gamma_1\leq r_1$, which includes $\gamma_1=0$ in {\AG} as a special case, player 1 can guarantee being the winner by choosing the demand $\max \left\{ a_1 \in A^K \left| a_1 \leq \frac{r_2}{r_1+r_2}\right. \right\}$. The result holds because the inequality above can be rearranged as
$$
r_2(1-a_1)>(\gamma_1-r_1)(a_1+a_2-1)+r_1a_1 \Longleftrightarrow r_2-(r_1+r_2)a_1>(\gamma_1-r_1)(a_1+a_2-1).
$$
Given the negative term on the right-hand side of the inequality, player 1's Rubinstein-like demand guarantees his being the winner. Analogously, player 2 is the winner if
$$
r_1-(r_1+r_2)a_2>(-\gamma_1-r_2)(a_1+a_2-1),
$$
and she can guarantee being the winner by demanding $\max \left\{ a_2 \in A^K \left| a_2 \leq \frac{r_1}{r_1+r_2}\right. \right\}$.

However, when $r_1<\gamma_1$, player 1 can no longer guarantee  $\max \left\{ a_1 \in A^K \left| a_1 \leq \frac{r_2}{r_1+r_2}\right. \right\}$. Rearranging the inequality, we have that player 1 wins if
$$
r_2(1-a_1)>(r_1-\gamma_1)(1-a_2)+\gamma_1a_1 \Longleftrightarrow r_2-(r_2+\gamma_1)a_1>(r_1-\gamma_1)(1-a_2).
$$
Given that the right-hand side of the inequality is negative, but can be close to 0, player 1 can guarantee winning by choosing any $a_1 \leq \frac{r_2}{\gamma_1+r_2}$. 

Conversely, player 2 can guarantee the payoff $\frac{\gamma_1}{\gamma_1+r_2}-1/K$ by choosing the demand $1-1/K$ (the inequality is flipped whenever $a_1$ is at least $\frac{r_2}{\gamma_1+r_2}+1/K$). Observe that player 2 guarantees this high payoff by choosing the greediest demand. This is in contrast to the existing results in the literature, in which players tend to make compromise demands to get their Rubinstein-like payoffs.

Note that none of the arguments above depends on the order of moves, so the limit payoffs in a rich type space are independent of the order of players' moves.

\section{Two-sided ultimatum}\label{sec:twosided}
Now consider the setting in which each player $i=1,2$ has a single demand type $a_i$, with the amount of disagreement $D=a_1+a_2-1>0$, but both players can ultimate. Specifically, a justified player $i$ challenges according to a Poisson process with arrival rate $\gamma_i\geq 0$, and an unjustified player can time their challenge strategically. At each instant $t$, each justified player can (i) give in to the other player's demand, (ii) hold on to their demand, or (iii) challenge. If the players neither challenge nor concede, then the game continues. Player $i$ who challenges at time $t$ incurs a cost $c_iD$ and player $j\neq i$ must respond to the challenge, by either yielding to the challenge and getting $1-a_j$, or seeing the challenge by paying a cost $k_jD$. When player $j$ sees the challenge, the shares of the pie are determined as follows. An unjustified player $i$'s payoff against a justified player $j$ is $1-a_j$. If two unjustified players meet, then the challenging player $i$ wins with probability $w_i<1/2$: Player $i$ gets $a_i$ with probability $w_i$ and $1-a_j$ with probability $1-w_i$, so the challenging player $i$'s expected payoff is $1-a_j+w_iD$, and the defending player $j$'s expected payoff is $1-a_i+(1-w_i)D$. To make challenging and seeing a challenge worthwhile for player $i$, assume $w_i<c_i<1$ and $0<k_i<1-w_i$ for $i=1,2$.

In summary, $\left(\{a_i, z_i, r_i, \gamma_i, c_i, k_i, w_i\}_{i=1}^2\right)$, \emph{a bargaining game with two-sided ultimatum and single demand types}, is described by demands $a_1$ and $a_2$, players' prior probabilities $z_1$ and $z_2$ of being justified, discount rates $r_1$ and $r_2$, challenge opportunity arrival rates $\gamma_1$ and $\gamma_2$, challenge costs $c_1D$ and $c_2D$, seeing costs $k_1D$ and $k_2D$, and unjustified challengers' winning probabilities $w_1$ and $w_2$ against unjustified defendants.

Formally, let $\Sigma_i=(F_i,G_i,q_i)$ denote an unjustified player $i$'s strategy, where $F_i(t)$ is player $i$'s probability of conceding by time $t$, $G_i(t)$ is player $i$'s probability of challenging by time $t$, and $q_i(t)$ is player $i$'s probability of conceding to a challenge at time $t$. Restrict $F_i$ and $G_i$ to be right-continuous and increasing functions with $F_i(t)+G_i(t)\le 1$ for every $t\ge 0$, and $q_i(t)\in [0,1]$ to be a measurable function. We again study the Bayesian Nash equilibrium of this game. The belief process is naturally defined, with $\mu_i(t)$, $\nu_i(t)$, and $\chi_i(t)$ analogously defined as in the game with one-sided ultimatum.

\begin{figure}[t!]
\begin{subfigure}{0.33\textwidth}
\includegraphics[width=\textwidth]{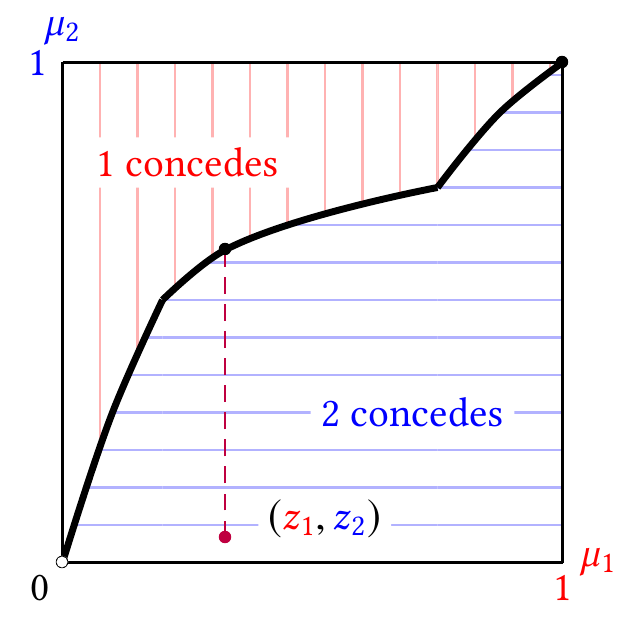}
\caption{\label{subfig:finite1}$\gamma_1\leq \lambda_1$ and $\gamma_2\le \lambda_2$.}
\end{subfigure}
\begin{subfigure}{0.33\textwidth}
\includegraphics[width=\textwidth]{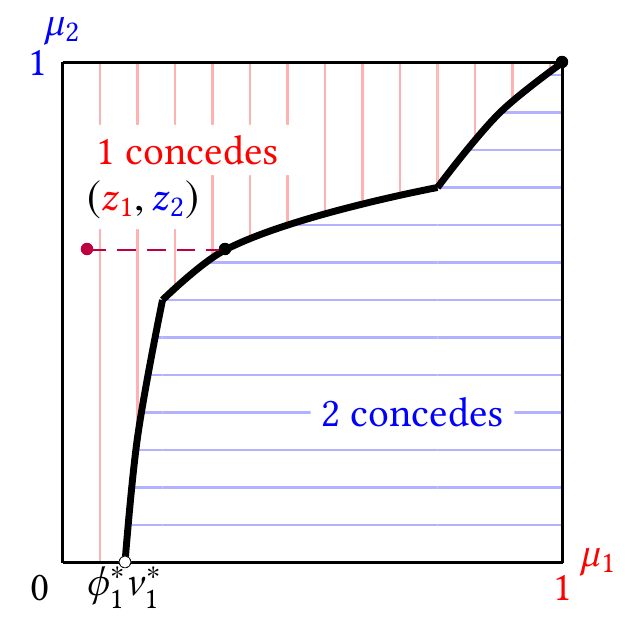}
\caption{\label{subfig:finite2}$\gamma_1 > \lambda_1$ and $\gamma_2\leq \lambda_2$.}
\end{subfigure}
\begin{subfigure}{0.33\textwidth}
\includegraphics[width=\textwidth]{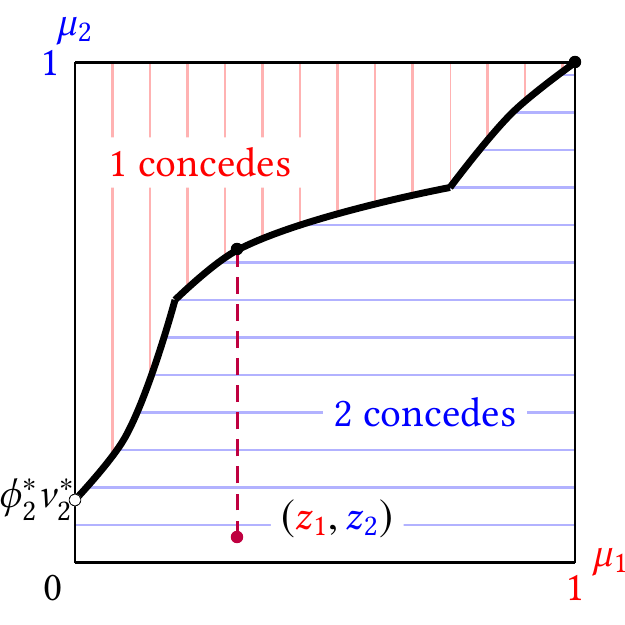}
\caption{\label{subfig:finite3}$\gamma_1 \le \lambda_1$ and $\gamma_2> \lambda_2$.}
\end{subfigure}
\caption{\label{fig:finite}
Reputation coevolution curves and initial concessions in games with two-sided ultimatum and single demand types when $\gamma_i\le \lambda_i$ for some $i=1,2$.}
{\footnotesize There is a unique equilibrium outcome, and the game ends in finite time. The reputation coevolution curve divides the plane into two regions that differ in the player who concedes with a positive probability at time $0$. The curve tends to $(0,0)$ when $\gamma_1\le \lambda_1$ and $\gamma_2\le \lambda_2$, to $(\phi_1^*\nu_1^*,0)$ when $\gamma_1> \lambda_1$ and $\gamma_2\le \lambda_2$, and to $(0,\phi_2^*\nu_2^*)$ when $\gamma_1\le \lambda_1$ and $\gamma_2> \lambda_2$.}
\end{figure}

\subsection{Unique equilibrium in games with sufficiently slow ultimatum opportunity arrival for at least one player}
There is a unique equilibrium outcome under the assumption that $\gamma_i\leq \lambda_i:=r_j(1-a_i)/D$ for some $i=1,2$. This assumption is automatically satisfied in the one-sided ultimatum setting, which is essentially a two-sided ultimatum setting with $\gamma_2=0<\lambda_2$. This assumption guarantees that the reputations always increase in equilibrium and the game ends in finite time. The four properties in Theorem \ref{T:existenceandproperties} are modified to incorporate the possibility of player 2 challenging, as follows.

\begin{theorem}\label{T:existenceandproperties2}
Consider a bargaining game $B=\left(\{a_i, z_i, r_i, \gamma_i, c_i, k_i, w_i\}_{i=1}^2\right)$ with two-sided ultimatum and single demand types. If $\lambda_i\geq \gamma_i$ for some $i=1,2$, there exist finite times $T$ and $T_1, T_2 \in [0,T)$ such that equilibrium strategies satisfy the following properties. For both $i=1,2$,
\begin{enumerate}
\item $\widehat{F}_i$ is strictly increasing in $(0,T)$ and constant for $t\geq T$;\label{property'1}
\item $\widehat{F}_i$ is atomless in $(0,T]$ and at most one of the two has an atom at $t=0$;\label{property'2}
\item\label{property'3} \begin{enumerate}
    \item $\widehat{G}_i$ is atomless, strictly increasing in $[0,T_i]$, and constant for $t\geq T_i$;\label{property'3a}
    \item For almost every $t\in [0,T]$, $\widehat q_i(t)\in (0,1)$ if $t\in [0,T_i]$ and $\widehat q_i(t)=1$ if $t\in (T_i,T]$;\label{property'3b}
\end{enumerate}
\item $\widehat{F}_i(T)+\widehat{G}_i(T_i)=1$.\label{property'4}
\end{enumerate}
Moreover, $\widehat F_i$ and $\widehat G_i$ are unique, and $\widehat q_i$ is unique almost everywhere for $t\le T$.
\end{theorem}

We include the derivation of equilibrium strategies in Appendix \ref{sec:twosidedslow}, which largely modifies the derivation in the one-sided ultimatum setting. Equilibrium strategies are analogous to those in the one-sided ultimatum setting: After at most one player concedes initially, each player $i$ concedes at the overall {\AG} concession rate $\lambda_i$, each strategic player $i$ challenges at an increasing rate $\chi_i(t)=\frac{1-\nu_i^*}{\nu_i^*}\frac{\mu_i(t)}{1-\mu_i(t)}\gamma_i$ up to time $T_i$ to guarantee a challenger $i$ a reputation $\nu_i^*:=1-k_j/(1-w_i)$, the level that renders an unjustified opponent $j$ indifferent between seeing and yielding to a challenge.

Again, the reputation coevolution diagram can be used to determine the player and magnitude of the initial concession. Figure \ref{fig:finite} illustrates the three possible reputation coevolution curves when $\gamma_i\le \lambda_i$ for some $i=1,2$. When $\gamma_i\le \lambda_i$ for both players (Figure \ref{subfig:finite1}), the reputation coevolution curve tends to $(0,0)$. When $\gamma_i>\lambda_i$ for some $i=1,2$  (Figures \ref{subfig:finite2} and \ref{subfig:finite3}), the reputation coevolution curve tends to the intercept $\phi_i^*\nu_i^*$, where $\phi_i^*:=1-\lambda_i/\gamma_i$.

The implications in this setting with two-sided ultimatum and slow arrival of ultimatum opportunities for at least one side are mostly analogous to those in the setting with one-sided ultimatum. Namely, the hazard rates are discontinuous and piecewise monotonic, with the possibility of having two discontinuities at the finite times when each player ends challenging (modifying the one discontinuity at the finite time when player 1 ends challenging in Proposition \ref{prop:hazardrates}). Ultimatum opportunities may benefit or hurt players (preserving the qualitative results of Proposition \ref{prop:gamma1}), but definitely hurt them in the limit case of rationality, i.e., the case with vanishing probabilities of being justified (preserving the qualitative results of Proposition \ref{prop:11limit}). More precisely, in the limit case of rationality, the outcome is efficient if $\lambda_1-\gamma_1\neq \lambda_2-\gamma_2$, and the winner is player $i$ if $\lambda_i-\gamma_i > \lambda_j-\gamma_j$ (modifying Proposition \ref{prop:11limit}). The comparative statics results in Proposition \ref{prop:compstat1} are generalized for both $i=1,2$, with player $i$'s payoff (weakly) hurt by decreasing initial reputation $z_i$, increasing discount rate $r_i$, increasing challenging cost $c_i$, increasing challenge response cost $k_i$, and decreasing challenge winning probability $w_i$.

\subsection{The possibility of multiple equilibria and infinite delay in games with fast ultimatum opportunity arrival for both players}

One main difference from the one-sided ultimatum setting is that when $\gamma_i>\lambda_i$ for both $i=1,2$ and both players' initial reputations are sufficiently small, there are equilibria in which reputations do not reach 1 and/or do not build up at all, and possibly multiple equilibria with varying initial concession possibilities. Consequently, inefficient infinite delay (i.e., $T=\infty$) may arise. The inefficient infinite delays manifest in two classes of equilibria. In the first class, players concede at {\AG} rates, but their reputations cannot build up because of the fast arrival of ultimatum opportunities for justified types, and consequently they challenge at decreasing rates. While players' reputations approach zero, they never reach it. This type of equilibria, with ever declining reputations, exists when both players' initial reputations are sufficiently small. In this case, one of the players may concede with a strictly positive\threeemdashes{}but sufficiently small\threeemdashes{}probability at time zero, and still both players experience subsequent declining reputations. This creates the indeterminacy of the initial concessions and the existence of a continuum of equilibria with different initial concession probabilities by different players.\footnote{There may also be equilibria in which one player's reputation stays at a positive constant and the other's reputation declines to zero but never reaches it. If the reputations before or after initial concessions lie on the purple lines in Figure \ref{fig:inf}, such equilibria arise.} In the second class of equilibria, the players concede at {\AG} rates and reputations may decrease or increase toward an absorbing belief $\mu_i^*:=1-c_i$, the reputation level that renders the opponent indifferent between challenging and not challenging. Upon the reputation reaching this absorbing level, the challenge rates balance the exit of unjustified and justified types for each player such that their reputations, conditional on the game not ending, stay constant at $\mu_1^*$ and $\mu_2^*$, respectively. This second class of equilibria may or may not exist, depending on the parameters of the model.
\begin{figure}[t!]
{
\begin{subfigure}{0.33\textwidth}
\includegraphics[width=\textwidth]{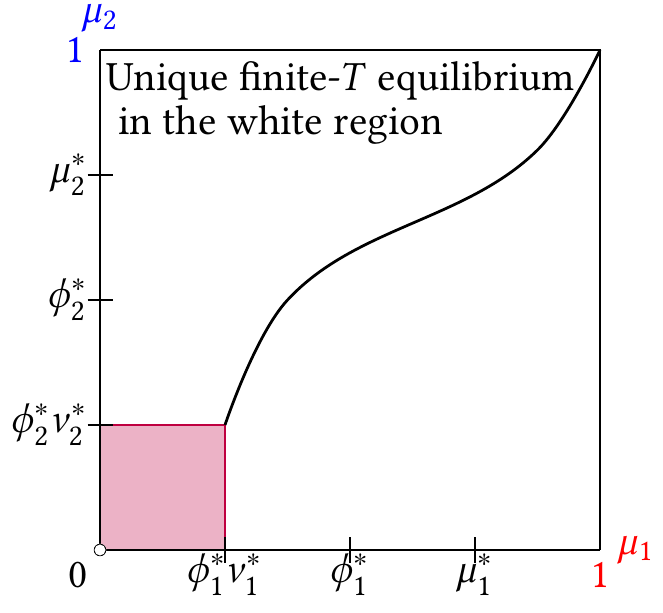}
\caption{\label{subfig:inf1}$\mu_i^*>\phi_i^*$, $i=1,2$.}
\end{subfigure}
\begin{subfigure}{0.33\textwidth}
\includegraphics[width=\textwidth]{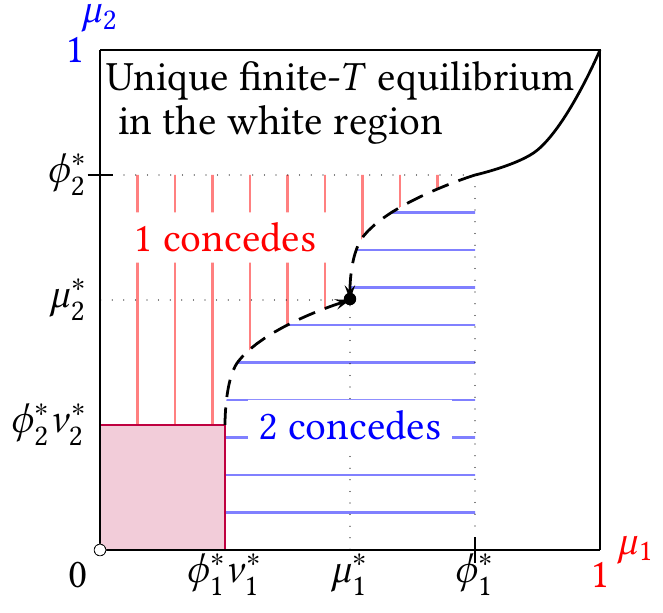}
\caption{\label{subfig:inf2}$\mu_i^*\in (\phi_i^*\nu_i^*,\phi_i^*)$, $i=1,2$.}
\end{subfigure}
\begin{subfigure}{0.33\textwidth}
\includegraphics[width=\textwidth]{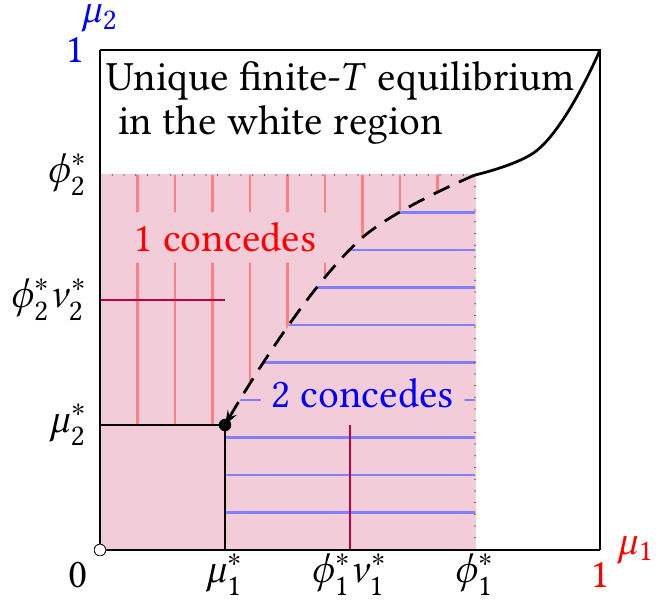}
\caption{\label{subfig:inf3}$\mu_i^*<\phi_i^*\nu_i^*$, $i=1,2$.}
\end{subfigure}
\caption{\label{fig:inf}
Demonstration of the range of initial reputations with infinite-delay equilibria in bargaining games with two-sided ultimatum and single demand types when $\gamma_i>\lambda_i$ for both $i=1,2$.
}
{\footnotesize (a) Type-1 equilibria in which players concede at {\AG} rates for $t>0$ exist if $z_i\le \phi_i^*\nu_i^*$, and there is no type-2 equilibrium, one in which both players' reputations eventually converge to $(\mu_1^*,\mu_2^*)$. (b) Type-1 equilibria exist if $z_i\le \phi_i^*\nu_i^*$ for both $i$, and type-2 equilibria exist if $z_i\in (\phi_i^*\nu_i^*, \phi_i^*)$ for at least one $i$ and $z_i<\phi_i^*$ for both $i$; (c) Type-1 equilibria exist if $z_i<\phi_i^*\nu_i^*$ for both $i$, and type-2 equilibria exist if $z_i\in (\phi_i^*\nu_i^*, \phi_i^*)$ for at least one $i$ and $z_i<\phi_i^*$ for both $i$. In the regions not covered, a unique finite-$T$ equilibrium exists.
}
}
\end{figure}

Figure \ref{fig:inf} illustrates the regions of initial reputations with these two classes of equilibria with possibly infinite delays. The first class of equilibria always exists when $\gamma_i>\lambda_i$ for both $i=1,2$ for a range of initial reputations (the purple areas in the graphs, with the boundary highlighted if such equilibria may exist on it). The second type of equilibria (indicated by the player of initial concession in the graphs) may not exist (Figure \ref{subfig:inf1}), may exist as the unique equilibrium in a range of initial reputations (Figure \ref{subfig:inf2}), and may coexist with the first class of equilibria for a range of parameters (Figure \ref{subfig:inf3}). Appendix \ref{sec:twosidedfast} provides a comprehensive description of equilibrium reputations and strategies in this setting.\footnote{Note that the three demonstrations do not encapsulate all possible scenarios of the model. For example, the game in which $\mu_1^*>\phi_1^*$ but $\mu_2^*<\phi_2^*$ is not captured. However, in the cases not covered in the demonstrations, no new type of equilibria arises, and the characterization of equilibria falls into one of the three categories described.} Multiplicity of equilibria arises in previous reputational bargaining models (e.g., \citet{AtakanEkmekci2014REStud} and \citet{Sanktjohanser2020}), but to the best of our knowledge, multiplicity due to inefficient infinite delays and reputations not building up is a new feature in the literature.

When $\lambda_i<\gamma_i$ for both $i=1,2$, payoffs in the limit case of rationality are indeterminate due to the multiplicity of equilibria. As in the other cases, efficient equilibria with no delay can be sustained in the limit. However, different from the other cases, the most inefficient equilibrium in which player $i$'s payoff is $1-a_j$ for both $i=1,2$ can also be sustained. Thus, the fast arrival of challenge opportunities may be detrimental for efficiency.

The analysis of multiple demand types is feasible, but will inevitably lead to a multiplicity of outcomes. This multiplicity also carries into the limit case of rationality with a rich type space. Performing a more predictive analysis requires additional criteria to select from multiple equilibria.

\section{Relation to literature}\label{sec:literature}
Our paper builds on the seminal work of \citet{AbreuGul2000Ecta}, which introduces the two-sided reputational bargaining model.\footnote{\citet{Myerson1991book} introduces one-sided reputational bargaining. Subsequent contributions to reputational bargaining include \citet{Kambe1999GEB}; \citet{AbreuPearce2007Ecta};
\citet{Wolitzky2011GEB, Wolitzky2012Ecta}; \citet{AtakanEkmekci2013JET};  \citet{AbreuPearceStacchetti2015TE}; and \citet{Sanktjohanser2020}. See \citet{FanningWolitzky2019} for a comprehensive survey.} They show the convergence of the equilibrium outcomes of discrete-time bargaining games with incomplete information to the unique equilibrium of a continuous-time war-of-attrition model. We build on their war-of-attrition model by adding the opportunity for players to ultimate. When the exogenous arrival rate of ultimatum opportunities to the justified type is zero, our model is equivalent to {\AG}'s model. When this arrival rate is strictly positive, a new possibility of negotiations being resolved by a nonstrategic third party opens up. Compared with {\AG}, our model requires new techniques and leads to new predictions. Specifically, (i) the addition of ultimatum opportunities results in richer yet tractable strategic behavior and reputation dynamics, solved by new methods (elaborated below); (ii) the hazard rate of dispute resolution is discontinuous and piecewise monotonic in time; (iii) the payoffs in the limit case of rationality and rich type spaces depart from {\AG}'s payoffs when the ultimatum opportunity arrival rate exceeds the discount rate; and (iv) reputation building and the efficient division of surplus in the limit case of rationality may fail when ultimatum opportunities are abundant for both players.

Our analysis has two main technical differences from {\AG}'s. First, in our model players have a larger strategy space due to the challenge opportunities. A priori, players may have more or less incentive for waiting to concede due to anticipation of challenges. However, we show that in equilibrium, a player's payoff when being challenged is equal to the payoff from conceding. Moreover, equilibrium distribution of challenges is continuously strictly increasing up to a finite time, and halting afterward. These findings allow us to show that the equilibrium structure of our model is a tractable enrichment of {\AG}'s.

Second, in {\AG}'s model, players' equilibrium behavior does not depend on their opponent's reputation, whereas in our model it inevitably does, as we note above. {\AG} develops a ``forward-looking'' method that first calculates the time it takes for each player's reputation to reach 1 in the absence of an initial concession to determine the winning player, and then characterizes the initial concession probability to ensure that players' reputations reach 1 at the same time. This method no longer applies to our model, because of the interdependence of the evolution of players' reputations. Instead, we develop a ``backward-looking'' method that characterizes players' reputations jointly on a diagram. The reputation coevolution curve, which depicts players' reputations as functions of each other's reputation, characterizes the locus of players' reputations in any equilibrium of all games with all possible initial reputations, after the start of the game. This locus divides the reputation plane into two regions that identify the winning player and the initial concession of the losing player.\footnote{We also generalize the locus to regions in the setting with two-sided ultimatum opportunities to represent all equilibrium reputations after initial concessions. \citet{KrepsWilson1982JET} have a similar representation of the state space by two players' reputations and a similar curve that divides the plane into two areas, but they do not use the reputation coevolution curve to derive the probability of initial concession or pin down additional strategy dynamics, as we do in the setting with two-sided ultimatum opportunities.}

Three important features differ from previous literature of reputational bargaining: (i) each player's disagreement payoff depends the opponent's type, (ii) the distribution of deadlines is endogenous, and (iii) players' outside options are endogenously evolving. First, the dependence of players' payoffs on players' and opponents' types, has not been studied in reputational bargaining. See \citet{Pei2020Ecta} for reputation effects under interdependent values.

The ultimatum in our model can be seen as invoking an immediate deadline. \citet{Fanning2016Ecta} studies reputational bargaining with exogenous deadlines, and obtain a monotonic hazard rate of dispute resolution when the deadline distribution is tightly compressed in a time interval. In our model, we assume that the arrival rate of ultimatum opportunities to the justified type is constant, yet we obtain a piecewise monotonic rate of dispute resolution in the middle of the negotiation due to the endogeneity of the ultimatum usage rates by strategic players. In addition, we obtain a discontinuity in the hazard rate of resolution due to the endogeneity of the payoffs when an ultimatum is issued. Relatedly, \citet{Fanning2020} studies a reputational bargaining model in which a mediator makes nonbinding recommendations. In our model, our third party resembles an arbitrator who makes binding resolution, but only when consulted.

Another interpretation of the ultimatum is an endogenously evolving outside option. A player can use an ultimatum to have a third party cast a division of the surplus. \citet{CompteJehiel2002Ecta} study exogenous outside options that generate a value strictly higher than concession, and show that these high-value outside options cancel out reputation effects. \citet{AtakanEkmekci2014REStud} study reputational bargaining in a market setting with many buyers and sellers. In their model, the market serves as the endogenous outside option, and they show that even in the limit case of rationality inefficiency may persist. We obtain a similar inefficiency result when both players can ultimate frequently and when the probability of being justified is small. Whereas in \citet{AtakanEkmekci2014REStud} the cause of the inefficiency is that the players exercise their outside option when their opponent has built a reputation for being a commitment type, in our model the cause of the inefficiency is the inability of the players to build a reputation. In addition, the models of \citet{Ozyurt2014GEB, Ozyurt2015JEBO} share the similarity whereby the value of the outside option depends on the players' evolving reputations, but the motivations and the modeling choices of the papers are different otherwise.
There is a further related literature on the exogenous arrival of outside options in bargaining with one-sided incomplete information. In \citet{HwangLi2017JET} and \citet{Hwang2018GEB}, not taking an outside option opens up the possibility of nonincreasing reputations and equilibrium multiplicity.

\section{Conclusion}\label{sec:conclusion}

We study negotiation when two parties have private information about the justifiability of their demand and have chances to issue an ultimatum to end the bargaining process by verifying the demand justifiability. In our stationary setting, equilibrium hazard rates of ultimatum and conflict resolution are discontinuous and piecewise monotonic in time. The presence of ultimatum opportunity affects reputation building in two opposite directions: The opportunity erodes a player's commitment power, but if used appropriately, the ultimatum can be used as an effective strategic posture. However, in the limit case of rationality, the ultimatum opportunity is detrimental. For sufficiently fast arrival of ultimatum opportunities, the opportunity arrival rate replaces the discount rate in the determination of the limit payoff of the players. When both players have frequent opportunities to challenge, reputations may not build up in equilibrium, and multiple equilibria arise.

There are further questions worth exploring. For example, we can model continuous-discrete-time games and study other equilibria in which players' continuation payoffs after revealing rationality do not coincide with their concession payoffs. Another direction would be to include deadlines, and finally,  nonstationary arrival rates of ultimatums opportunities or more complex demands such as nonstationary justified demands.

\bibliographystyle{econometrica}
\bibliography{rep}

\appendix

\section{Omitted proofs}\label{sec:proofs}
\subsection{Proof of Theorem \ref{T:existenceandproperties}}

\begin{proof}[{\bf Proof of Theorem \ref{T:existenceandproperties}}]
Let $\widehat \Sigma=(\widehat \Sigma_1,\widehat \Sigma_2)=((\widehat F_1(\cdot),\widehat G_1(\cdot)),(\widehat F_2(\cdot),\widehat q_2(\cdot)))$ denote an equilibrium strategy profile. We argue that $\widehat \Sigma$ must have the form specified in the theorem (hence proving the uniqueness of equilibrium outcome) and that these strategies indeed define an equilibrium (hence proving the existence of equilibrium strategies). Let $u_i(t)$ denote the expected utility of an unjustified player $i$ who concedes at time $t$. Define $\mathcal T_i:=\{t|u_i(t)=\max_s u_i(s)\}$ as the set of conceding times that attain the highest expected utility for player $i$ given opponent $j$'s strategy $\widehat \Sigma_j$. Because $\widehat \Sigma$ is an equilibrium, $\mathcal T_i$ is nonempty for $i=1,2$. Furthermore, define $\tau_i:=\inf\{t\geq 0|\widehat F_i(t)=\lim_{s\rightarrow \infty}\widehat F_i(s)\}$ as the time of last concession for player $i$, with $\inf \emptyset:=\infty$. Finally, the support of player 1's challenge distribution is $[0,\infty)$ due to the justified type's challenge behavior. Hence, in any equilibrium, $\widehat q_2(t)$ maximizes player 2's expected payoff at time $t$ when she faces a challenge when player 1's reputation is $\nu_1(t)$ upon challenging, for almost every $t\leq \tau_2$ in both the $\widehat G_1$ measure and the Lebesgue measure. In the remainder of the proof, we will drop the ``almost everywhere'' qualifier. Then we have the following results.
\begin{enumerate}[leftmargin=*, label=(\emph{\alph*})]

\item\label{result(a)}
{\bfseries $\widehat G_1$ is continuous for $t\geq 0$.} To show that $\widehat G_1$ does not have any atoms, suppose to the contrary that $\widehat G_1$ jumps at time $t$ so that an unjustified player 1 challenges with a positive probability at time $t$; that is, $\widehat G_1(t)>0$ for $t=0$, or $\widehat G_1(t)-\widehat G_1(t^-)>0$ for $t>0$. Given that an unjustified player 1 challenges with a positive probability and a justified player 1 challenges with probability 0, player 2 facing a challenge believes that a challenging player 1 is unjustified with probability 1: $\nu_1(t)=0$. Consequently, she is strictly better off responding to the challenge and obtaining a payoff of $1-a_1+(1-w_1)D-k_2D$ than yielding to the challenge and obtaining a payoff of $1-a_1$, because $k_2<1-w_1$ by assumption. But if player 2 responds to a challenge with probability 1, an unjustified player 1's payoff from challenging is less than $1-a_1+w_1D-c_1D$ (an unjustified player 1's expected payoff when the player 2 who responds to a challenge is unjustified with probability 1), which is strictly less than his payoff from conceding, because $c_1>w_1$ by assumption, so an unjustified player 1 has a profitable deviation to conceding at $t$ from challenging with a positive probability at $t$, a contradiction.

\item\label{result(b)}
{\bfseries $\widehat q_2(t)$ is positive for almost all $t\leq \tau_2$.} Suppose to the contrary that $\widehat q_2(t)=0$ on a set $A$ of positive Lebesgue measure. Then $\int_Ad\widehat G_1(t)dt=0$. Then $\nu_1(t)=1$ for almost every $t\in A$. Then $\widehat q_2(t)=1$ for $t\in A$ is a profitable deviation, a contradiction.

\item\label{result(c)}
{\bf Player 2's payoff when being challenged at time $t$ is $1-a_1$ for almost all $t\leq \tau_2$.}
Whenever an unjustified player 2 yields to a challenge with a positive probability at time $t$ in equilibrium, her payoff when being challenged at time $t$ is equal to $1-a_1$. By \ref{result(b)}, player 2 yields to a challenge with a positive probability for almost all $t\leq \tau_2$, so her payoff when being challenged at time $t$ is equal to $1-a_1$.

\item\label{result(d)}
{\bfseries  The last instant at which two unjustified players concede is the same: $\tau_1=\tau_2$.}
An unjustified player will not delay conceding upon learning that the opponent will never concede. Note that even if an unjustified player 1 might challenge with a positive probability but never concedes, an unjustified player 2's payoff from being challenged is $1-a_1$ (by \ref{result(c)}), so she does not benefit from waiting for a challenge. Denote the last concession time by $\tau$.

\item\label{result(e)} 
{\bfseries If $\widehat F_i$ jumps at $t$, then $\widehat F_j$ does not jump at $t$ for $j\neq i$.}
If $\widehat F_i$ has a jump at $t$, then player $j$ receives a strictly higher utility by conceding an instant after $t$ than by conceding exactly at $t$; note that whether or not player 1 challenges at $t$ does not affect the result, by \ref{result(c)}.

\item\label{result(f)}
{\bfseries  If $\widehat F_2$ is continuous at time $t$, then $u_1(s)$ is continuous at $s=t$. If $\widehat F_1$ and $\widehat G_1$ are continuous at time $t$, then $u_2(s)$ is continuous at $s=t$.}
These claims follow immediately from the definition of $u_1(s)$ in Equation \eqref{eq:u1t} and the definition of $u_2(s)$ in Equation \eqref{eq:u2t}, respectively.

\item\label{result(g)} 
{\bfseries There is no interval $(t',t'')$ $\subseteq [0, \tau]$ such that both $\widehat F_1$ and $\widehat F_2$ are constant on the interval $(t',t'')$.} Assume the contrary and without loss of generality, let $t^*\leq \tau$ be the supremum of $t''$ for which $(t', t'')$ satisfies the above properties. Fix $t\in (t',t^*)$ and note that for $\varepsilon$ small enough there exists $\delta>0$ such that $u_i(t)-\delta>u_i(s)$ for all $s\in (t^*-\varepsilon, t^*)$. In words, conditional on the opponent not conceding in an interval, it is strictly better for a player to concede earlier within that interval, and it is sufficiently significantly better by conceding early than by conceding close to the end of the time interval. By \ref{result(e)} and \ref{result(f)}, there exists $i$ such that $u_i(s)$ is continuous at $s=t^*$, so for some $\eta>0$, $u_i(s)<u_i(t)$ for all $s\in (t^*, t^*+\eta)$ (observe that this relies on player 2 not benefiting from waiting for a challenge from player 1, by \ref{result(c)}). In words, because of the continuity of the expected utility function at time $t^*$, the expected utility of conceding a bit after time $t^*$ is still lower than the expected utility of conceding at time $t$ within the time interval. Since $\widehat F_i$ is optimal, $\widehat F_i$ must be constant on the interval $(t',t^*+\eta)$. The optimality of $\widehat F_i$ implies that $\widehat F_j$ is also constant on the interval $(t',t^*+\eta)$, because player $j$ is strictly better off conceding before or after the interval than conceding during the interval. Hence, both functions are constant on the interval $(t',t^*+\eta) \subseteq (t',\tau)$. However, this contradicts the definition of $t^*$.

\item\label{result(h)}
{\bfseries  If $t'<t''<\tau$, then $\widehat F_i(t'')>\widehat F_i(t')$ for $i=1,2$.}
If $\widehat F_i$ is constant on some interval, then the optimality of $\widehat F_j$ implies that $\widehat F_j$ is constant on the same interval, for $j\neq i$ (again, by \ref{result(c)}). However, \ref{result(g)} shows that $\widehat F_1$ and $\widehat F_2$ cannot be constant simultaneously.

\item\label{result(i)}
{\bfseries $\widehat F_i$ is continuous for $t>0$.}
Assume the contrary: Suppose $\widehat F_i$ has a jump at time $t$. Then $\widehat F_j$ is constant on interval $(t-\varepsilon,t)$ for $j\neq i$. This contradicts \ref{result(h)}.

\end{enumerate}

\begin{enumerate}
\item Strictly increasing $\widehat F_1$ and $\widehat F_2$ for $t<T$ follows from \ref{result(h)}, and constant $\widehat F_1$ and $\widehat F_2$ for $t\ge T$ follows from \ref{result(d)}.

\item No atom for $\widehat F_i$ follows from \ref{result(i)}. At most one atom for $\widehat F_1$ and $\widehat F_2$ at $t=0$ follows from \ref{result(e)}.

\item 

\begin{enumerate}

\item $\widehat G_1$ has no atom follows from \ref{result(a)}, and \ref{result(b)} implies that $\widehat G_1$ is strictly increasing; if $\widehat G_1$ is constant, then $\widehat q_2(t)=1$, which contradicts \ref{result(b)}.

\item $\widehat q_2(t)\in (0,1)$ for $t\in [0,T_1]$ follows from \ref{result(b)}. From \ref{result(f)} and \ref{result(i)}, it follows that $v_1(t)$ is continuous on $(0,\tau]$. Furthermore, $v_1(t)$ is strictly smaller than $1-a_1$ when $\mu_2(t)>\mu_2^*$ (i.e., $\widehat F_2(t)>1-\frac{k_2}{1-z_2}$).
Therefore, after $\mu_2(t)>\mu_2^*$, a strategic player 1 does not challenge. Since player 2's reputation strictly increases over time, there is a finite time $T_1$ such that player 1 challenges from time $0$ to $T_1$ and does not challenge from $T_1$ onward. Hence, $\widehat q_2(t)=0$ for $t\ge T_1$.
\end{enumerate}

\item It follows from \ref{result(h)} that $\mathcal T_i$ is dense in $[0,\tau]$ for $i=1,2$. From \ref{result(d)}, \ref{result(f)}, and \ref{result(i)}, it follows that $u_i(s)$ is continuous on $(0,\tau]$, and hence $u_i(s)$ is constant for all $s\in (0,\tau]$. Consequently, $\mathcal T_i=(0,\tau]$. Hence, $u_i(t)$ is differentiable as a function of $t$ and $du_i(t)/dt=0$ for all $t\in (0,\tau)$.

In particular, player 1's expected utility from conceding at time $t$ is
\begin{equation}\label{eq:uit}
u_1(t)=(1-z_2)\int_0^t a_1 e^{-r_1 s}d\widehat F_2(s)+(1-a_2)e^{-r_1 t}[1-(1-z_2)\widehat F_2(t)].
\end{equation}
The differentiability of $\widehat F_2$ follows from the differentiability of $u_1(t)$ on $(0,\tau)$. Differentiating Equation \eqref{eq:uit} and applying Leibnitz's rule, we obtain
$$
0=a_1e^{-r_1t}(1-z_2)\widehat f_2(t)-(1-a_2)r_1e^{-r_1t}(1-(1-z_2)\widehat F_2(t))-(1-a_2)e^{-r_1t}(1-z_2)\widehat f_2(t),
$$
where $\widehat f_2(t)=d\widehat F_2(t)/dt$. This in turn implies $\widehat F_2(t)=\frac{1-C_2e^{-\lambda_2t}}{1-z_2}$, where constant $C_2$ is yet to be determined. This characterization implies that $\tau_2$ is finite. At $\tau_1=\tau_2$, optimality for player $i$ implies $\widehat F_1(\tau_1)+\widehat G_1(\tau_1)=1$ and $\widehat F_2(\tau_2)=1$.
\end{enumerate}

This completes the proof that the structure of equilibrium strategies is unique. We now proceed to show the uniqueness of equilibrium strategies. We derive the reputation coevolution diagram using the reputation dynamics in Section \ref{sec:reputation}. Recall
$$
\widetilde\mu_{1}(\mu_{2})=\begin{cases}
\frac{\lambda_{1}-\gamma_{1}}{\lambda_{1}(\mu_{2})^{\frac{\gamma_{1}-\lambda_{1}}{\lambda_{2}}}-\gamma_{1}} & \text{if }\mu_{2}^{*}<\mu_{2}\leq 1,\\
\frac{\lambda_{1}-\gamma_{1}}{\lambda_{1}(\mu_{2})^{\frac{\gamma_{1}-\lambda_{1}}{\lambda_{2}}}+\left(\frac{\gamma_{1}}{\nu_{1}^*}-\gamma_{1}\right)\left(\frac{\mu_{2}}{\mu_{2}^{*}}\right)^{\frac{\gamma_{1}-\lambda_{1}}{\lambda_{2}}}-\frac{\gamma_{1}}{\nu_{1}^*}} & \text{if }0<\mu_{2}\leq \mu_{2}^{*}.
\end{cases}
$$
The reputation coevolution curve is strictly increasing. $\widetilde\mu_1(\mu_2)$ is well defined for $\mu_2\in (0,1]$. Hence, the unique equilibrium entails $F_1(0)=0$ and $\widehat F_2(0)>0$ if $z_1<\widetilde\mu_1(z_2)$; $\widehat F_1(0)>0$ and $\widehat F_2(0)=0$ if $z_1>\widetilde\mu_1(z_2)$; and $\widehat F_1(0)=0$ and $\widehat F_2(0)=0$ if $z_1=\widetilde\mu_1(z_2)$. Moreover, $F_1(0)$ is uniquely determined by Equation \eqref{eq:Q}, and $\widehat F_2(0)$ is uniquely determined analogously. This completes the uniqueness of equilibrium strategies. 
\end{proof}

\subsection{Proof of Proposition \ref{prop:gamma1}}
\begin{proof}[{\bf Proof of Proposition \ref{prop:gamma1}}]
When $z_1\geq \mu_1^N$, a strategic player 1 cannot benefit from the introduction of the challenge opportunity, because $\widetilde\mu_2(z_1)$ strictly decreases as $\gamma_1$ increases when $z_1\geq \mu_1^N$. When $z_1<\mu_1^N$, player 1 strictly benefits given $z_2<\widetilde\mu_2(z_1|\gamma_1)$ if and only if $\widetilde\mu_2(z_1|\gamma_1)>\widetilde\mu_2(z_1|0)$. Explicitly,
$$
\left[\dfrac{\frac{\lambda_{1}-\gamma_{1}}{\lambda_{1}}\frac{1}{z_{1}}+\frac{\gamma_{1}}{\lambda_{1}}\frac{1}{\nu_{1}^{*}}}{1+\frac{1-\nu_{1}^{*}}{\nu_{1}^{*}}\frac{\gamma_{1}}{\lambda_{1}}(\mu_{2}^{*})^{\frac{\lambda_{1}-\gamma_{1}}{\lambda_{2}}}}\right]^{\frac{\lambda_{2}}{\gamma_{1}-\lambda_{1}}}>\left(\frac{1}{z_1}\right)^{-\frac{\lambda_2}{\lambda_1}}=\left[\left(\frac{1}{z_1}\right)^{\frac{\lambda_1-\gamma_1}{\lambda_1}}\right]^{\frac{\lambda_2}{\gamma_1-\lambda_1}},
$$
which, by dividing the left-hand side by the right-hand side of the inequality, rearranges to
$$
\left[\dfrac{\frac{\lambda_{1}-\gamma_{1}}{\lambda_{1}}(\frac{1}{z_{1}})^{\frac{\gamma_1}{\lambda_1}}+\frac{\gamma_{1}}{\lambda_{1}}\frac{1}{\nu_{1}^{*}}(\frac{1}{z_{1}})^{\frac{\gamma_1}{\lambda_1}-1}}{1+\frac{1-\nu_{1}^{*}}{\nu_{1}^{*}}\frac{\gamma_{1}}{\lambda_{1}}(\mu_{2}^{*})^{\frac{\lambda_{1}-\gamma_{1}}{\lambda_{2}}}}\right]^{\frac{\lambda_{2}}{\gamma_{1}-\lambda_{1}}}>1.
$$
Since $x^a>1$ for $x>0$ if and only if $x>1$ and $a>0$ or $x<1$ and $a<0$, which simplifies to $a(x-1)>0$, the inequality above is equivalent to
$$
(\gamma_1-\lambda_1)\left[\frac{\lambda_1-\gamma_1}{\lambda_1}(z_1)^{-\frac{\gamma_1}{\lambda_1}}+\frac{\gamma_1}{\lambda_1}\frac{1}{\nu_1^*}(z_1)^{-\frac{\gamma_1-\lambda_1}{\lambda_1}}-1-\frac{1-\nu_1^*}{\nu_1^*}\frac{\gamma_1}{\lambda_1}(\mu_2^*)^{\frac{\lambda_1-\gamma_1}{\lambda_2}}\right]>0.
$$
Consider the left-hand side of this inequality. Its derivative with respect to $z_1$ can be simplified to
$$
\frac{(\lambda_1-\gamma_1)^2}{\lambda_1}\frac{\gamma_1}{\lambda_1}z_1^{-\frac{\gamma_1}{\lambda_1}-1}\left(1-\frac{z_1}{\nu_1^*}\right).
$$
Therefore, the left-hand side is increasing when $z_1<\nu_1^*$ and decreasing when $z_1>\nu_1^*$, and reaches the maximum when $z_1=\nu_1^*$. Therefore, the inequality will hold for a range of $z_1$ around $\nu_1^*$ if and only if it holds for $z_1=\nu_1^*$. That is,
$$
(\gamma_1-\lambda_1)\left[\frac{\lambda_1-\gamma_1}{\lambda_1}(\nu_1^*)^{-\frac{\gamma_1}{\lambda_1}}+\frac{\gamma_1}{\lambda_1}\frac{1}{\nu_1^*}(\nu_1^*)^{-\frac{\gamma_1-\lambda_1}{\lambda_1}}-1-\frac{1-\nu_1^*}{\nu_1^*}\frac{\gamma_1}{\lambda_1}(\mu_2^*)^{\frac{\lambda_1-\gamma_1}{\lambda_2}}\right]>0,
$$
which is simplified to
$$
(\gamma_1-\lambda_1)\left[(\nu_1^*)^{-\frac{\gamma_1}{\lambda_1}}-1-\frac{1-\nu_1^*}{\nu_1^*}\frac{\gamma_1}{\lambda_1}(\mu_2^*)^{\frac{\lambda_1-\gamma_1}{\lambda_2}}\right]>0,
$$
and rearranged as
$$
(\lambda_1-\gamma_1)\left[\frac{1-(\nu_1^*)^{-\frac{\gamma_1}{\lambda_1}}}{1-\nu_1^*}\nu_1^*+\frac{\gamma_1}{\lambda_1}(\mu_2^*)^{\frac{\lambda_1-\gamma_1}{\lambda_2}}\right]>0.
$$
\end{proof}

\subsection{Proof of Proposition \ref{prop:11limit}}
\begin{proof}[{\bf Proof of Proposition \ref{prop:11limit}}]
We now consider a sequence of games in which all parameters of the game are fixed but the initial probabilities of commitment types, $\{z_1^n,z_2^n\}_n$, satisfy that $\lim \frac{z_1^n}{z_2^n}\in (0,\infty)$ and $\lim z_1^n=\lim z_2^n=0$. Recall the reputation coevolution curve,
 $$
\widetilde\mu_{1}(\mu_{2}|\gamma_1)=\begin{cases}
\frac{\lambda_{1}-\gamma_{1}}{\lambda_{1}(\mu_{2})^{\frac{\gamma_{1}-\lambda_{1}}{\lambda_{2}}}-\gamma_{1}} & \text{if }\mu_{2}^{*}<\mu_{2}\leq 1,\\
\frac{\lambda_{1}-\gamma_{1}}{\lambda_{1}(\mu_{2})^{\frac{\gamma_{1}-\lambda_{1}}{\lambda_{2}}}+(\frac{\gamma_{1}}{\nu_{1}^*}-\gamma_{1})(\frac{\mu_{2}}{\mu_{2}^{*}})^{\frac{\gamma_{1}-\lambda_{1}}{\lambda_{2}}}-\frac{\gamma_{1}}{\nu_{1}^*}} & \text{if }0<\mu_{2}\leq \mu_{2}^{*}.
\end{cases}
$$

\noindent (i) If $\lambda_1<\gamma_1$, then
\[
\lim_{\mu_2\to 0^+} \widetilde\mu_{1}(\mu_{2}|\gamma_1)=\nu_1^*(\gamma_1-\lambda_1)/\gamma_1=[1-k_2/(1-w)](1-\lambda_1/\gamma_1)>0.
\]
Therefore, in this case, along the equilibrium sequence of the sequence of games with vanishing probability of commitment types, player 1 concedes at time 0 with a probability converging to 1 (since otherwise after time 0, the reputations would not land on the reputation coevolution diagram). Hence, we obtain efficiency in this case, where players agree on player 2's terms right away, i.e., player 2 is the ``winner.''

\noindent (ii) If $\lambda_1=\gamma_1$, the expression of $\widetilde\mu_1(\mu_2|\gamma_1\neq\lambda_1)$ becomes
$$
\widetilde\mu_1(\mu_2|\gamma_1)=\begin{cases}
\frac{1}{-\frac{\gamma_1}{\lambda_2}\log(\mu_2)+1} & \text{if }\mu_2^*< \mu_2<1,\\
\frac{1}{-\frac{\gamma_1}{\nu_1^*}\frac{1}{\lambda_2}\log\left(\frac{\mu_2}{\mu_2^*}\right)+\mu_1^N} & \text{if }0<\mu_2\leq \mu_2^*,
\end{cases}
$$
where in this case $\mu_1^N=1/\left[-\frac{\gamma_1}{\lambda_2}\log(\mu_2^*)+1\right]$. Hence,
\begin{align*}
&\lim_{\mu_{2}\rightarrow0}\widetilde{\mu}_{1}'(\mu_{2}|\gamma_{1}\neq\lambda_{1})\\
=& \lim_{\mu_{2}\rightarrow0} \frac{\frac{\gamma_{1}}{\nu_{1}^{*}}\frac{1}{\lambda_{2}}\frac{1}{\mu_{2}}}{\left[-\frac{\gamma_{1}}{\nu_{1}^{*}}\frac{1}{\lambda_{2}}\log\left(\frac{\mu_{2}}{\mu_{2}^{*}}\right)+\mu_{1}^{N}\right]^{2}}\\
=& \lim_{\mu_{2}\rightarrow0} \frac{-\frac{\gamma_{1}}{\nu_{1}^{*}}\frac{1}{\lambda_{2}}\frac{1}{\mu_{2}^{2}}}{-2\left[-\frac{\gamma_{1}}{\nu_{1}^{*}}\frac{1}{\lambda_{2}}\log\left(\frac{\mu_{2}}{\mu_{2}^{*}}\right)+\mu_{1}^{N}\right]\frac{\gamma_{1}}{\nu_{1}^{*}}\frac{1}{\lambda_{2}}\frac{1}{\mu_{2}}}=\lim_{\mu_{2}\rightarrow0}\frac{\frac{1}{\mu_{2}}}{2\left[-\frac{\gamma_{1}}{\nu_{1}^{*}}\frac{1}{\lambda_{2}}\log\left(\frac{\mu_{2}}{\mu_{2}^{*}}\right)+\mu_{1}^{N}\right]}\\
=& \lim_{\mu_{2}\rightarrow0} \frac{-\frac{1}{\mu_{2}^{2}}}{-2\frac{\gamma_{1}}{\nu_{1}^{*}}\frac{1}{\lambda_{2}}\frac{1}{\mu_{2}}}=\lim_{\mu_{2}\rightarrow0}\frac{1}{2\frac{\gamma_{1}}{\nu_{1}^{*}}\frac{1}{\lambda_{2}}\mu_{2}}=\infty,
\end{align*}
where L'Hospital's rule is applied once on each line. Hence, player 2 will be the ``winner.''

\noindent (iii) If $\lambda_1 >\gamma_1$, then 
\[
  \lim_{\mu_2\to 0^+} \widetilde\mu_{1}(\mu_{2}|\gamma_1)=0.
\]
If $\lambda_1>\gamma_1+\lambda_2$, then
\[
  \lim_{\mu_2\to 0^+} \widetilde\mu_{1}'(\mu_{2}|\gamma_1)=0,
\]
if $\lambda_1=\gamma_1+\lambda_2$, then
\[
  \lim_{\mu_2\to 0^+} \widetilde\mu_{1}'(\mu_{2}|\gamma_1)>0,
\]
and if $\lambda_1<\gamma_1+\lambda_2$, then
\[
  \lim_{\mu_2\to 0^+} \widetilde\mu_{1}'(\mu_{2}|\gamma_1)=\infty.
\]
The limits of $\widetilde\mu_{1}'(\mu_{2}|\gamma_1)$ above can be derived from the expression of $\widetilde\mu_1(\mu_2|\gamma_1)$ for $\mu_2\leq \mu_2^*$, which can be rearranged as
$$
\widetilde \mu_1(\mu_2|\gamma_1)=\frac{(\lambda_1-\gamma_1)(\mu_2)^{\frac{\lambda_1-\gamma_1}{\lambda_2}}}{\lambda_1+\gamma_1\frac{1-\nu_1^*}{\nu_1^*}(\mu_2^*)^{\frac{\lambda_1-\gamma_1}{\lambda_2}}-\frac{\gamma_1}{\nu_1^*}(\mu_2)^{\frac{\lambda_ 1-\gamma_1}{\lambda_2}}}.
$$
The derivative is
$$
\widetilde \mu_1'(\mu_2|\gamma_1)=(\mu_2)^{\frac{\lambda_1-\gamma_1-\lambda_2}{\lambda_2}} \frac{\left[\lambda_1+\gamma_1\frac{1-\nu_1^*}{\nu_1^*}(\mu_2^*)^{\frac{\lambda_1-\gamma_1}{\lambda_2}}\right](\lambda_1-\gamma_1)}{\left[\lambda_1+\gamma_1\frac{1-\nu_1^*}{\nu_1^*}(\mu_2^*)^{\frac{\lambda_1-\gamma_1}{\lambda_2}}-\frac{\gamma_1}{\nu_1^*}(\mu_2)^{\frac{\lambda_ 1-\gamma_1}{\lambda_2}}\right]^2},
$$
which in the limit is
$$
\lim_{\mu_2\to 0^+} \widetilde \mu_1'(\mu_2|\gamma_1)=\lim_{\mu_2\rightarrow 0^+}(\mu_2)^{\frac{\lambda_1-\gamma_1-\lambda_2}{\lambda_2}}\frac{\lambda_1-\gamma_1}{\lambda_1+\gamma_1\frac{1-\nu_1^*}{\nu_1^*}(\mu_2^*)^{\frac{\lambda_1-\gamma_1}{\lambda_2}}}.
$$
The ``winner'' is player 1 (resp., player 2) if $\lambda_1>$ (resp., $<$) $\gamma_1+\lambda_2$, so there is efficiency.
\end{proof}

\subsection{Proof of Proposition \ref{prop:1limit}}
\begin{proof}[\bf Proof of Proposition \ref{prop:1limit}]

Our result does not depend on the initial order of moves of the players in their demand choice. We will perform the analysis for the case in which player 1 first picks a demand, and then player 2, observing this, chooses her demand, and then the war of attrition starts. Let $\sigma_1^n(i)$ be the equilibrium probability that player 1 chooses type $i/K$ in the $n^{\text{th}}$ game, and let $\sigma_2^n(j|i)$ be the equilibrium probability that player 2 chooses type $j/K$ after observing that player 1 chooses $i/K$ in the $n^{\text{th}}$ game. Let $\left(\sigma_1, \{\sigma_2(\cdot|i)\}_{i\in\left\{2,...,K-1\right\}}\right)$ be the limits of these strategies (along a convergent subsequence).

The first case is $\gamma_1\leq r_1$. In this case, if player 1 chooses
$$a_1=\max \left\{ a\in A^K \left| a \leq \frac{r_2}{r_1+r_2}\right. \right\},$$ then for any incompatible demand of player 2, $\lambda_1=\frac{r_2(1-a_1)}{a_1+a_2-1}$ is decreasing in $a_2$, so it is minimized at $a_2=(K-1)/K$. In that case, $\lambda_1>\gamma_1$. Hence, when player 2 makes an incompatible demand, either $\sigma_2(\cdot|a_1)=0$ or $\sigma_1(a_1)=0$, and player 1 is the winner, or the winner is determined by the comparison of $\lambda_1-\gamma_1$ versus $\lambda_2$.
\begin{eqnarray}
&&\lambda_1-\gamma_1 >\lambda_2 \nonumber \\
&\iff&
r_2(1-a_1)-\gamma_1(a_1+a_2-1)>r_1(1-a_2) \nonumber \\
&\iff& r_2(1-a_1)-\gamma_1a_1>(1-a_2)(r_1-\gamma_1).\label{eqn:condition}
\end{eqnarray}
It is then routine to verify that if $a_1=\max \left\{ a\in A^K \left| a \leq \frac{r_2}{r_1+r_2}\right. \right\},$ and if $a_2>1-a_1$, player 1 is the winner.

Turning to player 2 in this case,  for any $a_1>\frac{r_2}{r_1+r_2}$ such that $\sigma_1(a_1)>0$, player 2 is the winner if she demands $\max \left\{ a\in A^K \left| a \leq \frac{r_1}{r_1+r_2}\right. \right\}$. This is again routine to verify. This completes the proof for $r_1\geq \gamma_1$.

The second case is $\gamma_1>r_1$. In this case, if player 1 chooses $$\max \left\{ a\in A^K \left| a \leq \frac{r_2}{\gamma_1+r_2}\right. \right\},$$ then for any incompatible demand of player 2, $\lambda_1>\gamma_1$. This is because $\lambda_1$ is decreasing in player 2's demand, $a_2$, and when $a_2<1$ and when player 1's demand is not more than $\frac{r_2}{\gamma_1+r_2}$, $\lambda_1>\gamma_1$. Moreover, the right-hand side of Equation \eqref{eqn:condition}, $(1-a_2)(r_1-\gamma_1)<0$, and the left-hand side, $r_2(1-a_1)-\gamma_1a_1\geq 0$. Hence, whenever player 2 chooses an incompatible demand $a_2$ with $\sigma_2(a_2|a_1)>0$, player 1 is the winner. Hence, player 1 secures the payoff of  $\frac{r_2}{\gamma_1+r_2}-1/K$.

Turning to player 2 in this case, consider the strategy for player 2 of always choosing $a_2=(K-1)/K$. When player 1's demand, $a_1$, is less than $\frac{r_2}{r_2+\gamma_1}+1/K$, player 2's payoff is at least $1-a_1$, and our claim is true. If  $a_1\geq \frac{r_2}{r_2+\gamma_1}+1/K$, and if $\sigma_1(a_1)>0$, then 
\[
\lambda_1=\frac{(1-a_1)r_2}{a_1+a_2-1}= \frac{(1-a_1)r_2}{a_1-1/K}<\gamma_1,
\]
which implies that player 2 is the winner. Hence, player 2 secures the payoff of $\frac{\gamma_1}{\gamma_1+r_2}-1/K$.
\end{proof}

\subsection{Two-sided ultimatum and single demand types}
\subsubsection{Formal description of the game}
Let us formally describe the strategies and payoffs of the (unjustified) players. Let $\Sigma_i=(F_i,G_i,q_i)$ denote an unjustified player $i$'s strategy, where $F_i(t)$ is player $i$'s probability of conceding by time $t$, $G_i(t)$ is player $i$'s probability of challenging by time $t$, and $q_i(t)$ is player $i$'s probability of conceding to a challenge at time $t$. Restrict $F_i$ and $G_i$ to be right-continuous and increasing functions with $F_i(t)+G_i(t)\le 1$ for every $t\ge 0$, and $q_i(t)\in [0,1]$ to be a measurable function. For $i=1,2$, 
player $i$'s time-zero expected utility of conceding at time $t$ is
\begin{eqnarray}\label{eq:2uit}
U_i(t,q_i,\Sigma_j)&=&W_i(t, q_i,\Sigma_j)+e^{-r_it}(1-a_j)\Big[1-(1-z_j)F_j(t)-(1-z_j)G_j(t)-z_j(1-e^{-\gamma_j t})\Big]\nonumber\\
&&\quad\quad+(1-z_j)\Big[F_j(t)-F_j(t^-)\Big]\frac{a_i+1-a_j}{2},
\end{eqnarray}
where
\begin{eqnarray*}
W_i(t, q_i,\Sigma_j)&=&(1-z_j) \int_0^t a_i e^{-r_i s}dF_j(s)+z_j\int_0^t \Big\{1-a_j-\Big[1-q_i(s)\Big]k_iD\Big\} e^{-r_i s}\gamma_j e^{-\gamma_j s}ds\\
&&\quad\quad+(1-z_j)\int_0^t \Big\{1-a_j+\Big[1-q_i(s)\Big]\Big[(1-w_j)D-k_iD\Big]\Big\} e^{-r_js}dG_i(s),
\end{eqnarray*}
and it is assumed that players equally divide their surplus if they concede simultaneously, which happens with probability zero in equilibrium. Player $i$'s time-zero expected utility of challenging at time $t$ is
\begin{eqnarray*}
&&V_i(t, q_i,\Sigma_j)=\\
&&W_i(t, q_i,\Sigma_j)+(1-z_j)[1-F_j(t)-G_j(t^-)]e^{-r_it}[(1-q_j(t))w_i+q_j(t)]D\\
&& +\left[1-(1-z_j)F_j(t)-(1-z_j)G_j(t)-z_j\left(1-e^{-\gamma_jt}\right)\right]e^{-r_it}(1-a_j-c_iD)+(1-z_j)\times \\
&& [G_j(t)-G_j(t^-)]\left\{1-a_j+\frac{1}{2}[(1-q_i(s))(1-w_j)-k_i]D+\frac{1}{2}[(1-q_j(t))w_i+q_j(t)]D\right\}
\end{eqnarray*}
where it is assumed that players resolve the dispute in court and players are equally likely to be the challenger if they challenge simultaneously at time $t$, which happens with probability zero in equilibrium, and 
Player $i$'s expected utility from strategy $\Sigma_i$ is
$$
u_i(\Sigma_i,\Sigma_j)=\int_0^\infty U_i(s, q_i,\Sigma_j) dF_i(s) + \int_0^\infty V_i(s,q_i,\Sigma_j) dG_i(s).
$$

We again study the Bayesian Nash equilibria of this game. Let $\mu_i(t)$ denote the posterior belief (of player $j \neq i$) that player $i$ is justified conditional on the game not ending by game time $t$. By Bayes' rule,
\[\mu_i(t):=\frac{z_i\left[1-\int_0^t \gamma_i e^{-\gamma_i s}ds\right]}{z_i\left[1-\int_0^t \gamma_i e^{-\gamma_i s}ds\right]+(1-z_i)\left[1-F_i(t^-)-G_i(t^-)\right]}.\]
Let $\nu_i(t)$ denote the posterior belief that player $i$ is justified if player $i$ challenges at time $t$. If $G_i$ has an atom at $t$, then $\nu_i(t)=0$. If $G_i$ is differentiable at $t$, then
\[
\nu_i(t)=\frac{\mu_i(t)\gamma_i}{\mu_i(t)\gamma_i+[1-\mu_i(t)]\chi_i(t)},
\]
where $\chi_i(t)$ is the hazard rate of challenging for an unjustified player $i$,
\[
\chi_i(t)=\frac{G_i'(t)}{1-F_i(t^-)-G_i(t^-)}.
\]

\subsubsection{Equilibrium strategies and reputations in games with single demand types and slow ultimatum opportunity arrival for at least one player}\label{sec:twosidedslow}

\begin{proof}[{\bf Proof of Theorem \ref{T:existenceandproperties2}}]
All the properties in the equilibrium characterization in the setting with one-sided ultimatum are satisfied. Therefore, we can derive the equilibrium strategies and reputations as follows.

\noindent{\bfseries Players' conceding strategies.} In equilibrium, players concede at the same rates as in {\AG}. Players are indifferent between conceding and waiting to concede the next instant. An unjustified player concedes at a rate $\kappa_i=\lambda_i/(1-\mu_i)$ to make the opposing unjustified player indifferent between conceding and not conceding, where $\lambda_i=r_j(1-a_i)/D.$
\noindent{\bfseries Player $i$'s optimal yielding strategy.} An unjustified player $i$ is indifferent between responding and yielding when player $j\neq i$ is believed to be justified with probability $\nu_j=1-{k_i}/{(1-w)}=:\nu_j^*$, strictly prefers to respond when $\nu_j<\nu_j^*$, and strictly prefers to yield when $\nu_j>\nu_j^*$.

\noindent{\bfseries Player $i$'s optimal challenging strategy.} We consider the optimal challenging strategy of an unjustified player $i$ who believes that player $j\neq i$ is justified with probability $\mu_j$ and an unjustified player $j$ yields to a challenge with probability $q_j$. An unjustified player $i$ is indifferent between challenging and not challenging if $\mu_j=1-c_i/[q_j+(1-q_j)w]$. In particular, an unjustified player $i$ strictly prefers not to challenge when $\mu_j<1-c_i=:\mu_j^*$.

\noindent{\bfseries Candidate equilibrium challenging and yielding strategies.} If player $j$ is justified with a probability more than $\mu_j^*$, an unjustified player $i$ strictly prefers not to challenge. If player $j$ is justified with a probability less than $\mu_j^*$, an unjustified player $i$ must challenge at rate $\chi_j$ to make player $i$ believe that a challenging player $i$ is justified with probability $\nu_i^*:=1-{k_j}/{(1-w_i)}$:
$$
\frac{\mu_i\gamma_i}{\mu_i\gamma_i+(1-\mu_i)\chi_i}=\nu_i^*\Longrightarrow \chi_i(\mu_i)=\frac{1-\nu_i^*}{\nu_i^*}\frac{\mu_i}{1-\mu_i}\gamma_i.
$$
If an unjustified player $i$ challenges at a rate higher than the specified rate, then an unjustified player $j$ is strictly better off responding than yielding to the challenge. If an unjustified player $i$ challenges at a rate lower than the specified rate, then an unjustified player $2$ is strictly worse off responding than yielding to the challenge. On the other hand, to make player $i$ indifferent between challenging and not challenging, player $j$ yields to a challenge with probability
$$
q_j(\mu_j)=\frac{1}{1-w_i}\left(\frac{k_i}{1-\mu_j}-w_i\right).
$$

\noindent{\bfseries Reputation in the challenge phase.}
When an unjustified player $i$ challenges, player $i$'s reputation follows the following Bernoulli differential equation:
$$
\mu_i'(t)=(\lambda_i-\gamma_i)\mu_i(t)+\frac{\gamma_i}{\nu_i^*}\mu_i^2(t).
$$

\noindent{\bfseries Reputation in the no-challenge phase.}
When an unjustified player $i$ does not challenge, player $i$'s reputation follows the following Bernoulli differential equation:
$$
\mu_i'(t)=(\lambda_i-\gamma_i)\mu_i(t)+\gamma_i\mu_i^2(t).
$$

\noindent{\bfseries Finite time.} If $\lambda_i\geq \gamma_i$ for some $i=1,2$, then $\mu_i'(t)\geq \gamma_iz_i^2$ for all $\mu_i(t)\geq z_i$. Hence, $\tau<\infty$.

According to the differential equations characterizing the players' reputations, a reputation coevolution diagram can be uniquely drawn backwards from $(1,1)$, the pair of terminal reputations. Hence, the strategies are uniquely pinned down as claimed.
\end{proof}

\subsubsection{Equilibrium reputations and strategies in games with single demand types and fast ultimatum opportunity arrival for both players}\label{sec:twosidedfast}

\begin{figure}[t!]
\includegraphics[width=\textwidth]{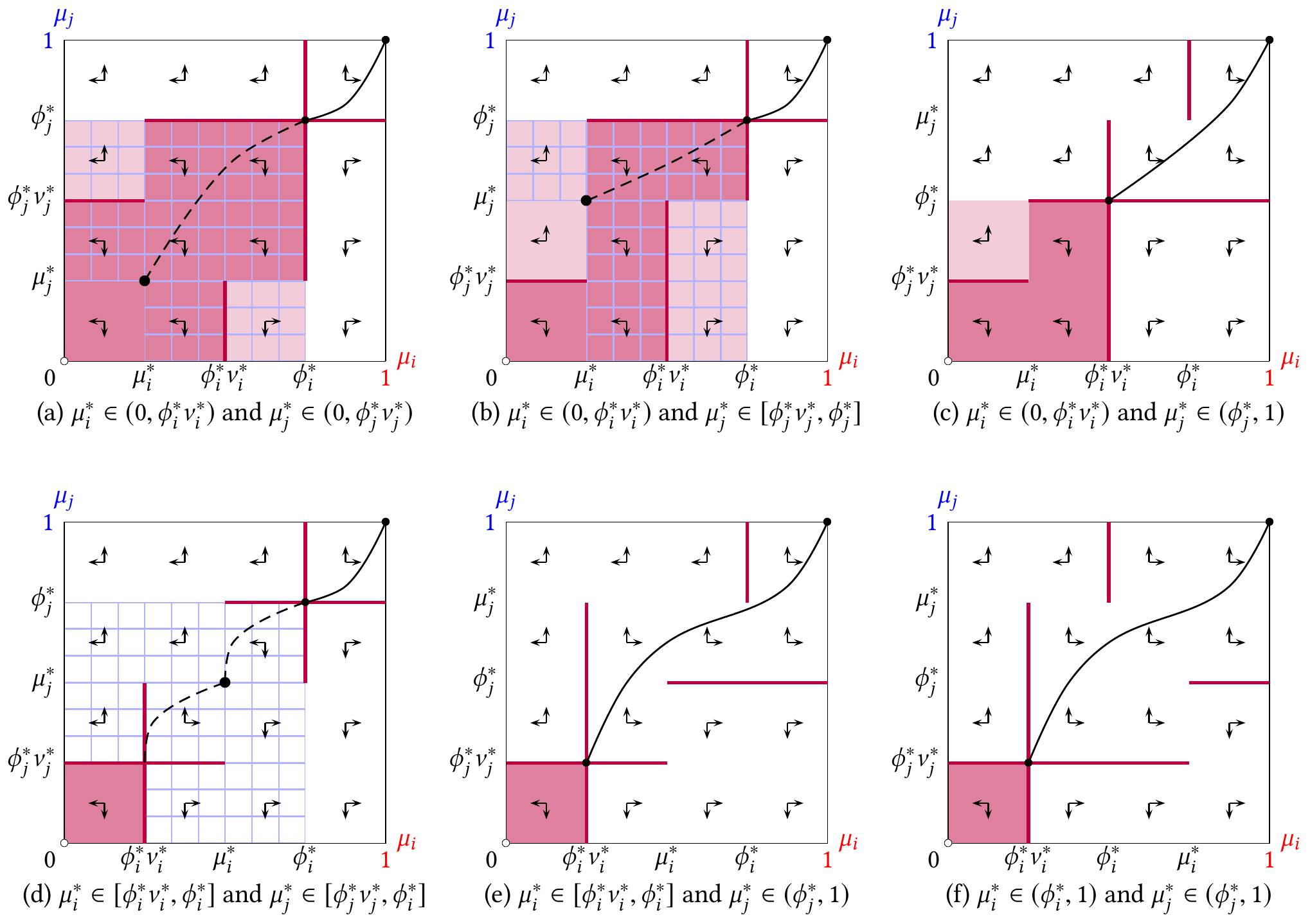}
\caption{Illustration of characterization of equilibrium reputations in bargaining games with single demand types and high ultimatum opportunity arrival rates $\gamma_1>\lambda_1$ and $\gamma_2>\lambda_2$.}\label{fig:fulltwosided}

{\footnotesize

The six figures, with appropriate labeling of $i$ and $j$ as $1$ and $2$, cover all possible settings with (1) $\mu_1^*\in (0,\phi_1^*\nu_1^*)$, $\mu_1^*\in [\phi_1^*\nu_1^*,\phi_1^*]$, or $\mu_1^*\in (\phi_1^*,1)$, and (2) $\mu_2^*\in (0,\phi_2^*\nu_2^*)$, $\mu_2^*\in [\phi_2^*\nu_2^*,\phi_2^*]$, or $\mu_2^*\in (\phi_2^*,1)$. Each of the six figures contains a reputation plane with player $i$'s reputation on the x-axis and player $j$'s reputation on the y-axis. Each reputation plane is divided into sixteen regions with $\mu_k^*$, $\phi_k^*\nu_k^*$, and $\phi_k^*$, $k=i,j$, as dividing lines. {\bf (Finite-$T$ equilibrium)} If the initial reputation vector lies in the white region and its boundary, there is a unique equilibrium, which is a finite-$T$ equilibrium with players' reputations coevolving on the solid line to $(1,1)$ after at most one player concedes at time zero. {\bf (Type-1 infinite-$T$ equilibrium)} If the initial reputation vector lies in the (lighter and darker) purple region and its boundary, there are many infinite-$T$ equilibria for each of which at most one player concedes at time zero, the reputation vector after initial concession lies in the darker purple region and any darker purple lines on the boundary of the region, and players' reputations evolve to but never reach $(0,0)$, or $(0,\omega)$ if $(0,\omega)$ lies on a darker purple line. {\bf (Type-2 infinite-$T$ equilibrium)} If the initial reputation vector lies in the crosshatched region excluding its boundary, there is an infinite-$T$ equilibrium in which at most one player concedes at time zero, players' reputations after time zero coevolve on the dashed line to $(\mu_i^*, \mu_j^*)$.
}
\end{figure}

The six reputation planes in Figure \ref{fig:fulltwosided}, with appropriate labeling of $i$ and $j$ as $1$ and $2$, cover all possible settings with (1) $\mu_1^*\in (0,\phi_1^*\nu_1^*)$, $\mu_1^*\in [\phi_1^*\nu_1^*,\phi_1^*]$, or $\mu_1^*\in (\phi_1^*,1)$, and (2) $\mu_2^*\in (0,\phi_2^*\nu_2^*)$, $\mu_2^*\in [\phi_2^*\nu_2^*,\phi_2^*]$, or $\mu_2^*\in (\phi_2^*,1)$. Each reputation plane has player $i$'s reputation on the x-axis and player $j$'s reputation on the y-axis, and is divided into sixteen regions with $\mu_k^*$, $\phi_k^*\nu_k^*$, and $\phi_k^*$, $k=i,j$, as dividing lines. Fix any of the sixteen regions. For any initial reputation vector in the region or on its boundary, if neither player concedes at time zero and both players follow the strategies specified above, the horizontal arrow in the region represents the direction of player $i$'s reputation building, and the vertical arrow represents the direction of player $j$'s reputation building, with the direction strict on the dividing lines unless the dividing line is darker purple.

Equilibrium reputations must eventually reach $(1,1)$ in a finite-$T$ equilibrium, or approach $(0,0)$, approach $(0,\omega)$ if $(0,\omega)$ is on the purple line, or reach and stay at $(\mu_i^*,\mu_j^*)$ in infinite-$T$ equilibria. Using the directions of reputation building after initial concessions, if players follow specified equilibrium strategies, we can derive contradictions with the eventual reputation vector for any initial concession that is not part of any equilibrium. Hence, the directions of reputation building restrict candidate equilibrium reputation vectors immediately after initial concessions, and consequently initial concessions in equilibrium. The set of reputation vectors that can be equilibrium reputation vectors immediately after initial concessions is represented by a solid line, a darker purple region (and selective darker purple lines on its boundary), and a dashed line. Specifically, the solid line represents the collection of reputation vectors immediately after initial concessions that situates on the path to eventually reach $(1,1)$ if players follow specified post-concession strategies, a darker purple region and its selected darker purple lines on its boundary represent the collection of reputation vectors that can be supported as equilibrium reputation vector immediately after initial concessions, and the dashed line, if it exists in a figure, collects the reputation vector that situates on a reputation coevolution curve that eventually reaches\threeemdashes{}increases or decreases to\threeemdashes{}$(\mu_i^*,\mu_j^*)$ if players follow specified post-initial-concession strategies.

The equilibrium post-initial-concession strategies must be consistent with the reputation building specified by the solid line, the darker purple region and its appropriate boundary, and the dashed line, and the equilibrium initial concession. The initial concession by one player is part of an equilibrium as long as the posterior reputations after the initial concession lies on the lines or the darker purple region.

%howmanypages

\newpage
\setcounter{page}{1}\renewcommand\thepage{O\arabic{page}}
\setcounter{figure}{0}\renewcommand\thefigure{O\arabic{figure}}
\setcounter{table}{0}\renewcommand{\thetable}{O\arabic{table}}

\section*{Online appendices (not for publication)}
\section{Additional omitted details}\label{sec:details}
We present more detailed proofs and derivations regarding (i) equilibrium strategies and reputations with one-sided ultimatum, (ii) comparative statics, (iii) equilibrium existence and uniqueness with multiple demand types, and (iv) equilibrium strategies and reputations with two-sided ultimatum.

\subsection{One-sided ultimatum and single demand types}

\subsubsection{Bernoulli differential equations}\label{sec:Bernoulli}

\begin{lemma}\label{lem:ODE}
The solution to the Bernoulli differential equation $\mu'(t)=A\mu(t)+B\mu^2(t)$ given $\mu(0)=\mu^0$ is
$$
\mu(t; \mu^0, A, B)=
\begin{cases}
1\bigg/\left[\left(\frac{1}{\mu^0}+\frac{B}{A}\right)\exp(-At)-\frac{B}{A}\right] &
\text{ if } A\neq 0,\\
1\bigg/\left[-Bt+\frac{1}{\mu^0}\right] & \text{ if } A=0.
\end{cases}
$$
If $\mu^0>-A/B$, then $\mu'(t)>0$ for $t\geq t^0$, and the time length it takes to reach reputation $\mu$ from $\mu^0$ is
$$
t(\mu;\mu^0,A,B)=\frac{1}{A}\ln\left(\frac{{\frac{1}{\mu^0}+\frac{B}{A}}}{{\frac{1}{\mu}+\frac{B}{A}}}\right).
$$
\end{lemma}

\subsubsection{Equilibrium strategies, reputations, and payoffs}\label{sec:explicitresuls}

\begin{theorem}\label{prop:eq1}\label{thm:1single}
Consider a bargaining game $B=(a_1,a_2,z_1,z_2,r_1,r_2,\gamma_1,c_1,k_2,w_1)$ with one-sided ultimatum and single demand types. Equilibrium strategies and reputations $(\widehat F_1,\widehat G_1, \widehat F_2, \widehat q_2, \widehat\mu_1, \widehat\mu_2)$ satisfy
$$\widehat f_i(t)=\exp\left[{-\int_0^t \widehat\kappa_i(s)ds}\right]\widehat\kappa_i(t), \text{ where } \widehat\kappa_i(s)=1_{s<T}\frac{\lambda_i}{1-\widehat\mu_i(s)};$$
$$\widehat g_1(t)=\exp\left[-\int_0^t\widehat\chi_1(s)ds\right]\widehat\chi_1(t), \text{ where } \widehat \chi_1(s)=1_{s<T-t_2^N}\frac{1-\nu_1^*}{\nu_1^*}\frac{\widehat\mu_1(s)}{1-\widehat\mu_1(s)}\gamma_1;$$
$$\widehat q_2(t)=1_{t<T-t_2^N} \frac{1}{1-w}\left[\frac{c_1}{1-\widehat\mu_2(t)}-w\right];$$
$$\widehat\mu_i(T-t)=\check\mu_i(-t),$$
where $$
\check\mu_1(-t)=\begin{cases}
\mu(-t;1,\lambda_1-\gamma_1,\frac{\gamma_1}{\nu_1^*}) & \text{if }t< T-T_1,\\
\mu(t_2^N-t;\mu_1^N,\lambda_1-\gamma_1,{\gamma_1}) & \text{if }t\geq T-T_1,
\end{cases}
$$
$$\check\mu_2(-t)=\mu(-t;1,\lambda_2,0),$$ $T_i$ solves $\check\mu_i(-T_i)=z_i$, and $T=\min\{T_1,T_2\}$. Player $i$'s equilibrium payoff is
$$
\widehat u_i=1-a_j+1_{z_i\geq \widetilde\mu_i(z_j)}\left[1-\frac{z_j}{1-z_j}\bigg/\frac{\widetilde\mu_j(z_i)}{1-\widetilde\mu_j(z_i)}\right]D.$$
\end{theorem}

\subsubsection{Reputation coevolution curves}\label{sec:curves}

When player 2's reputation is $\mu_2^*$, player 1's reputation is
$$
\mu_1^N:=\frac{\lambda_1-\gamma_1}{\lambda_1(\mu_2^*)^{\frac{\gamma_1-\lambda_1}{\lambda_2}}-\gamma_1}.
$$

The reputation coevolution curve can be represented by
$$
\widetilde\mu_1(\mu_2)=\begin{cases}
\frac{1}{-\frac{\gamma_1}{\lambda_2}\log(\mu_2)+1} & \text{if }\mu_2^*< \mu_2\leq 1,\\
\frac{1}{-\frac{\gamma_1}{\nu_1^*}\frac{1}{\lambda_2}\log\left(\frac{\mu_2}{\mu_2^*}\right)+\frac{1}{1-\frac{\gamma_1}{\lambda_2}\log(\mu_2^*)}} & \text{if }0<\mu_2\leq \mu_2^*,
\end{cases}
$$
when $\gamma_1=\lambda_1$.
Equivalently, the curve is represented by the inverse
$$
\widetilde\mu_2(\mu_1)=
\begin{cases}
\left[\left(1-\frac{\gamma_1}{\lambda_1}\right)\frac{1}{\mu_1}+\frac{\gamma_1}{\lambda_1}\right]^{\frac{\lambda_2}{\gamma_1-\lambda_1}} & \text{if } \mu_1^N<\mu_1\leq 1, \\
\left[\dfrac
{
\frac{\lambda_1-\gamma_1}{\lambda_1}\frac{1}{\mu_1}+\frac{\gamma_1}{\lambda_1}\frac{1}{\nu_1^*}
}
{
1+\frac{1-\nu_1^*}{\nu_1^*}\frac{\gamma_1}{\lambda_1}(\mu_2^*)^{\frac{\lambda_1-\gamma_1}{\lambda_2}}
}
\right]^{\frac{\lambda_2}{\gamma_1-\lambda_1}} & \text{if } \max\left\{0,\left(1-\frac{\lambda_1}{\gamma_1}\right)\nu_1^*\right\}<\mu_1\leq \mu_1^N.
\end{cases}
$$
 and
 $$
 \widetilde\mu_2(\mu_1)=
 \begin{cases}
 \exp\left[\frac{\lambda_2}{\gamma_1}\left(1-\frac{1}{\mu_1}\right)\right] & \text{if }\mu_1^N< \mu_1\leq 1,\\
 \mu_{2}^{*}\exp\left[\frac{\frac{1}{1-\frac{\gamma_{1}}{\lambda_{2}}\log\left(\mu_{2}^{*}\right)}-\frac{1}{\mu_{1}}}{\frac{\gamma_{1}}{\nu_{1}^{*}}\frac{1}{\lambda_{2}}}\right] & \text{if }\max\left\{0,1-\frac{\lambda_1}{\gamma_1}\right\}<\mu_1\leq \mu_1^N.
 \end{cases}
 $$

\subsection{Comparative statics with one-sided ultimatum and single demands}
\subsubsection{Proof of Proposition \ref{prop:compstat1}}
\begin{proof}[{\bf (i) Effects of $z_i$}]
Player 1's payoff can be rearranged as $$u_1=1-a_2+1_{z_1\ge\widetilde \mu_1 (z_2)}\left[1-\frac{z_2}{1-z_2}/\frac{\widetilde\mu_2(z_1)}{1-\widetilde\mu_2(z_1)}\right]D.$$
Player 2's payoff can be rearranged as $$u_2=1-a_1+1_{z_2\le\widetilde\mu_2(z_1)}\left[1-\frac{z_1}{1-z_1}/\frac{\widetilde\mu_1(z_2)}{1-\widetilde\mu_1(z_2)}\right]D.$$
Note that $z_i$ only influences the term that involves the indicator function and that $D$ does not depend on $z_i$. In particular, the indicator function $1_{z_i\geqslant \widetilde\mu_i(z_j)}$ is increasing in $z_i$ and decreasing in $z_j$, and the term enclosed in the square brackets is also increasing in $z_i$ and decreasing in $z_j$. Therefore, when $z_i$ increases, player $i$'s payoff strictly increases and player $j$'s payoff strictly decreases only when the condition of the indicator function is satisfied.
\end{proof}

\begin{proof}[{\bf (ii) Effects of $r_i$}]
(i) Consider $\partial u_1/\partial r_1$ first. The only term affected by $r_1$ is  $$-1_{z_1> \widetilde \mu_1(z_2)}\frac{z_1}{1-z_1}/\frac{\widetilde\mu_1(z_2)}{1-\widetilde\mu_1(z_2)},$$ whose derivative has the same sign as that of $-1_{z_1>\widetilde \mu_1(z_2)}/\widetilde \mu_2(z_1)$. The derivative of $-1/\widetilde\mu_2(z_1)$ is $\frac{1}{\widetilde \mu_2^2(z_1)}\frac{\partial \widetilde\mu_2(z_1)}{\partial r_1}$. Therefore, the sign of $\partial u_1/\partial r_1$ is the same as that of $\frac{\partial \widetilde\mu_2(z_1)}{\partial r_1}$, whenever $z_1\geqslant \widetilde \mu_1(z_2)$. In the expression of $\widetilde \mu_2$, $r_1$ only enters through the expression of $\lambda_2$, which is strictly increasing in $r_1$. Therefore, the sign of the expression is the same as $\frac{\partial \widetilde\mu_2(z_1)}{\partial \lambda_2}$. Mathematically,
$$
\frac{\partial u_1}{\partial r_1}=1_{z_1>\widetilde \mu_1(z_2)}\frac{z_1}{1-z_1}\frac{1}{\widetilde\mu_2^2(z_1)}\frac{\partial \widetilde\mu_2(z_1)}{\partial \lambda_2}\frac{\partial \lambda_2}{\partial r_1}.
$$
For $z_1\geqslant \mu_1^N$,
$$
\widetilde\mu_2(z_1)=\left[(1-\frac{\gamma_1}{\lambda_1})\frac{1}{z_1}+\frac{\gamma_1}{\lambda_1}\right]^{\frac{\lambda_2}{\gamma_1-\lambda_1}}.
$$
Its derivative with respect to $\lambda_2$ is
$$
\frac{\partial \widetilde\mu_2(z_1)}{\partial \lambda_2}=\widetilde\mu_2(z_1)\log\left[\widetilde\mu_2^{\frac{1}{\lambda_2}}(z_1)\right]=\widetilde\mu_2(z_1)\frac{1}{\lambda_2}\log\left[\widetilde\mu_2(z_1)\right]\leqslant 0,
$$
as $\widetilde\mu_2(z_1)\leqslant 1$. For $z_1<\mu_1^N$,
$$
\widetilde \mu_2(z_1)=\left[\dfrac{\frac{\lambda_{1}-\gamma_{1}}{\lambda_{1}}\frac{1}{\mu_{1}}+\frac{\gamma_{1}}{\lambda_{1}}\frac{1}{\nu_{1}^{*}}}{1+\frac{1-\nu_{1}^{*}}{\nu_{1}^{*}}\frac{\gamma_{1}}{\lambda_{1}}(\mu_{2}^{*})^{\frac{\lambda_{1}-\gamma_{1}}{\lambda_{2}}}}\right]^{\frac{\lambda_{2}}{\gamma_{1}-\lambda_{1}}}.
$$
Its partial derivative with respect to $\lambda_2$ is
\begin{eqnarray*}
\frac{\partial \widetilde\mu_2(z_1)}{\partial \lambda_2}&=&\widetilde \mu_2(z_1) \frac{1}{\lambda_2} \underline{\log [\widetilde\mu_2(z_1)]} \frac{1}{\gamma_1-\lambda_1} [\widetilde\mu_2(z_1)]^{\frac{\lambda_2}{\gamma_1-\lambda_1}[\frac{\lambda_2}{\gamma_1-\lambda_1}-1]}\underline{(-1)}\dfrac{\frac{\lambda_{1}-\gamma_{1}}{\lambda_{1}}\frac{1}{\mu_{1}}+\frac{\gamma_{1}}{\lambda_{1}}\frac{1}{\nu_{1}^{*}}}{\left[1+\frac{1-\nu_{1}^{*}}{\nu_{1}^{*}}\frac{\gamma_{1}}{\lambda_{1}}(\mu_{2}^{*})^{\frac{\lambda_{1}-\gamma_{1}}{\lambda_{2}}}\right]^2}\cdot\\
&&(\mu_2^*)^{\frac{\lambda_1-\gamma_1}{\lambda_2}}(\lambda_1-\gamma_1)\underline{\log(\mu_2^*)}\underline{(-\frac{1}{\lambda_2^2})}.
\end{eqnarray*}
The expression above is negative because the four underlined terms are negative, the terms $\frac{1}{\gamma_1-\lambda_1}$ and $\lambda_1-\gamma_1$ multiply to $-1$, and the other terms are positive.

\noindent (ii) Consider $\partial u_2/\partial r_1$ next. We have
$$
\frac{\partial u_2}{\partial r_1}=1_{z_2>\widetilde\mu_2(z_1)}\frac{z_1}{1-z_1}\frac{1}{\widetilde\mu_1^2(z_2)}(1-a_2)\frac{\partial\widetilde\mu_1(z_2)}{\partial\lambda_2}.
$$
It remains to show that $\partial \widetilde\mu_1(z_2)/\partial \lambda_2>0$.  It suffices to show that $\partial \log \widetilde\mu_1(z_2)/\partial \lambda_2>0$. For $z_2\geqslant \mu_2^*$,
$$
\log \widetilde \mu_1(z_2)=\log(\lambda_1-\gamma_1)-\log\left[\lambda_1(z_2)^\frac{\gamma_1-\lambda_1}{\lambda_2}-\gamma_1\right],
$$
and
\begin{eqnarray*}
\frac{\partial \log\widetilde\mu_1(z_2)}{\partial \lambda_2}&=&-\frac{\lambda_1(z_2)^{\frac{\gamma_1-\lambda_1}{\lambda_2}}(\gamma_1-\lambda_1)[\log(z_2)](-1/\lambda_2^2)}{\lambda_1(z_2)^{\frac{\gamma_1-\lambda_1}{\lambda_2}}-\gamma_1}\\
&=&\frac{1}{\lambda_1}(z_2)^{\frac{\gamma_1-\lambda_1}{\lambda_2}}[-\log(z_2)]\frac{\lambda_1-\gamma_1}{\lambda_1(z_2)^{\frac{\gamma_1-\lambda_1}{\lambda_2}}-\gamma_1}\\
&=&\frac{1}{\lambda_1}(z_2)^{\frac{\gamma_1-\lambda_1}{\lambda_2}}\left[\log\left(\frac{1}{z_2}\right)\right]\widetilde\mu_1(z_2)>0.
\end{eqnarray*}
For $z_2<\mu_2^*$,
$$
\log\widetilde\mu_1(z_2)=\log(\lambda_1-\gamma_1)-\log\left[\lambda_{1}(z_{2})^{\frac{\gamma_{1}-\lambda_{1}}{\lambda_{2}}}+(\frac{\gamma_{1}}{\nu_{1}^{*}}-\gamma_{1})(\frac{z_{2}}{\mu_{2}^{*}})^{\frac{\gamma_{1}-\lambda_{1}}{\lambda_{2}}}-\frac{\gamma_{1}}{\nu_{1}^{*}}\right],
$$
and
\begin{eqnarray*}
\frac{\partial\log\widetilde{\mu}_{1}(z_{2})}{\partial\lambda_{2}}& =&-\frac{\lambda_{1}(z_{2})^{\frac{\gamma_{1}-\lambda_{1}}{\lambda_{2}}}(\gamma_{1}-\lambda_{1})\left[\log(z_{2})\right](-\frac{1}{\lambda_{2}^{2}})+(\frac{\gamma_{1}}{\nu_{1}^{*}}-\gamma_{1})(\frac{z_{2}}{\mu_{2}^{*}})^{\frac{\gamma_{1}-\lambda_{1}}{\lambda_{2}}}(\gamma_{1}-\lambda_{1})\left[\log(\frac{z_{2}}{\mu_{2}^{*}})\right](-\frac{1}{\lambda_{2}^{2}})}{\lambda_{1}(z_{2})^{\frac{\gamma_{1}-\lambda_{1}}{\lambda_{2}}}+(\frac{\gamma_{1}}{\nu_{1}^{*}}-\gamma_{1})(\frac{z_{2}}{\mu_{2}^{*}})^{\frac{\gamma_{1}-\lambda_{1}}{\lambda_{2}}}-\frac{\gamma_{1}}{\nu_{1}^{*}}}\\
 & =&-\widetilde{\mu}_{1}(z_{2})\frac{1}{\lambda_{2}^{2}}\left\{ \lambda_{1}(z_{2})^{\frac{\gamma_{1}-\lambda_{1}}{\lambda_{2}}}\left[\log(z_{2})\right]+(\frac{\gamma_{1}}{\nu_{1}^{*}}-\gamma_{1})(\frac{z_{2}}{\mu_{2}^{*}})^{\frac{\gamma_{1}-\lambda_{1}}{\lambda_{2}}}\left[\log\left(\frac{z_{2}}{\mu_{2}^{*}}\right)\right]\right\}>0,
\end{eqnarray*}
where the strict inequality follows from $\log(z_2)<0$ and $z_2\le\mu_2^*$.

\noindent (iii) Consider $\partial u_1/\partial r_2$ next. We have
$$
\frac{\partial u_1}{\partial r_2}=1_{z_1\ge\widetilde\mu_1(z_2)}\frac{z_2}{1-z_2}\frac{1}{\widetilde\mu_2^2(z_1)}(1-a_1)\frac{\partial \widetilde\mu_2(z_1)}{\partial \lambda_1}.
$$
It remains to show the sign of $\partial \widetilde\mu_2(z_1)/\partial \lambda_1$, which is equivalent to showing the sign of $\partial \log\widetilde\mu_2(z_1)/\partial \lambda_1$. For $z_{1}\ge\mu_{1}^{N}$,
\[
\log[\widetilde{\mu}_{2}(z_{1})]=\frac{\lambda_{2}}{\gamma_{1}-\lambda_{1}}\log\left[\left(1-\frac{\gamma_{1}}{\lambda_{1}}\right)\frac{1}{z_{1}}+\frac{\gamma_{1}}{\lambda_{1}}\right].
\]
Hence,
\[
\frac{\partial\log[\widetilde{\mu}_{2}(z_{1})]}{\partial\lambda_{1}}=\frac{\lambda_{2}}{\left(\gamma_{1}-\lambda_{1}\right)^{2}}\log\left[\left(1-\frac{\gamma_{1}}{\lambda_{1}}\right)\frac{1}{z_{1}}+\frac{\gamma_{1}}{\lambda_{1}}\right]+\frac{\frac{\gamma_{1}}{\lambda_{1}^{2}}\left(\frac{1}{z_{1}}-1\right)}{\left(1-\frac{\gamma_{1}}{\lambda_{1}}\right)\frac{1}{z_{1}}+\frac{\gamma_{1}}{\lambda_{1}}}\frac{\lambda_{2}}{\gamma_{1}-\lambda_{1}}.
\]
Let $z\equiv\left[\widetilde{\mu}_{2}(z_{1})\right]^{\frac{\gamma_{1}-\lambda_{1}}{\lambda_{2}}}=\left(1-\frac{\gamma_{1}}{\lambda_{1}}\right)\frac{1}{z_{1}}+\frac{\gamma_{1}}{\lambda_{1}}$.
Since $z_{1}=\widetilde{\mu}_{1}(\widetilde{\mu}_{2}(z_{1}))=(\lambda_{1}-\gamma_{1})/(\lambda_{1}z-\gamma_{1})$,
$1/z_{1}-1=\frac{\lambda_{1}}{\lambda_{1}-\gamma_{1}}(z-1)$. The
expression above is simplified to
\[
\frac{\partial\log[\widetilde{\mu}_{2}(z_{1})]}{\partial\lambda_{1}}=\frac{\lambda_{2}}{\left(\gamma_{1}-\lambda_{1}\right)^{2}}\log z-\frac{z-1}{z}\frac{\gamma_{1}}{\lambda_{1}^{2}}\frac{\lambda_{1}\lambda_{2}}{\left(\gamma_{1}-\lambda_{1}\right)^{2}},
\]
which has the same sign as
\[
\Delta(z)\equiv\log z+\left(\frac{1}{z}-1\right)\frac{\gamma_{1}}{\lambda_{1}}.
\]
The first derivative of $\Delta(z)$ above is
\[
\Delta'(z)=\frac{1}{z}-\frac{1}{z^{2}}\frac{\gamma_{1}}{\lambda_{1}}=\frac{1}{z}\left(1-\frac{1}{z}\frac{\gamma_{1}}{\lambda_{1}}\right),
\]
which reaches its extreme at $z^{*}=\gamma_{1}/\lambda_{1}$. The
second derivative of $\Delta(z)$ is
\[
\Delta''(z)=-\frac{1}{z^{2}}-(-2)\frac{1}{z^{3}}\frac{\gamma_{1}}{\lambda_{1}}=\frac{1}{z^{2}}\left(2\frac{1}{z}\frac{\gamma_{1}}{\lambda_{1}}-1\right),
\]
which is $1/\left(z^{*}\right)^{2}$, positive at $z^{*}$. Therefore,
the minimum is reached at the point, and $\Delta(z)$ is decreasing
for $z<z^{*}$ and increasing for $z>z^{*}$. On one hand, when $\gamma_{1}>\lambda_{1}$,
$z\le1$, the minimum is achieved at $z^{*}=\gamma_{1}/\lambda_{1}>1$.
As $\Delta(z)$ is decreasing for $z\le1$, the minimum of $\Delta(z)$
is achieved when $z\le1$ and $\gamma_{1}>\lambda_{1}$ is achieved
at $z^{**}=1$, which is
\[
\Delta(1)=\log1+\left(\frac{1}{1}-1\right)\frac{\gamma_{1}}{\lambda_{1}}=0.
\]
When $\gamma_{1}<\lambda_{1}$, $z\ge1$, the minimum is achieved
at $z^{*}=\gamma_{1}/\lambda_{1}<1$. As $\Delta(z)$ is increasing
for $z\ge1$, the minimum of $\Delta(z)$ is achieved when $z\ge1$
and $\gamma_{1}<\lambda_{1}$ is achieved at $z^{**}=1$, which, as
calculated above, is $\Delta(1)=0$. Finally, when $\gamma_{1}=\lambda_{1}$,
the minimum is achieved at $z^{*}=1$, and the minimum is $0$. Therefore,
overall, regardless of the parameter, $\Delta(z)\ge0$. Because $\widetilde{\mu}_{2}(z_{1})\neq1$
and consequently $z\neq1$, the inequality holds strictly: $\Delta(z)>0$
for $z\neq1$. Therefore, $\partial u_{1}/\partial r_{2}>0$ for $z_1\geqslant \mu_1^N$.

For $z_{1}<\mu_{1}^{N}$,
\[
\log\widetilde{\mu}_{2}\left(z_{1}\right)=\frac{\lambda_{2}}{\gamma_{1}-\lambda_{1}}\log\left[\dfrac{\frac{\lambda_{1}-\gamma_{1}}{\lambda_{1}}\frac{1}{z_{1}}+\frac{\gamma_{1}}{\lambda_{1}}\frac{1}{\nu_{1}^{*}}}{1+\frac{1-\nu_{1}^{*}}{\nu_{1}^{*}}\frac{\gamma_{1}}{\lambda_{1}}\left(\mu_{2}^{*}\right)^{\frac{\lambda_{1}-\gamma_{1}}{\lambda_{2}}}}\right]=\frac{\lambda_{2}}{\gamma_{1}-\lambda_{1}}\log\left[\frac{\lambda_{1}\frac{1}{z_{1}}+\gamma_{1}\frac{1}{\nu_{1}^{*}}-\gamma_{1}\frac{1}{z_{1}}}{\lambda_{1}+\frac{1-\nu_{1}^{*}}{\nu_{1}^{*}}\gamma_{1}\left(\mu_{2}^{*}\right)^{\frac{\lambda_{1}-\gamma_{1}}{\lambda_{2}}}}\right].
\]
Denote $x\equiv\left[\widetilde{\mu}_{2}\left(z_{1}\right)\right]^{\frac{\gamma_{1}-\lambda_{1}}{\lambda_{2}}}=\frac{\lambda_{1}\frac{1}{z_{1}}+\gamma_{1}\frac{1}{\nu_{1}^{*}}-\gamma_{1}\frac{1}{z_{1}}}{\lambda_{1}+\frac{1-\nu_{1}^{*}}{\nu_{1}^{*}}\gamma_{1}\left(\mu_{2}^{*}\right)^{\frac{\lambda_{1}-\gamma_{1}}{\lambda_{2}}}}$
and $\mu\equiv\left(\mu_{2}^{*}\right)^{\frac{\gamma_{1}-\lambda_{1}}{\lambda_{2}}}$.
Then,
\begin{eqnarray*}
 &  & \frac{\partial\log\widetilde{\mu}_{2}\left(z_{1}\right)}{\partial\lambda_{1}}\\
 & = & \frac{\lambda_{2}}{\left(\gamma_{1}-\lambda_{1}\right)^{2}}\log x+\frac{\lambda_{2}}{\gamma_{1}-\lambda_{1}}\frac{1}{x}\cdot\\
 &  & \left\{ -\frac{\lambda_{1}\frac{1}{z_{1}}+\gamma_{1}\frac{1}{\nu_{1}^{*}}-\gamma_{1}\frac{1}{z_{1}}}{\left[\lambda_{1}+\frac{1-\nu_{1}^{*}}{\nu_{1}^{*}}\gamma_{1}\left(\mu_{2}^{*}\right)^{\frac{\lambda_{1}-\gamma_{1}}{\lambda_{2}}}\right]^{2}}\left[1+\frac{1-\nu_{1}^{*}}{\nu_{1}^{*}}\gamma_{1}\frac{1}{\mu}\frac{1}{\lambda_{2}}\log\left(\mu_{2}^{*}\right)\right]+\frac{\frac{1}{z_{1}}}{\lambda_{1}+\frac{1-\nu_{1}^{*}}{\nu_{1}^{*}}\gamma_{1}\left(\mu_{2}^{*}\right)^{\frac{\lambda_{1}-\gamma_{1}}{\lambda_{2}}}}\right\} ,
\end{eqnarray*}
which, because $\frac{\lambda_{2}}{\left(\gamma_{1}-\lambda_{1}\right)^{2}}\frac{1}{\lambda_{1}+\frac{1-\nu_{1}^{*}}{\nu_{1}^{*}}\gamma_{1}\left(\mu_{2}^{*}\right)^{\frac{\lambda_{1}-\gamma_{1}}{\lambda_{2}}}}>0$,
has the same sign as
\begin{align*}
 & \left(\lambda_{1}+\frac{1-\nu_{1}^{*}}{\nu_{1}^{*}}\gamma_{1}\frac{1}{\mu}\right)\log x+\left(\gamma_{1}-\lambda_{1}\right)\frac{1}{x}\left[-x-\frac{1-\nu_{1}^{*}}{\nu_{1}^{*}}\gamma_{1}\frac{1}{\mu}\frac{1}{\lambda_{2}}\log\left(\mu_{2}^{*}\right)x+\frac{1}{z_{1}}\right]\\
= & \left(\lambda_{1}+\frac{1-\nu_{1}^{*}}{\nu_{1}^{*}}\gamma_{1}\frac{1}{\mu}\right)\log x+\left(\gamma_{1}-\lambda_{1}\right)\frac{1}{x}\left(\frac{1}{z_{1}}-x\right)-\frac{1-\nu_{1}^{*}}{\nu_{1}^{*}}\gamma_{1}\frac{1}{\mu}\log\mu
\end{align*}
Since $1/z_{1}=\frac{\lambda_{1}x+\frac{1-\nu_{1}^{*}}{\nu_{1}^{*}}\gamma_{1}\frac{x}{\mu}-\frac{\gamma_{1}}{\nu_{1}^{*}}}{\lambda_{1}-\gamma_{1}}=\frac{\frac{\gamma_{1}}{\nu_{1}^{*}}\frac{1}{x}-\frac{1-\nu_{1}^{*}}{\nu_{1}^{*}}\gamma_{1}\frac{1}{\mu}-\lambda_{1}}{\gamma_{1}-\lambda_{1}}x$,
\[
\frac{1}{z_{1}}-x=\frac{\frac{\gamma_{1}}{\nu_{1}^{*}}\frac{1}{x}-\frac{1-\nu_{1}^{*}}{\nu_{1}^{*}}\gamma_{1}\frac{1}{\mu}-\gamma_{1}}{\gamma_{1}-\lambda_{1}}x.
\]
Plugging this into the previous expression, we have
\[
\left(\lambda_{1}+\frac{1-\nu_{1}^{*}}{\nu_{1}^{*}}\gamma_{1}\frac{1}{\mu}\right)\log x+\frac{\gamma_{1}}{\nu_{1}^{*}}\frac{1}{x}-\frac{1-\nu_{1}^{*}}{\nu_{1}^{*}}\gamma_{1}\frac{1}{\mu}-\gamma_{1}-\frac{1-\nu_{1}^{*}}{\nu_{1}^{*}}\gamma_{1}\frac{1}{\mu}\log\mu
\]
which has the same sign as
\[
\Delta(x)=\left(\frac{\lambda_{1}}{\gamma_{1}}+\frac{1-\nu_{1}^{*}}{\nu_{1}^{*}}\frac{1}{\mu}\right)\log x+\frac{1}{\nu_{1}^{*}}\frac{1}{x}-\frac{1-\nu_{1}^{*}}{\nu_{1}^{*}}\frac{1}{\mu}-\frac{1-\nu_{1}^{*}}{\nu_{1}^{*}}\frac{1}{\mu}\log\mu-1.
\]
Its first derivative is
\[
\Delta'(x)=\left(\frac{\lambda_{1}}{\gamma_{1}}+\frac{1-\nu_{1}^{*}}{\nu_{1}^{*}}\frac{1}{\mu}\right)\frac{1}{x}-\frac{1}{\nu_{1}^{*}}\frac{1}{x^{2}}.
\]
Its second derivative is
\[
\Delta''(x)=-\left(\frac{\lambda_{1}}{\gamma_{1}}+\frac{1-\nu_{1}^{*}}{\nu_{1}^{*}}\frac{1}{\mu}\right)\frac{1}{x^{2}}+2\frac{1}{\nu_{1}^{*}}\frac{1}{x^{3}}=-\frac{1}{x}\left[\left(\frac{\lambda_{1}}{\gamma_{1}}+\frac{1-\nu_{1}^{*}}{\nu_{1}^{*}}\frac{1}{\mu}\right)\frac{1}{x}-\frac{1}{\nu_{1}^{*}}\frac{1}{x^{2}}\right]+\frac{1}{\nu_{1}^{*}}\frac{1}{x^{3}}.
\]
Therefore, when $\Delta'(x^{*})=0$, $\Delta''(x^{*})>0$. The minimum
is reached at 
\[
x^{*}=\frac{\frac{1}{\nu_{1}^{*}}}{\frac{\lambda_{1}}{\gamma_{1}}+\frac{1-\nu_{1}^{*}}{\nu_{1}^{*}}\frac{1}{\mu}}=\frac{\frac{1}{\nu_{1}^{*}}\frac{\gamma_{1}}{\lambda_{1}}}{1+\frac{1-\nu_{1}^{*}}{\nu_{1}^{*}}\frac{\gamma_{1}}{\lambda_{1}}\frac{1}{\mu}},
\]
and $\Delta'(x)<0$ for $x<x^{*}$ and $\Delta'(x)>0$ for $x>x^{*}$.
Since 
\[
\left(x^{*}\right)^{\frac{\lambda_{2}}{\gamma_{1}-\lambda_{1}}}=\left(\frac{\frac{1}{\nu_{1}^{*}}\frac{\gamma_{1}}{\lambda_{1}}}{1+\frac{1-\nu_{1}^{*}}{\nu_{1}^{*}}\frac{\gamma_{1}}{\lambda_{1}}\mu}\right)^{\frac{\lambda_{2}}{\gamma_{1}-\lambda_{1}}}>\left(\frac{\frac{\lambda_{1}-\gamma_{1}}{\lambda_{1}}\frac{1}{z_{1}}+\frac{1}{\nu_{1}^{*}}\frac{\gamma_{1}}{\lambda_{1}}}{1+\frac{1-\nu_{1}^{*}}{\nu_{1}^{*}}\frac{\gamma_{1}}{\lambda_{1}}\mu}\right)^{\frac{\lambda_{2}}{\gamma_{1}-\lambda_{1}}}=\widetilde{\mu}_{2}\left(z_{1}\right),
\]
for $z_{1}\le\mu_{1}^{N}$, it holds for $z_{1}=\mu_{1}^{N}$, it
holds that $\left(x^{*}\right)^{\frac{\lambda_{2}}{\gamma_{1}-\lambda_{1}}}>\mu_{1}^{N}$.
When $\gamma_{1}>\lambda_{1}$, $x<\mu<x^{*}$, so the minimum is
attained at $\mu$; when $\gamma_{1}<\lambda_{1}$, $x>\mu>x^{*}$,
so the minimum value is also attained at $\mu$.
\[
\Delta(\mu)=\left(\frac{\lambda_{1}}{\gamma_{1}}+\frac{1-\nu_{1}^{*}}{\nu_{1}^{*}}\frac{1}{\mu}\right)\log\mu+\frac{1}{\nu_{1}^{*}}\frac{1}{\mu}-\frac{1-\nu_{1}^{*}}{\nu_{1}^{*}}\frac{1}{\mu}-\frac{1-\nu_{1}^{*}}{\nu_{1}^{*}}\frac{1}{\mu}\log\mu-1=\frac{\lambda_{1}}{\gamma_{1}}\log\mu+\frac{1}{\mu}-1.
\]
Denote this function of $\mu$ by $\psi$:
\[
\psi\left(\mu\right)\equiv\frac{\lambda_{1}}{\gamma_{1}}\log\mu+\frac{1}{\mu}-1.
\]
Its first derivative is
\[
\psi'\left(\mu\right)=\frac{\lambda_{1}}{\gamma_{1}}\frac{1}{\mu}-\frac{1}{\mu^{2}}
\]
and its second derivative is
\[
\psi''\left(\mu\right)=-\frac{\lambda_{1}}{\gamma_{1}}\frac{1}{\mu^{2}}+\frac{2}{\mu^{3}}=-\left(\frac{\lambda_{1}}{\gamma_{1}}\frac{1}{\mu}-\frac{1}{\mu^{2}}\right)\frac{1}{\mu}+\frac{1}{\mu^{3}}.
\]
When $\psi'(\mu^{*})=0$, $\psi''(\mu^{*})>0$, so the minimum of
$\psi(\mu)$ is achieved at $\mu^{*}=\gamma_{1}/\lambda_{1}$, and
$\psi'(\mu)<0$ for $\mu<\mu^{*}$ and $\psi'(\mu)>0$ for $\mu>\mu^{*}$.
When $\gamma_{1}>\lambda_{1}$, $\mu\le1<\mu^{*}$, so the minimum
of $\psi(\mu)$ is achieved at $\mu^{**}=1$; when $\gamma_{1}<\lambda_{1}$,
$\mu\ge1>\mu^{*}$, so the minimum of $\psi(\mu)$ is also achieved
at $\mu^{**}=1$; and when $\gamma_{1}=\lambda_{1}$, $\mu^{*}=1$,
so the minimum of $\psi(\mu)$ is also achieved at $\mu^{**}=1$.
The minimum value is $\psi(\mu^{**})=0$. Therefore, the sign of $\partial\log\widetilde{\mu}_{2}\left(z_{1}\right)/\partial\lambda_{1}$,
and consequently the sign of $\partial u_{1}/\partial r_{2}$, is positive
for any $z_{1}<\mu_{1}^{N}$.

\noindent (ii) Consider $\partial u_2/\partial r_2$. We have
$$
\frac{\partial u_2}{\partial r_2}=1_{z_2>\widetilde\mu_2(z_1)}\frac{z_1}{1-z_1}\frac{1}{\widetilde\mu_1^2(z_2)}(1-a_1)\frac{\partial\widetilde\mu_1(z_2)}{\partial\lambda_1}.
$$
It remains to show the sign of $\partial \widetilde\mu_1(z_2)/\partial \lambda_1$, which is equivalent to showing the sign of $\partial \log\widetilde\mu_1(z_2)/\partial \lambda_1$. When $\mu_{2}^{*}\leqslant z_{2}\le1$,
\[
\widetilde{\mu}_{1}(z_{2})=\frac{\lambda_{1}-\gamma_{1}}{\lambda_{1}\left(z_{2}\right)^{\frac{\gamma_{1}-\lambda_{1}}{\lambda_{2}}}-\gamma_{1}}.
\]
Hence
\begin{align*}
\frac{\partial\widetilde{\mu}_{1}(z_{2})}{\partial\lambda_{1}} & =-\frac{\left(\lambda_{1}-\gamma_{1}\right)\left[\left(z_{2}\right)^{\frac{\gamma_{1}-\lambda_{1}}{\lambda_{2}}}-\lambda_{1}\left(z_{2}\right)^{\frac{\gamma_{1}-\lambda_{1}}{\lambda_{2}}}\log\left(z_{2}^{\frac{1}{\lambda_{2}}}\right)\right]}{\left[\lambda_{1}\left(z_{2}\right)^{\frac{\gamma_{1}-\lambda_{1}}{\lambda_{2}}}-\gamma_{1}\right]^{2}}+\frac{1}{\lambda_{1}\left(z_{2}\right)^{\frac{\gamma_{1}-\lambda_{1}}{\lambda_{2}}}-\gamma_{1}}\\
 & =\frac{\lambda_{1}\left(z_{2}\right)^{\frac{\gamma_{1}-\lambda_{1}}{\lambda_{2}}}-\gamma_{1}+\left(\gamma_{1}-\lambda_{1}\right)\left[\left(z_{2}\right)^{\frac{\gamma_{1}-\lambda_{1}}{\lambda_{2}}}-\lambda_{1}\left(z_{2}\right)^{\frac{\gamma_{1}-\lambda_{1}}{\lambda_{2}}}\log\left(z_{2}^{\frac{1}{\lambda_{2}}}\right)\right]}{\left[\lambda_{1}\left(z_{2}\right)^{\frac{\gamma_{1}-\lambda_{1}}{\lambda_{2}}}-\gamma_{1}\right]^{2}},
\end{align*}
which has the same sign as
\[
\lambda_{1}-\gamma_{1}/\left(z_{2}^{\frac{\gamma_{1}-\lambda_{1}}{\lambda_{2}}}\right)+\gamma_{1}-\lambda_{1}-\lambda_{1}\log\left(z_{2}^{\frac{\gamma_{1}-\lambda_{1}}{\lambda_{2}}}\right).
\]
Let $z\equiv z_{2}^{\frac{\gamma_{1}-\lambda_{1}}{\lambda_{2}}}$.
The expression becomes
\[
\Delta(z)\equiv\gamma_{1}-\gamma_{1}/z-\lambda_{1}\log z.
\]
Its first derivative is 
\[
\Delta'(z)=\gamma_{1}/z^{2}-\lambda_{1}/z=\frac{1}{z^{2}}\left(\gamma_{1}-\lambda_{1}z\right).
\]
Its second derivative is
\[
\Delta''(z)=-2\gamma_{1}/z^{3}+\lambda_{1}/z^{2}=\frac{1}{z^{3}}\left(\lambda_{1}z-2\gamma_{1}\right).
\]
It reaches the extreme at $z^{*}=\gamma_{1}$, and at the extreme, $\Delta''(z^{*})=-\gamma_{1}/(z^{*})^{3}<0$.
Therefore, the maximum is reached at $z^{*}=\gamma_{1}/\lambda_{1}$
for any given $\gamma_{1}$, and $\Delta(z)$ is increasing in $z$
for $z<z^{*}$ and decreasing in $z$ for $z>z^{*}$. When $\gamma_{1}>\lambda_{1}$,
$z\le1$, so the maximum is reached at $z^{**}=1<\gamma_{1}/\lambda_{1}$.
If $\gamma_{1}\le\lambda_{1}$, $z\ge1$, so the maximum is reached
at $z^{**}=1>\gamma_{1}/\lambda_{1}$. Altogether, $\Delta(z)\le0$,
and when $z_{2}\neq1$, $\Delta(z)<0$, thus $\partial u_{2}/\partial r_{2}<0$
for $z_{2}\ge\mu_{2}^{*}$. For $z_2<\mu_2^*$, For $z_{2}<\mu_{2}^{*}$,
\[
\widetilde{\mu}_{1}(z_{2})=\frac{\lambda_{1}-\gamma_{1}}{\lambda_{1}\left(z_{2}\right)^{\frac{\gamma_{1}-\lambda_{1}}{\lambda_{2}}}+\left(\frac{\gamma_{1}}{\nu_{1}^{*}}-\gamma_{1}\right)\left(\frac{z_{2}}{\mu_{2}^{*}}\right)^{\frac{\gamma_{1}-\lambda_{1}}{\lambda_{2}}}-\frac{\gamma_{1}}{\nu_{1}^{*}}}.
\]
Then,
\begin{eqnarray*}
 \frac{\partial\widetilde{\mu}_{1}\left(z_{2}\right)}{\partial\lambda_{1}}
 &= & \frac{1}{\left[\lambda_{1}\left(z_{2}\right)^{\frac{\gamma_{1}-\lambda_{1}}{\lambda_{2}}}+\left(\frac{\gamma_{1}}{\nu_{1}^{*}}-\gamma_{1}\right)\left(\frac{z_{2}}{\mu_{2}^{*}}\right)^{\frac{\gamma_{1}-\lambda_{1}}{\lambda_{2}}}-\frac{\gamma_{1}}{\nu_{1}^{*}}\right]^{2}}\times\\
 &  & \Bigg\{\lambda_{1}\left(z_{2}\right)^{\frac{\gamma_{1}-\lambda_{1}}{\lambda_{2}}}+\left(\frac{\gamma_{1}}{\nu_{1}^{*}}-\gamma_{1}\right)\left(\frac{z_{2}}{\mu_{2}^{*}}\right)^{\frac{\gamma_{1}-\lambda_{1}}{\lambda_{2}}}-\frac{\gamma_{1}}{\nu_{1}^{*}}+\\
 &  & (\gamma_{1}-\lambda_{1})\left[(z_{2})^{\frac{\gamma_{1}-\lambda_{1}}{\lambda_{2}}}-\lambda_{1}(z_{2})^{\frac{\gamma_{1}-\lambda_{1}}{\lambda_{2}}}\log\left(z_{2}^{\frac{1}{\lambda_{2}}}\right)-\left(\frac{\gamma_{1}}{\nu_{1}^{*}}-\gamma_{1}\right)\left(\frac{z_{2}}{\mu_{2}^{*}}\right)^{\frac{\gamma_{1}-\lambda_{1}}{\lambda_{2}}}\frac{1}{\lambda_{2}}\log\left(\frac{z_{2}}{\mu_{2}^{*}}\right)\right]\Bigg\}.
\end{eqnarray*}
Let $z\equiv z_{2}^{\frac{\gamma_{1}-\lambda_{1}}{\lambda_{2}}}$
and $\mu\equiv\left(\mu_{2}^{*}\right)^{\frac{\gamma_{1}-\lambda_{1}}{\lambda_{2}}}$.
The expression has the same sign as
\[
\lambda_{1}z+\left(\frac{\gamma_{1}}{\nu_{1}^{*}}-\gamma_{1}\right)\frac{z}{\mu}-\frac{\gamma_{1}}{\nu_{1}^{*}}+(\gamma_{1}-\lambda_{1})z-\lambda_{1}z\log z-\left(\frac{\gamma_{1}}{\nu_{1}^{*}}-\gamma_{1}\right)\frac{z}{\mu}\log\left(\frac{z}{\mu}\right),
\]
which, because $z>0$, has the same sign as
\[
\lambda_{1}+\frac{1-\nu_{1}^{*}}{\nu_{1}^{*}}\frac{1}{\mu}\gamma_{1}-\frac{\gamma_{1}}{\nu_{1}^{*}}\frac{1}{z}+\gamma_{1}-\lambda_{1}-\lambda_{1}\log z-\frac{1-\nu_{1}^{*}}{\nu_{1}^{*}}\frac{1}{\mu}\gamma_{1}\log\left(\frac{z}{\mu}\right),
\]
and, because $\gamma_{1}>0$, has the same sign as
\[
1+\frac{1-\nu_{1}^{*}}{\nu_{1}^{*}}\frac{1}{\mu}-\frac{1}{\nu_{1}^{*}}\frac{1}{z}-\frac{\lambda_{1}}{\gamma_{1}}\log z-\frac{1-\nu_{1}^{*}}{\nu_{1}^{*}}\frac{1}{\mu}\log\left(\frac{z}{\mu}\right)\equiv\Delta\left(z\right).
\]
Its first derivative is
\[
\Delta'(z)=\frac{1}{z}\left[\frac{1}{\nu_{1}^{*}}\frac{1}{z}-\left(\frac{\lambda_{1}}{\gamma_{1}}+\frac{1-\nu_{1}^{*}}{\nu_{1}^{*}}\frac{1}{\mu}\right)\right].
\]
It reaches the extreme at
\[
z^{*}=\frac{\frac{1}{\nu_{1}^{*}}}{\frac{\gamma_{1}}{\lambda_{1}}+\frac{1-\nu_{1}^{*}}{\nu_{1}^{*}}\frac{1}{\mu}}.
\]
Its second derivative is 
\[
\Delta''(z)=-\frac{1}{z^{2}}\left[2\frac{1}{\nu_{1}^{*}}\frac{1}{z}-\left(\frac{\lambda_{1}}{\gamma_{1}}+\frac{1-\nu_{1}^{*}}{\nu_{1}^{*}}\frac{1}{\mu}\right)\right].
\]
At $z^{*}$, $\Delta''(z^{*})=-\frac{1}{(z^{*})^{2}}\left(\frac{\lambda_{1}}{\gamma_{1}}+\frac{1-\nu_{1}^{*}}{\nu_{1}^{*}}\frac{\gamma_{1}}{\lambda_{1}}\frac{1}{\mu}\right)<0$,
so the maximum value is reached at $z^{*}$. 
\[
(z^{*})^{\frac{\lambda_{2}}{\gamma_{1}-\lambda_{1}}}=\left(\frac{\frac{\gamma_{1}}{\lambda_{1}}\frac{1}{\nu_{1}^{*}}}{1+\frac{1-\nu_{1}^{*}}{\nu_{1}^{*}}\frac{\gamma_{1}}{\lambda_{1}}\frac{1}{\mu}}\right)^{\frac{\lambda_{2}}{\gamma_{1}-\lambda_{1}}}>\left(\frac{\frac{\lambda_{1}-\gamma_{1}}{\lambda_{1}}\frac{1}{z_{1}}+\frac{\gamma_{1}}{\lambda_{1}}\frac{1}{\nu_{1}^{*}}}{1+\frac{1-\nu_{1}^{*}}{\nu_{1}^{*}}\frac{\gamma_{1}}{\lambda_{1}}\frac{1}{\mu}}\right)^{\frac{\lambda_{2}}{\gamma_{1}-\lambda_{1}}}
\]
for any $z_{1}\le\mu_{1}^{N}$; in particular, the inequality holds
for $z_{1}=\mu_{1}^{N}$, in which case the right-hand side is $\mu$.
Hence, $z^{*}>\mu$. However, when $\gamma_{1}>\lambda_{1}$, $z<\mu$.
Therefore, $\Delta(z)$ reaches maximum value at $\mu$, and has value
\[
\Delta\left(\mu\right)=1-\frac{1}{\mu}-\frac{\lambda_{1}}{\gamma_{1}}\log\mu.
\]
Let's denote by $\psi(\mu)$ the function that treats $\mu$ as a
variable. Its first derivative is
\[
\psi'(\mu)=\frac{1}{\mu^{2}}-\frac{\lambda_{1}}{\gamma_{1}}\frac{1}{\mu}=\frac{1}{\mu^{2}}\left(1-\frac{\lambda_{1}}{\gamma_{1}}\mu\right).
\]
The extreme value of $\psi(\mu)$ is reached at $\mu^{*}=\gamma_{1}/\lambda_{1}$,
and it is strictly decreasing when $\mu>\mu^{*}$ and strictly increasing
when $\mu<\mu^{*}$, so the maximum value is reached at $\mu^{*}$.
The maximum of $\psi(\mu)$ is
\[
\psi(\mu^{*})=1-\frac{\lambda_{1}}{\gamma_{1}}-\frac{\lambda_{1}}{\gamma_{1}}\log\left(\frac{\gamma_{1}}{\lambda_{1}}\right).
\]
When $\gamma_{1}>\lambda_{1}$, $\mu\le1<\mu^{*}$, so when $\gamma_{1}>\lambda_{1}$,
$\psi(\mu)$ for $\mu\le1$ reaches its maximum value $\psi(1)=0$
at $1$. When $\gamma_{1}<\lambda_{1}$, $\mu\ge1>\mu^{*}$, so when
$\gamma_{1}<\lambda_{1}$, $\psi(\mu)$ for $\mu\ge1$ also reaches
its maximum value $\psi(1)=0$ at $1$. Finally, when $\gamma_{1}=\lambda_{1}$,
$\psi(\mu)<\psi(\mu^{*})=\Delta(1)=0$.

Altogether, $\Delta(z)\le\Delta(\mu)=\psi(\mu)<\psi(1)=0$. Hence,
${\partial\widetilde{\mu}_{1}\left(z_{2}\right)}/{\partial\lambda_{1}}<0$
for $z_{2}<\mu_{2}^{*}$.
\end{proof}

\begin{proof}[{\bf (iii) Effects of $c_1$}]
Consider $\partial u_1/\partial c_1$ first. The cost $c_1$ only affects $u_1$ through $\mu_2^*$ in $\widetilde \mu_2(z_1)$ when $z_1$ is sufficiently small (to be precise, when $z_1\leqslant \mu_1^N$). Formally,
$$
\frac{\partial u_1}{\partial c_1}=1_{z_1\ge\widetilde\mu_1(z_2)}\frac{z_2}{1-z_2}\frac{1}{\widetilde \mu_2^2(z_1)}\frac{\partial \widetilde\mu_2(z_1)}{\partial \mu_2^*}\frac{\partial \mu_2^*}{\partial c_1}.
$$
Because $\mu_2^*=1-c_1$, $\partial \mu_2^*/\partial k=-1$, and as a result, $\partial u_1/\partial c_1$ has the opposite sign as $\partial \widetilde \mu_2(z_1)/\partial \mu_2^*$, which is zero when $z_1>\mu_1^N$, and is the following expression when $z_1\leqslant \mu_1^N$:
$$
\frac{\lambda_2}{\gamma_1-\lambda_1}\left[\dfrac{\frac{\lambda_{1}-\gamma_{1}}{\lambda_{1}}\frac{1}{\mu_{1}}+\frac{\gamma_{1}}{\lambda_{1}}\frac{1}{\nu_{1}^{*}}}{1+\frac{1-\nu_{1}^{*}}{\nu_{1}^{*}}\frac{\gamma_{1}}{\lambda_{1}}(\mu_{2}^{*})^{\frac{\lambda_{1}-\gamma_{1}}{\lambda_{2}}}}\right]^{\frac{\lambda_{2}}{\gamma_{1}-\lambda_{1}}-1}\frac{\frac{\lambda_{1}-\gamma_{1}}{\lambda_{1}}\frac{1}{\mu_{1}}+\frac{\gamma_{1}}{\lambda_{1}}\frac{1}{\nu_{1}^{*}}}{\left[1+\frac{1-\nu_{1}^{*}}{\nu_{1}^{*}}\frac{\gamma_{1}}{\lambda_{1}}(\mu_{2}^{*})^{\frac{\lambda_{1}-\gamma_{1}}{\lambda_{2}}}\right]^2}\frac{1-\nu_1^*}{\nu_1^*}\frac{\gamma_1}{\lambda_1}(\mu_2^*)^{\frac{\lambda_1-\gamma_1}{\lambda_2}-1}\frac{\gamma_1=\lambda_1}{\lambda_2}.
$$
The expression above is positive, because it can be simplified to
$$
\widetilde \mu_2(z_1)\frac{\frac{1-\nu_1^*}{\nu_1^*}\frac{\gamma_1}{\lambda_1}(\mu_2^*)^{\frac{\lambda_1-\gamma_1}{\lambda_2}}}{1+\frac{1-\nu_1^*}{\nu_1^*}\frac{\gamma_1}{\lambda_1}(\mu_2^*)^{\frac{\lambda_1-\gamma_1}{\lambda_2}}},
$$
which is a product of all positive terms. Overall, the expression of $\partial u_1/\partial c_1$ becomes
$$
\frac{\partial u_1}{\partial c_1}=-1_{z_1\geqslant \widetilde \mu_1(z_2)}\frac{1}{1-z_2}\frac{z_2}{\widetilde\mu_2(z_1)} \frac{\frac{1-\nu_1^*}{\nu_1^*}\frac{\gamma_1}{\lambda_1}(\mu_2^*)^{\frac{\lambda_1-\gamma_1}{\lambda_2}}}{1+\frac{1-\nu_1^*}{\nu_1^*}\frac{\gamma_1}{\lambda_1}(\mu_2^*)^{\frac{\lambda_1-\gamma_1}{\lambda_2}}}.
$$
Therefore, when $z_1\leqslant \mu_1^N$, $\partial u_1/\partial c_1<0$.

Consider $\partial u_2/\partial c_1$ next. The cost coefficient $c_1$ only affects $u_2$ through $\mu_2^*$ in $\widetilde \mu_1(z_2)$ when $z_2$ is sufficiently small, when $z_2\leqslant \mu_2^*$, to be precise. Formally,
$$
\frac{\partial u_2}{\partial c_1}=1_{z_2\geqslant \widetilde\mu_2(z_1)}\frac{z_1}{1-z_1}\frac{1}{\widetilde\mu_1^2(z_2)}\frac{\partial \widetilde \mu_1(z_2)}{\partial \mu_2^*}\frac{\partial \mu_2^*}{\partial k_2}.
$$
In the terms of the product, $\partial \mu_2^*/\partial k_2=-1$ and when $z_2\leqslant \mu_2^*$,
\begin{eqnarray*}
&&\frac{\partial \widetilde\mu_1(z_2)}{\partial \mu_2^*}\\
&=&-\frac{\lambda_1-\gamma_1}{\left[\lambda_{1}(\mu_{2})^{\frac{\gamma_{1}-\lambda_{1}}{\lambda_{2}}}+(\frac{\gamma_{1}}{\nu_{1}^{*}}-\gamma_{1})(\frac{\mu_{2}}{\mu_{2}^{*}})^{\frac{\gamma_{1}-\lambda_{1}}{\lambda_{2}}}-\frac{\gamma_{1}}{\nu_{1}^{*}}\right]^2}(\frac{\gamma_{1}}{\nu_{1}^{*}}-\gamma_{1})(\mu_2)^{\frac{\gamma_{1}-\lambda_{1}}{\lambda_{2}}}(\mu_2^*)^{\frac{\lambda_1-\gamma_1}{\lambda_2}-1}\frac{\lambda_1-\gamma_1}{\lambda_2}\\
&=&-\widetilde\mu_1^2(z_2)\frac{1}{\lambda_2}\left(\frac{\gamma_1}{\nu_1^*}-\gamma_1\right)(z_2)^{\frac{\gamma_1-\lambda_1}{\lambda_2}}(\mu_2^*)^{\frac{\lambda_1-\gamma_1}{\lambda_2}-1}.
\end{eqnarray*}
Overall, with some rearrangements,
$$
\frac{\partial u_2}{\partial c_1}=1_{\widetilde\mu_2(z_1)\leqslant z_2\leqslant \mu_2^*}\frac{z_1}{1-z_1}\frac{1-\nu_1^*}{\nu_1^*}\frac{\gamma_1}{\mu_2^*\lambda_2}\left(\frac{z_2}{\mu_2^*}\right)^{\frac{\gamma_1-\lambda_1}{\lambda_2}}.
$$
Therefore, $\partial u_2/\partial c_1$ is positive when $\widetilde\mu_2(z_1)\leqslant z_2\leqslant \mu_2^*$.
\end{proof}

\begin{proof}[{\bf (iv) Effects of $k_2$}]
Consider $\partial u_1/\partial k_2$ first.
\begin{eqnarray*}
\frac{\partial u_1}{\partial k_2}&=&
1_{z_1\ge\widetilde\mu_1(z_2)}1_{z_1\leqslant \mu_1^N}\frac{z_2}{1-z_2}\frac{1}{\widetilde\mu_2^2(z_1)}\frac{\partial\widetilde\mu_2(z_1)}{\partial \nu_1^*}\frac{\partial \nu_1^*}{\partial k_2}\\
&=&-1_{\widetilde\mu_1(z_2) \leqslant z_1\leqslant \mu_1^N} \frac{z_2}{1-z_2}\frac{1}{\widetilde\mu_2^2(z_1)}\frac{1}{1-w_1}\frac{\partial\widetilde\mu_2(z_1)}{\partial \nu_1^*}.
\end{eqnarray*}
The sign of $\partial \widetilde\mu_2(z_1)/\partial \nu_1^*$ is the same as the sign of $\partial\log\widetilde\mu_2(z_1)/\partial \nu_1^*$, so we consider the latter. When $z_1\leqslant \mu_1^N$,
$$
\log\widetilde\mu_2(z_1)=\frac{\lambda_2}{\gamma_1-\lambda_1}\left\{
\log\left[\frac{\lambda_1-\gamma_1}{\lambda_1}\frac{1}{z_1}+\frac{\gamma_1}{\lambda_1}\frac{1}{\nu_1^*}\right]-\log\left[1+\frac{1}{\nu_1^*}\frac{\gamma_1}{\lambda_1}(\mu_2^*)^{\frac{\lambda_1-\gamma_1}{\lambda_2}}-\frac{\gamma_1}{\lambda_1}(\mu_2^*)^{\frac{\lambda_1-\gamma_1}{\lambda_2}}\right]
\right\}.
$$
Hence,
\begin{eqnarray*}
\frac{\partial\log\widetilde\mu_2(z_1)}{\partial \nu_1^*}&=&\frac{\lambda_2}{\gamma_1-\lambda_1}\bigg\{\left[\frac{\lambda_1-\gamma_1}{\lambda_1}\frac{1}{z_1}+\frac{\gamma_1}{\lambda_1}\frac{1}{\nu_1^*}\right]^{-1}\frac{\gamma_1}{\lambda_1}\frac{1}{(\nu_1^*)^2}(-1)\\
&&\quad-\left[1+\frac{1}{\nu_1^*}\frac{\gamma_1}{\lambda_1}(\mu_2^*)^{\frac{\lambda_1-\gamma_1}{\lambda_2}}-\frac{\gamma_1}{\lambda_1}(\mu_2^*)^{\frac{\lambda_1-\gamma_1}{\lambda_2}}\right]^{-1}\frac{\gamma_1}{\lambda_1}(\mu_2^*)^{\frac{\lambda_1-\gamma_1}{\lambda_2}}\left(\frac{1}{\nu_1^*}\right)^2(-1)\bigg\}\\
&=&\frac{\lambda_2}{\lambda_1-\gamma_1}\frac{\gamma_1}{\lambda_1}\frac{1}{(\nu_1^*)^2}\left[\frac{1}{\frac{\lambda_1-\gamma_1}{\lambda_1}\frac{1}{z_1}+\frac{\gamma_1}{\lambda_1}\frac{1}{\nu_1^*}}-\frac{(\mu_2^*)^{\frac{\lambda_1-\gamma_1}{\lambda_2}}}{1+\frac{1}{\nu_1^*}\frac{\gamma_1}{\lambda_1}(\mu_2^*)^{{\frac{\lambda_1-\gamma_1}{\lambda_2}}}-\frac{\gamma_1}{\lambda_1}(\mu_2^*)^{{\frac{\lambda_1-\gamma_1}{\lambda_2}}}}\right],
\end{eqnarray*}
which has the same sign as
\begin{eqnarray*}
&&\frac{1}{\lambda_1-\gamma_1}\left[1+\frac{1}{\nu_1^*}\frac{\gamma_1}{\lambda_1}(\mu_2^*)^{{\frac{\lambda_1-\gamma_1}{\lambda_2}}}-\frac{\gamma_1}{\lambda_1}(\mu_2^*)^{{\frac{\lambda_1-\gamma_1}{\lambda_2}}}-(\mu_2^*)^{\frac{\lambda_1-\gamma_1}{\lambda_2}}\left(\frac{\lambda_1-\gamma_1}{\lambda_1}\frac{1}{z_1}+\frac{\gamma_1}{\lambda_1}\frac{1}{\nu_1^*}\right)\right]\\
&=&\frac{1}{\lambda_1-\gamma_1}\left\{1-\left[\left(1-\frac{\gamma_1}{\lambda_1}\right)\frac{1}{z_1}+\frac{\gamma_1}{\lambda_1}\right](\mu_2^*)^{\frac{\lambda_1-\gamma_1}{\lambda_2}}\right\}\\
&=&\frac{1}{\lambda_1-\gamma_1}-\left(\frac{1}{\lambda_1}\frac{1}{z_1}+\frac{1}{\lambda_1-\gamma_1}\frac{\gamma_1}{\lambda_1}\right)(\mu_2^*)^{\frac{\lambda_1-\gamma_1}{\lambda_2}}\\
&<&\frac{1}{\lambda_1-\gamma_1}-\left(\frac{1}{\lambda_1}\frac{1}{\mu_1^N}+\frac{1}{\lambda_1-\gamma_1}\frac{\gamma_1}{\lambda_1}\right)(\mu_2^*)^{\frac{\lambda_1-\gamma_1}{\lambda_2}}\\
&=&\frac{1}{\lambda_1-\gamma_1}-\left\{\frac{1}{\lambda_1-\gamma_1}\left[(\mu_2^*)^{\frac{\gamma_1-\lambda_1}{\lambda_2}}-\frac{\gamma_1}{\lambda_1}\right]+\frac{1}{\lambda_1-\gamma_1}\frac{\gamma_1}{\lambda_1}\right\}(\mu_2^*)^{\frac{\lambda_1-\gamma_1}{\lambda_2}}=0,
\end{eqnarray*}
where the strict inequality follows from $z_1<\mu_1^N$. Therefore, $\partial\widetilde\mu_2(z_1)/\partial \nu_1^*$ is negative, and consequently, $\partial u_1/\partial k_2$ is positive.

Consider $\partial u_2/\partial k_2$ next.
$$
\frac{\partial u_2}{\partial k_2}=1_{z_2>\widetilde\mu_2(z_1)}\frac{z_1}{1-z_1}\frac{1}{\widetilde\mu_1^2(z_2)}\frac{\partial\widetilde\mu_1(z_2)}{\partial \nu_1^*}\frac{\partial\nu_1^*}{\partial k_2}=-1_{\widetilde\mu_2(z_1)<z_2<\mu_2^*}\frac{z_1}{1-z_1}\frac{1}{\widetilde\mu_1^2(z_2)}\frac{1}{1-w_1}\frac{\partial\widetilde\mu_1(z_2)}{\partial \nu_1^*}.
$$
For $z_2<\mu_2^*$,
$$
\widetilde\mu_1(z_2)=\frac{\lambda_{1}-\gamma_{1}}{\lambda_{1}(\mu_{2})^{\frac{\gamma_{1}-\lambda_{1}}{\lambda_{2}}}+(\frac{\gamma_{1}}{\nu_{1}^{*}}-\gamma_{1})(\frac{z_{2}}{\mu_{2}^{*}})^{\frac{\gamma_{1}-\lambda_{1}}{\lambda_{2}}}-\frac{\gamma_{1}}{\nu_{1}^{*}}}.
$$
Then,
$$
\frac{\widetilde\mu_1(z_2)}{\partial \nu_1^*}=\widetilde\mu_1^2(z_2)\frac{\gamma_1}{(\nu_1^*)^2}\frac{\left(\frac{z_2}{\mu_2^*}\right)^{\frac{\gamma_1-\lambda_1}{\lambda_2}}-1}{\lambda_1-\gamma_1},
$$
Altogether,
$$
\frac{\partial u_2}{\partial k_2}=-1_{\widetilde\mu_2(z_1)<z_2<\mu_2^*}\frac{z_1}{1-z_1}\frac{1}{1-w_1}\frac{\gamma_1}{(\nu_1^*)^2}\frac{\left(\frac{z_2}{\mu_2^*}\right)^{\frac{\gamma_1-\lambda_1}{\lambda_2}}-1}{\lambda_1-\gamma_1},
$$
and it is negative if $\widetilde\mu_2(z_1)<z_2<\mu_2^*$.
\end{proof}

\begin{proof}[{\bf (v) Effects of $w_1$}]
The sign of the change due to $w_1$ is exactly the same as that due to $k_2$, because both $w_1$ and $k_2$ affect players' payoffs through $\nu_1^*$, and the changes in $\nu_1^*$ due to $w_1$ and $k_2$ are both negative.
\end{proof}

\subsubsection{Proof of Remark \ref{remark:gamma1} (effects of change in ultimatum opportunity arrival rate)}
\begin{proof}[{\bf Proof of Remark \ref{remark:gamma1}}]
Consider $\partial u_1/\partial \gamma_1$ first.
$$
\frac{\partial u_1}{\partial \gamma_1}=1_{z_1>\widetilde\mu_1(z_2)}\frac{z_2}{1-z_2}\frac{1}{\widetilde\mu_2^2(z_1)}\frac{\partial \widetilde\mu_2(z_1)}{\partial \gamma_1}D,
$$
which has the same sign as $\partial\log\widetilde\mu_2(z_1)/\partial \gamma_1$. For $\mu_1^N\leqslant z_1<1$,
$$
\log\widetilde\mu_2(z_1)=\frac{\lambda_2}{\gamma_1-\lambda_1}\log\left[\left(1-\frac{\gamma_{1}}{\lambda_{1}}\right)\frac{1}{z_{1}}+\frac{\gamma_{1}}{\lambda_{1}}\right],
$$
we have
\begin{eqnarray*}
\frac{\partial\log\widetilde{\mu}_{2}(z_{1})}{\partial\gamma_{1}} & = & -\frac{\lambda_{2}}{(\gamma_{1}-\lambda_{1})^{2}}\log\left[\left(1-\frac{\gamma_{1}}{\lambda_{1}}\right)\frac{1}{z_{1}}+\frac{\gamma_{1}}{\lambda_{1}}\right]+\frac{\lambda_{2}}{\gamma_{1}-\lambda_{1}}\frac{\left(1-\frac{1}{z_{1}}\right)\frac{1}{\lambda_{1}}}{(1-\frac{\gamma_{1}}{\lambda_{1}})\frac{1}{z_{1}}+\frac{\gamma_{1}}{\lambda_{1}}}\\
 & = & \frac{1}{\lambda_{1}-\gamma_{1}}\log\left[\widetilde{\mu}_{2}(z_{1})\right]+\frac{\lambda_{2}}{\lambda_{1}-\gamma_{1}}\left(\frac{1}{z_{1}}-1\right)\frac{\lambda_{2}}{\lambda_{1}}\left[\widetilde{\mu}_{2}(z_{1})\right]^{\frac{\lambda_{1}-\gamma_{1}}{\lambda_{2}}}\\
 & = & \frac{\lambda_{2}}{(\lambda_{1}-\gamma_{1})^{2}}\log\left[\widetilde{\mu}_{2}(z_{1})\right]^{\frac{\lambda_{1}-\gamma_{1}}{\lambda_{2}}}+\frac{\lambda_{2}}{\lambda_{1}-\gamma_{1}}\left(\frac{1}{z_{1}}-1\right)\frac{1}{\lambda_{1}}\left[\widetilde{\mu}_{2}(z_{1})\right]^{\frac{\lambda_{1}-\gamma_{1}}{\lambda_{2}}}.
\end{eqnarray*}
Since $z_{1}=\widetilde{\mu}_{1}\left(\widetilde{\mu}_{2}(z_{1})\right)={[\lambda_{1}-\gamma_{1}]}/{\left[\lambda_{1}\left[\widetilde{\mu}_{2}(z_{1})\right]^{\frac{\gamma_{1}-\lambda_{1}}{\lambda_{2}}}-\gamma_{1}\right]}$,
\[
\frac{1}{z_{1}}-1=\frac{\lambda_{1}\left[\widetilde{\mu}_{2}(z_{1})\right]^{\frac{\gamma_{1}-\lambda_{1}}{\lambda_{2}}}-\gamma_{1}}{\lambda_{1}-\gamma_{1}}-\frac{\lambda_{1}-\gamma_{1}}{\lambda_{1}-\gamma_{1}}=\lambda_{1}\frac{\left[\widetilde{\mu}_{2}(z_{1})\right]^{\frac{\gamma_{1}-\lambda_{1}}{\lambda_{2}}}-1}{\lambda_{1}-\gamma_{1}}.
\]
The expression becomes
\begin{eqnarray*}
\frac{\partial\log\widetilde{\mu}_{2}(z_{1})}{\partial\gamma_{1}} & = & \frac{\lambda_{2}}{(\lambda_{1}-\gamma_{1})^{2}}\log\left[\widetilde{\mu}_{2}(z_{1})\right]^{\frac{\lambda_{1}-\gamma_{1}}{\lambda_{2}}}+\frac{\lambda_{2}}{\lambda_{1}-\gamma_{1}}\left(\frac{1}{z_{1}}-1\right)\frac{1}{\lambda_{1}}\left[\widetilde{\mu}_{2}(z_{1})\right]^{\frac{\lambda_{1}-\gamma_{1}}{\lambda_{2}}}\\
 & = & \frac{\lambda_{2}}{(\lambda_{1}-\gamma_{1})^{2}}\log\left[\widetilde{\mu}_{2}(z_{1})\right]^{\frac{\lambda_{1}-\gamma_{1}}{\lambda_{2}}}+\frac{\lambda_{2}}{(\lambda_{1}-\gamma_{1})^{2}}\left[\left[\widetilde{\mu}_{2}(z_{1})\right]^{\frac{\gamma_{1}-\lambda_{1}}{\lambda_{2}}}-1\right]\left[\widetilde{\mu}_{2}(z_{1})\right]^{\frac{\lambda_{1}-\gamma_{1}}{\lambda_{2}}}\\
 & = & \frac{\lambda_{2}}{(\lambda_{1}-\gamma_{1})^{2}}\left[\log z+1-z\right],
\end{eqnarray*}
where $z\equiv\left[\widetilde{\mu}_{2}(z_{1})\right]^{\frac{\lambda_{1}-\gamma_{1}}{\lambda_{2}}}$.
Since $\log z+1-z$ attains its maximum value of $0$ at $z=1$ and
attains a value strictly less than $1$ for $z\neq1$, and $z\neq1$
as long as $\widetilde{\mu}_{2}(z_{1})\neq1$, we have $\partial\log\widetilde{\mu}_{2}(z_{1})/\partial\gamma_{1}<0$
for any $z_{1}\leqslant z_{1}<1$.

For $z_1\leqslant \mu_1^N$,
$$
\log \widetilde\mu_2(z_1)=\frac{\lambda_{2}}{\gamma_{1}-\lambda_{1}}\log\left[\dfrac{\frac{\lambda_{1}-\gamma_{1}}{\lambda_{1}}\frac{1}{z_{1}}+\frac{\gamma_{1}}{\lambda_{1}}\frac{1}{\nu_{1}^{*}}}{1+\frac{1-\nu_{1}^{*}}{\nu_{1}^{*}}\frac{\gamma_{1}}{\lambda_{1}}(\mu_{2}^{*})^{\frac{\lambda_{1}-\gamma_{1}}{\lambda_{2}}}}\right].
$$
Hence,
\begin{eqnarray*}
\frac{\partial\log\widetilde{\mu}_{2}(z_{1})}{\partial\gamma_{1}} & = & -\frac{\lambda_{2}}{(\gamma_{1}-\lambda_{1})^{2}}\log\left\{ \left[\widetilde{\mu}_{2}(z_{1})\right]^{\frac{\gamma_{1}-\lambda_{1}}{\lambda_{2}}}\right\} +\frac{\lambda_{2}}{\gamma_{1}-\lambda_{1}}\frac{1}{\left[\widetilde{\mu}_{2}(z_{1})\right]^{\frac{\gamma_{1}-\lambda_{1}}{\lambda_{2}}}}\times\\
 &  & \bigg\{\frac{\frac{1}{\lambda_{1}}\left(\frac{1}{\nu_{1}^{*}}-\frac{1}{z_{1}}\right)}{1+\frac{1-\nu_{1}^{*}}{\nu_{1}^{*}}\frac{\gamma_{1}}{\lambda_{1}}(\mu_{2}^{*})^{\frac{\lambda_{1}-\gamma_{1}}{\lambda_{2}}}}-\frac{\frac{\lambda_{1}-\gamma_{1}}{\lambda_{1}}\frac{1}{z_{1}}+\frac{\gamma_{1}}{\lambda_{1}}\frac{1}{\nu_{1}^{*}}}{\left[1+\frac{1-\nu_{1}^{*}}{\nu_{1}^{*}}\frac{\gamma_{1}}{\lambda_{1}}(\mu_{2}^{*})^{\frac{\lambda_{1}-\gamma_{1}}{\lambda_{2}}}\right]^{2}}\times\\
 &  & \left[\frac{1-\nu_{1}^{*}}{\nu_{1}^{*}}\frac{1}{\lambda_{1}}(\mu_{2}^{*})^{\frac{\lambda_{1}-\gamma_{1}}{\lambda_{2}}}-\frac{1-\nu_{1}^{*}}{\nu_{1}^{*}}\frac{\gamma_{1}}{\lambda_{1}}(\mu_{2}^{*})^{\frac{\lambda_{1}-\gamma_{1}}{\lambda_{2}}}\frac{1}{\lambda_{2}}\log(\mu_{2}^{*})\right]\bigg\},
\end{eqnarray*}
which simplifies to
\[
\frac{\lambda_{2}}{\gamma_{1}-\lambda_{1}}\left\{ -\frac{1}{\lambda_{2}}\log\left[\widetilde{\mu}_{2}(z_{1})\right]+\frac{\frac{1}{\nu_{1}^{*}}-\frac{1}{z_{1}}}{\frac{1}{z_{1}}+\gamma_{1}\left(\frac{1}{\nu_{1}^{*}}-\frac{1}{z_{1}}\right)}+\frac{1-\nu_{1}^{*}}{\nu_{1}^{*}}\frac{1}{\lambda_{1}}(\mu_{2}^{*})^{\frac{\lambda_{1}-\gamma_{1}}{\lambda_{2}}}\left[\frac{\gamma_{1}}{\lambda_{2}}\log(\mu_{2}^{*})-1\right]\right\} .
\]
\end{proof}

\subsection{Equilibrium existence and uniqueness with one-sided ultimatum and multiple demand types}\label{sec:1multiple}
Before proving Theorem \ref{thm:1multiple}, we prove a lemma that shows the uniqueness of equilibrium when player 1 has a single demand type and player 2 has multiple demand types.
\begin{lemma}\label{prop:11n}
For any game ($a_1$, $\pi_2$, $z_1$, $z_2$, $r_1$, $r_2$, $\gamma_1$, $c_1$, $k_2$, $w_1$) with ultimatum opportunities for player 1, a single demand for player 1, and multiple demands for player 2, there exists a unique equilibrium.
\end{lemma}
\begin{proof}[{\bf Proof of Lemma \ref{prop:11n}}]
Denote by $\sigma_2(\cdot)$, a probability distribution over $A_2\cup \{Q\}$, a mimicking strategy of a strategic player 2. Since mimicking $a_2<1-a_1$ is never optimal and mimicking $a_2=1-a_1$ is equivalent to conceding, we assume that in equilibrium $\sigma_2(a_2)=0$ for all $a_2\leq 1-a_1$. If $x=1$, then in equilibrium $\sigma_2(Q)=1$, because unjustified player 2 will not delay conceding if she knows that player 1 is justified. For the remainder of the proof we assume $x<1$.

Define $T_i(a_1,a_2,x)$ as the time it takes for player $i$'s reputation to increase from $x$ to $1$ on the equilibrium reputation path when each player $i$'s demand is $a_i$. Explicitly,
$$T_1(a_1,a_2,x):=
\begin{cases}
\infty & x\leq \left(1-\frac{\lambda_1}{\gamma_1}\right)\nu_1^*, \\
t(\mu_1^N;x,\lambda_1-\gamma_1,\frac{\gamma_1}{\nu_1^*})+t(1;\mu_1^N,\lambda_1-\gamma_1,{\gamma_1}) & \left(1-\frac{\lambda_1}{\gamma_1}\right)\nu_1^*<x < \mu_1^N, \\
t(1;x,\lambda_1-\gamma_1,{\gamma_1}) & \mu_1^N \leq x \leq 1,
\end{cases}
$$
that is,
$$
T_1(a_1,a_2,x):=
\begin{cases}
\infty & \text{if } x \leq \left(1-\frac{\lambda_1}{\gamma_1}\right)\nu_1^*, \\
\frac{1}{\lambda_1-\gamma_1}\log\left[\frac{\frac{\lambda_1-\gamma_1}{x}+\frac{\gamma_1}{\nu_1^*}}{\frac{\lambda_1-\gamma_1}{\mu_1^N}+\frac{\gamma_1}{\nu_1^*}}\right]-\frac{1}{\lambda_2}\log \mu_2^* & \text{if } \left(1-\frac{\lambda_1}{\gamma_1}\right)\nu_1^*<x < \mu_1^N, \\
\frac{1}{\lambda_1-\gamma_1}\log\left[\frac{\frac{\lambda_1-\gamma_1}{x}+\gamma_1}{\lambda_1}\right] & \text{if }\mu_1^N \leq x\leq 1,
\end{cases}
$$
and
$$T_2(a_1,a_2,y):= -\frac{a_1+a_2-1}{r_1(1-a_2)}\log y.$$
Note that $T_1(a_1,a_2,x)$ is continuous and strictly decreasing in $x$ on $(1-\frac{\lambda_1}{\gamma_1},1)$ and that $T_2(a_1,a_2,y)$ is continuous and strictly decreasing in $y$ on $(0,1)$.

It remains to be shown that an unjustified player 2's equilibrium behavior $\sigma_2(\cdot)$ and an unjustified player 1's conceding behavior $Q_1(a_1,a_2,x,\sigma_2)$ at time zero are uniquely determined. Subsequently, we provide a series of definitions and use them to prove a series of claims that lead to equilibrium existence and uniqueness. Define player 2's reputation at time 0 when she plays $a_2$ with probability $\sigma_2$ as
 $$
 y^*(a_2,\sigma_2)=\frac{z_2\pi_2(a_2)}{z_2\pi_2(a_2)+(1-z_2)\sigma_2}.
 $$
Note that the more likely an unjustified player 2 announces a particular demand $a_2$, the more likely she is believed to be unjustified, and the lower her payoff from demanding $a_2$ is.

Let $\overline \sigma_2(a_1,a_2,x)$ be the maximum probability player 2 plays $a_2$ in equilibrium so that the expected payoff from demanding $a_2$ is higher than directly conceding to player 1's demand. For any $a_2<1-a_1$, $\overline \sigma_2(a_1,a_2,x)=0$ because conceding to player 1's demand $a_1$\threeemdashes{}which results in a payoff of $1-a_1$\threeemdashes{}is a strictly better strategy than demanding strictly less than $1-a_1$ and a weakly better strategy than demanding $1-a_1$. For any $a_2>1-a_1$, after choosing $a_2$, in any equilibrium, player 2 should not concede with a positive probability at time 0. First, if player 1's reputation can reach 1 without conceding with a positive probability at time 0 and player 2's reputation reaches 1 slower than player 1 when she demands $a_2$ with probability 1, $\overline \sigma_2(a_1,a_2,x)$ is the unique solution of $\sigma_2$ to $T_1(a_1,a_2,x)=T_2(a_1,a_2,y^*(a_2,\sigma_2))$ so that the two players' reputations reach 1 at the same time. Explicitly, when we let $\psi:= 1-\frac{\gamma_1}{\lambda_1}$, and we have
$$
\overline \sigma_2(a_1,a_2,x):=
\begin{cases}
1 & x\leq \left(1-\frac{\lambda_1}{\gamma_1}\right)\nu_1^*, \\
z_2 \frac{\pi_2(a_2)}{1-\pi_2(a_2)} \left[ \left( \frac{\psi \frac{1}{x} + (1-\psi) \frac{1}{\nu_1^*}}{\psi \frac{1}{\mu_1^N} + (1-\psi) \frac{1}{\nu_1^*}} \right) ^ {\frac{\lambda_2}{\lambda_1}\frac{1}{\psi}} -1 \right] & \left(1-\frac{\lambda_1}{\gamma_1}\right)\nu_1^* < x < \mu_1^N, \\
z_2 \frac{\pi_2(a_2)}{1-\pi_2(a_2)} \left[ \left( \psi\frac{1}{x} + (1-\psi) \right)^{\frac{\lambda_2}{\lambda_1}\frac{1}{\psi}} -1 \right] & \mu_1^N\leq x < 1.
\end{cases}
$$
Note that in equilibrium $\sigma_2(a_2)\leq \overline\sigma_2(a_1,a_2,\sigma_2)$ for all $a_2>1-a_1$. To see why this claim must hold, suppose player 2 mimics $a_2$ with a probability strictly higher than $\overline\sigma_2(a_1,a_2,\sigma_2)<1$. Then player 2 must concede with a strictly positive probability at time zero in order for players' reputations to reach 1 at the same time. However, we have specified that player 2 does not concede at time zero after announcing her demand. Second, if player 1's reputation reaches 1 even slower than when player 2 demands $a_2$ with probability 1, $\overline\sigma_2(a_1,a_2,x)=1$. The scenario happens whenever $T_1(a_1,a_2,x)>T_2(a_1,a_2,y^*(a_2,1))$. In particular, it happens whenever $x<\mu_1^*(1-\frac{\lambda_1}{\gamma_1})$. In summary, in any equilibrium, $\sigma_2(a_2)\leq \overline\sigma_2(a_1,a_2,x)$, where $\overline\sigma_2(a_1,a_2,x)=0$ if $a_2\leq 1-a_1$; $\overline\sigma_2(a_1,a_2,x)$ is the unique solution of $\sigma_2$ in $T_1(a_1,a_2,x)=T_2(a_1,a_2,y^*(a_2,\sigma_2))$ if $a_2>1-a_1$ and $T_1(a_1,a_2,x)<T_2(a_1,a_2,y^*(a_2,1))$; and $\overline\sigma_2(a_1,a_2,x)=1$ if $a_2>1-a_1$ and $T_1(a_1,a_2,x)\geq T_2(a_1,a_2,y^*(a_2,1))$.

When player 2 demands $a_2$ with probability $\sigma_2 \leq \overline\sigma_2(a_1,a_2,x)$, player 1 must raise his time 0 reputation to $x^*(a_1,a_2,\sigma_2)$ so that their reputations reach 1 at the same time:
$$T_1(a_1,a_2,x^*(a_1,a_2,\sigma_2))=T_2(a_1,a_2,y^*(a_2,\sigma_2)).$$
In order to do so, an unjustified player 1 concedes with probability
$$Q_1(a_1,a_2,x,\sigma_2)=1-\frac{x}{1-x}\frac{1-x^*(a_1,a_2,\sigma_2)}{x^*(a_1,a_2,\sigma_2)},$$
so that player 1's reputation is raised to
$$x^*(a_1,a_2,\sigma_2)=\frac{x}{x+(1-x)[1-Q_1(a_1,a_2,x,\sigma_2)]}.$$
Explicitly,
$$
x^*(a_1,a_2,\sigma_2)
:=
\begin{cases}
\dfrac{\lambda_1-\gamma_1}{\left(\frac{\mu_2^*}{y}\right)^{\frac{\lambda_1-\gamma_1}{\lambda_2}}\left(\frac{\lambda_1-\gamma_1}{\mu_1^N}+\frac{\gamma_1}{\nu_1^*}\right)-\frac{\gamma_1}{\nu_1^*}} & \text{if } y^*(a_2,\sigma_2)\leq \mu_2^*\\
\dfrac{1-\frac{\gamma_1}{\lambda_1}}{\left(\frac{1}{y}\right)^{\frac{\lambda_1-\gamma_1}{\lambda_2}}-\frac{\gamma_1}{\lambda_1}} & \text{if } y^*(a_2,\sigma_2)>\mu_2^*
\end{cases}.
$$

When player 2 demands $a_2$ with probability $\sigma_2$ and an unjustified player 1 concedes with probability $Q_1(a_1,a_2,x,\sigma_2)$, an unjustified player 2's expected payoff is
$$
 u_2^*(a_1,a_2,x,\sigma_2)=1-a_1+(1-x)Q_1(a_1,a_2,x,\sigma_2)(a_1+a_2-1).
$$

Two additional properties restrict player 2's equilibrium strategy $\sigma_2(\cdot)$. First, for any $a_2$ and $a_2'>a_2$, if $\sigma_2(a_2)>0$, then $\sigma_2(a_2')>0$. We can prove this property by contradiction. Suppose $\sigma_2(a_2)>0$ and $\sigma_2(a_2')=0$. Because $\sigma_2(a_2')=0$, $u_2^*(a_1,a_2',x,\sigma_2(a_2'))=1-a_1+(1-x)(a_1+a_2'-1)$. Because $\sigma_2(a_2)>0$, $u_2^*(a_1,a_2,x,\sigma_2(a_2))=1-a_1+(1-x)Q_1(a_1,a_2,x,\sigma_2(a_2))(a_1+a_2-1)\leq 1-a_1+(1-x)(a_1+a_2'-1)=u_2^*(a_1,a_2',x,\sigma_2(a_2'))$. Second, whenever $\sum_{a_2} \overline\sigma_2(a_1,a_2,x)\leq 1$, $\sigma_2(a_2)=\overline\sigma_2(a_1,a_2,x)$ for all $a_2$, and $Q_2=1-\sum_{a_2}\overline\sigma_2(a_1,a_2,x)$. The two properties together imply that we only need to check first if $\sum_{a_2}\overline\sigma_2(a_1,a_2,x)\leq 1$, and, if the first condition does not hold, then we find the equilibrium strategy among the set of strategies $\sigma_2(\cdot)$ such that $\sigma_2(a_2')>0$ for all $a_2'\geq a_2$, for each $a_2\in A_2$.

Denote by
$$
\Delta_2(a_1,x):= \left\{\sigma_2(\cdot)\in \Delta \left| \begin{array}{ll}
\sigma_2(a_2)=0 & \forall a_2\leq 1-a_1\\
\sigma_2(a_2)\leq \overline \sigma_2(a_1,a_2,x) & \forall a_2>1-a_1
\end{array} \right. \right\}
$$
the set of candidate equilibrium mimicking strategies of player 2 in the game $B_1(a_1,x)$, where $\Delta$ denotes the set of all probability distributions on $A_2\cup\{Q\}$. Note that the set $\Delta_2(a_1,x)$ is nonempty, convex, and compact. For any candidate equilibrium mimicking strategy $\sigma_2(\cdot) \in \Delta_2(a_1,x)$, define
$$
\widehat u_2(x, \sigma_2(\cdot)) := \min_{a_2: \sigma_2(a_2)>0} u_2^*(a_1,a_2,x,\sigma_2(a_2)).
$$
Explicitly,
$$
\widehat u_2(x,\sigma_2(\cdot)):=
\begin{cases}
\displaystyle \min_{a_2: \sigma_2(a_2)>0} u_2^*(a_1, a_2, x, \sigma_2(a_2)) &\text{if } \sigma_2(Q)=0\\
1-a_1 &\text{if } \sigma_2(Q)\neq 0
\end{cases}.
$$
Note that $\widehat\sigma_2(\cdot)$ is an equilibrium strategy if and only if $\widehat\sigma_2(\cdot)$ solves $\max_{\sigma_2(\cdot) \in\Delta_2(a_1,x)}\widehat u_2(x,\sigma_2(\cdot))$. 
($\Rightarrow$) Suppose $\widehat\sigma_2(\cdot)$ is an equilibrium strategy. Any equilibrium strategy $\sigma_2(\cdot)$ satisfies that for all $a_2\in A_2\cup \{Q\}$ such that $\sigma_2(a_2)>0$, $u_2^*(a_1,a_2,x,\sigma_2(a_2))$ is the same. If $\sigma_2(Q)>0$, then $$u_2^*(a_1,a_2,x,\sigma_2(a_2))=1-a_1;$$ if $\sigma_2(Q)=0$, then $$u_2^*(a_1,a_2,x,\sigma_2(a_2))=\min_{a_2:\widehat\sigma_2(a_2)>0} u_2^*(a_1,a_2,x,\widehat\sigma_2(a_2)).$$ Hence, any  equilibrium strategy $\widehat\sigma_2(\cdot)$ must generate an equilibrium utility of $\widehat u_2(x,\sigma_2(\cdot))$. Hence, $\widehat\sigma_2(\cdot)$ maximizes $\widehat u_2(x,\sigma_2(\cdot))$ among all candidate equilibrium strategies $\sigma_2(\cdot)$.
($\Leftarrow$) Suppose $\widehat\sigma_2(\cdot)$ solves $\max_{\sigma_2(\cdot) \in\Delta_2(a_1,x)}\widehat u_2(x,\sigma_2(\cdot))$. By the strict monotonicity of $u_2^*(a_1,a_2,x,\cdot)$, for all $a_2\in A_2$ such that $\widehat\sigma_2(a_2)>0$, $u_2^*(a_1,a_2,x,\widehat\sigma_2(a_2))={\widehat u}_2(x,\widehat\sigma_2(\cdot))$. Coupled with the fact that $\widehat\sigma_2(\cdot)$ is the feasible strategy that maximizes $\widehat u_2(x,\sigma_2(\cdot))$, $\widehat\sigma_2(\cdot)$ is an equilibrium strategy.

Define $\Gamma ( \sigma_2(\cdot) )$, a correspondence from $\Delta_2(a_1,x)$ to $\Delta_2(a_1,x)$, as follows:
$$
\{
\widetilde \sigma_2(\cdot) \in \Delta_2(a_1,x) |
\widetilde \sigma_2(a_2)>0 \Rightarrow u_2^*(a_1,a_2,x,\sigma_2(a_2))\geq u_2^*(a_1,a_2', x, \sigma_2(a_2')) \; \forall a_2'\in A_2
\}.
$$
Note that $\widehat\sigma_2(\cdot)$ solves $\max_{\sigma_2(\cdot)\in\Delta_2(a_1,x)}\widehat u_2(x,\sigma_2(\cdot))$ if and only if $\widehat\sigma_2(\cdot)$ is a fixed point of $\Gamma$. ($\Rightarrow$) Suppose $\widehat\sigma_2(\cdot)$ solves $\max_{\sigma_2(\cdot)\in\Delta_2(a_1,x)}\widehat u_2(x,\sigma_2(\cdot))$. By the argument above, $\widehat\sigma_2(\cdot)$ is an equilibrium strategy. Therefore, $\widehat\sigma_2(a_2)>0$ implies $u_2^*(a_1,a_2,x,\widehat\sigma_2(a_2))\geq u_2^*(a_1,a_2',x,\widehat\sigma_2(a_2'))$ for any $a_2'\in A_2$. By the definition of $\Gamma$, $\widehat\sigma_2(\cdot) \in \Gamma (\widehat\sigma_2(\cdot))$. ($\Leftarrow$) Suppose $\widehat\sigma_2(\cdot) \in \Gamma(\widehat\sigma_2(\cdot))$. By the definition of $\Gamma$, $\widehat\sigma_2(a_2)>0$ implies $u_2^*(a_1,a_2,x,\widehat\sigma_2(a_2))\geq u_2^*(a_1,a_2',x,\widehat\sigma_2(a_2'))$ for any $a_2'\in A_2$. Assume by contradiction that $\widehat\sigma_2(\cdot)$ does not solve $\max_{\sigma_2(\cdot)\in \Delta_2(a_1,x)} \widehat u_2(x,\sigma_2(\cdot))$ but $\widetilde \sigma_2(\cdot)\neq \widehat \sigma_2(\cdot)$ does. There must exist an $a_2\in A_2$ such that $\widehat\sigma_2(a_2)>0$ and $\widetilde \sigma_2(a_2)<\widehat \sigma_2 (a_2)$ (otherwise, if $\widetilde \sigma_2(a_2)\geq \widehat \sigma_2(a_2)$ for all $a_2$ such that $\widehat\sigma_2(a_2)>0$, then by the strict monotonicity of $u_2^*$, $u_2^*(a_1,a_2,x,\widetilde\sigma_2(a_2))\leq u_2^*(a_1,a_2,x,\widehat \sigma_2(a_2))$, and $\widehat u_2(x,\widetilde\sigma_2(\cdot))\leq \widehat u_2(x,\widehat \sigma_2(\cdot))$). However, that implies that there exists $a_2'\in A_2\cup \{Q\}$ such that $\widetilde \sigma(a_2')>\widetilde \sigma(a_2')$. If $a_2'=Q$, then  $\widehat u_2(x,\widetilde\sigma_2(\cdot))\leq \widehat u_2(x,\widehat \sigma_2(\cdot))$. If $a_2'\in A_2$, then $\widehat u_2(x,\widetilde\sigma_2(\cdot))\leq \widehat u_2(x,\widehat \sigma_2(\cdot))$.

Hence, from the two claims above, we have that $\widehat\sigma_2(\cdot)$ is an equilibrium strategy for player 2 in the game $B_0(a_1,x)$ if and only if $\widehat\sigma_2(\cdot)$ is a fixed point of $\Gamma$. Equilibrium existence follows from the existence of a fixed point of $\Gamma$ by Kakutani's fixed point theorem. By construction, $\Delta_2(a_1,x)$ is compact. By construction, $\Gamma$ is convex-valued. Finally, $\Gamma$ is upper-hemicontinuous because $u_2^*$ is continuous in its last argument.

It remains to show the existence of a unique equilibrium. Equilibrium uniqueness follows from the strict monotonicity of $u_2^*$ in $x$. Suppose there are two equilibrium strategies $\widehat \sigma_2(\cdot)$ and $\widetilde \sigma_2(\cdot)$; without loss of generality, suppose $\widehat\sigma_2(a_2)>\widetilde \sigma_2(a_2)>0$ for some $a_2>1-a_1$. The utilities of playing the two strategies are different:
$$\widehat u_2(x,\widehat\sigma_2(\cdot)) = u_2^*(a_1,a_2,x,\widehat\sigma_2(a_2)) < u_2^*(a_1,a_2,x,\widetilde \sigma_2(a_2)) = \widehat u_2(x,\widetilde\sigma_2(\cdot)),$$ where the strict inequality follows from the strict monotonicity of $u_2^*$. This contradicts the property that equilibrium strategies $\widehat \sigma_2(\cdot)$ and $\widetilde \sigma_2(\cdot)$ both maximize $\widehat u_2(x,\sigma_2(\cdot))$. Multiple equilibrium distributions over types being conceded to are in conflict with the requirement that types mimicked with a positive probability must have equal payoffs that are not smaller than the payoffs of the types that are not mimicked. Suppose by contradiction there are two different equilibrium strategies for player 2: $\sigma_2(a_2)\neq \sigma_2'(a_2)$ for some $a_2$. If $\sigma_2(a_2)>0$ and $\sigma_2'(a_2)>0$, then $u_2(a_1, a_2, x, \sigma_2(a_2) ) \neq u_2(a_1, a_2, x, \sigma_2'(a_2))$. But $u_2(a_1,a_2,x,\sigma_2(a_2))=\widehat u_2(x,\sigma_2(\cdot))$ and $u_2(a_1,a_2,x,\sigma_2(a_2))=\widehat u_2(x,\sigma_2'(\cdot))$, but $\widehat u_2(x,\sigma_2(\cdot))\neq \widehat u_2(x,\sigma_2'(\cdot))$ contradicts the fact that $\sigma_2(\cdot)$ and $\sigma_2'(\cdot)$ both solve $\max_{\sigma_2(\cdot) \in \Delta_2(a_1, x)} \widehat u_2(x,\sigma(\cdot))$. If $\sigma_2(a_2)$ or $\sigma_2'(a_2)$ is zero, then by the first additional property of player 2's equilibrium strategy above, there is an $a_2'>a_2$ such that $\sigma_2(a_2')>0$, $\sigma_2'(a_2')>0$, and $\sigma_2(a_2')\neq \sigma_2'(a_2')$, so that the contradiction arises again. Player 1 receives $u_1(a_1, x)$ in the equilibrium of the bargaining game $B(a_1,x)$.
\end{proof}

\begin{proof}[\bf Proof of Theorem \ref{thm:1multiple}]
Denote by $u_1(a_1,x)$ the payoff of player 1 in the unique equilibrium of the bargaining game $B_1(a_1,x)$ with $A_1=\{a_1\}$ and $|A_2|\geq 1$. Note that it is a continuous function of $x$. Moreover, there exists an $\underline x$ such that $u_1^*(a_1,x)=u_1^*(a_1,\underline x)$ for any $x\leq \underline x$ and $u_1^*(a_1, x)$ is strictly increasing in $x$ on the interval $\left(\underline x,1\right)$.

We characterize the equilibrium distribution $\sigma_{1}$ as the solution
to $$\max_{\sigma_{1}}\,\widehat{u}(\sigma_{1}),$$ where
\[
\widehat{u}(\sigma_{1})=\min_{a_{1}\text{ s.t. }\sigma_{1}(a_{1})>0}u_{1}(a_{1},x(\sigma_{1}(a_{1}))),
\]
and
\[
x(\sigma_{1}(a_{1}))=\frac{z_{1}\pi_{1}(a_{1})}{z_{1}\pi_{1}(a_{1})+(1-z_{1})\sigma_{1}(a_{1})}.
\]
The continuity of $u_{1}(a_{1},x)$ in $x$ ensures that an equilibrium
exists; see the fixed-point argument establishing the existence of
an equilibrium strategy $\sigma_{2}$ in $B(a_{1},x)$ above.

Let $\overline{u}_{1}$ be the maximized value above; $\overline{u}_{1}$
is the utility that player 1 attains in any equilibrium. Clearly,
$\overline{u}_{1}\geq u_{1}(a_{1},\underline{x})$ for all $a_{1}$.
Let $\sigma_{1}$ and $\widehat{\sigma}_{1}$ be two equilibrium strategies
for player 1.

Claim: If $\overline{u}_{1}>u_{1}(a_{1},\underline{x})$, then $\widetilde{\sigma}_{1}(a_{1})=\widehat{\sigma}_{1}(a_{1})$.
Proof: To see this note that either $u_{1}(a_{1},1)>\overline{u}_{1}$
or $u_{1}(a_{1},1)\leq \overline{u}_{1}$. If $u_{1}(a_{1},1)>\overline{u}_{1}$,
then there is a unique $\sigma_{1}$ such that $\sigma_{1}(a_{1},x(\sigma_{1}))=\overline{u}_{1}$,
and hence $\widetilde{\sigma}_{1}(a_{1})=\widehat{\sigma}_{1}(a_{1})=\sigma_{1}$.
If $u_{1}(a_{1},1)\leq \overline{u}_{1}$, then by the strict monotonicity
of $u_{1}(a_{1},x)$ in $x$ for $x>\underline{x}$ and monotonicity
of $u_{1}(a_{1},x)$ in $x$ for $x\leq \underline{x}$, $u_{1}(a_{1},x)<u_{1}(a_{1},1)\leq \overline{u}_{1}$
for any $x<1$. Hence, $\widetilde{\sigma}_{1}(a_{1})=\widehat{\sigma}_{1}(a_{1})=0$.

Define $D_{1}=\{a_{1}\in A_{1}|u_{1}(a_{1},\underline{x})=\overline{u}_{1}\}$.
Recall that $\underline{x}$ depends on $a_{1}$. We have already
noted that $\widetilde{\sigma}_{1}(a_{1})=\widehat{\sigma}_{1}(a_{1})$
for $a_{1}\in A_{1}\backslash D_{1}$. Hence, $\sum_{a_{1}\in D_{1}}\widetilde{\sigma}_{1}(a_{1})=\sum_{a_{1}\in D_{1}}\widehat{\sigma}_{1}(a_{1})$.

We will conclude the proof that $\widetilde{\sigma}_{1}$ and $\widehat{\sigma}_{1}$
lead to the same random outcome $\widetilde{\theta}$ by first verifying
that the probability that player 1 chooses $a_{1}\in D_{1}$ and agreement
is reached at time 0 is the same with either $\widetilde{\sigma}_{1}$
or $\widehat{\sigma}_{1}$. This will imply that the random outcome,
conditional on agreement at time 0, is the same with either $\widetilde{\sigma}_{1}$
or $\widehat{\sigma}_{1}$. Finally, we show that for each $a_{1}\in D_{1}$,
the probability that a strategic player 1 will mimic $a_{1}$ and
not concede is the same with either $\widetilde{\sigma}_{1}$ or $\widehat{\sigma}_{1}$.

Let $A(\sigma_{1})$ denote the probability that player 1 mimics some
$a_{1}\in D_{1}$ and agreement is reached at time 0 given the equilibrium
strategy $\sigma_{1}$. Since $a_{1}\in D_{1}$ implies $\sigma_{2}(\overline{a}_{2}|a_{1})=1$,
it follows that $a_{1}\ge1-\overline{a}_{2}$; otherwise, player 1
would achieve a higher utility by mimicking $\max C_{1}>1-\overline{a}_{2}$.
Hence,
\begin{align*}
A(\sigma_{1}) & =\sum_{a_{1}\in D_{1}}q_{1}(a_{1},\overline{a}_{2},x(\sigma_{1}(a_{1})),1)\left[1-x(\sigma_{1}(a_{1}))\right]\left[z_{1}\pi_{1}(a_{1})+(1-z_{1})\sigma_{1}(a_{1})\right]\\
 & =\sum_{a_{1}\in D_{1}}\frac{K(a_{1},\overline{a}_{2},1)-x(\sigma_{1}(a_{1}))}{K(a_{1},\overline{a}_{2},1)}\left[z_{1}\pi_{1}(a_{1})+(1-z_{1})\sigma_{1}(a_{1})\right]\\
 & =\sum_{a_{1}\in D_{1}}\left[z_{1}\pi_{1}(a_{1})+(1-z_{1})\sigma_{1}(a_{1})\right]-\sum_{a_{1}\in D_{1}}\frac{x(\sigma_{1}(a_{1}))}{K(a_{1},\overline{a}_{2},1)}\left[z_{1}\pi_{1}(a_{1})+(1-z_{1})\sigma_{1}(a_{1})\right]\\
 & =\sum_{a_{1}\in D_{1}}(1-z_{1})\sigma_{1}(a_{1})+\sum_{a_{1}\in D_{1}}z_{1}\pi_{1}(a_{1})-\sum_{a_{1}\in D_{1}}\frac{z_{1}\pi_{1}(a_{1})}{K(a_{1},\overline{a}_{2},1)}.
\end{align*}
Since $\sum_{a_{1}\in D_{1}}\widetilde{\sigma}_{1}(a_{1})=\sum_{a_{1}\in D_{1}}\widehat{\sigma}_{1}(a_{1})$,
we have $A(\widetilde{\sigma}_{1})=A(\widehat{\sigma}_{1})$. For
any $a_{1}\in D_{1}$, the probability that a strategic player 1 will
mimic $a_{1}$ and not concede at time 0 is
\begin{align*}
 & \sigma_{1}(a_{1})\left[1-q_{1}(a_{1},\overline{a}_{2},x(\sigma_{1}(a_{1})),1)\right]\\
= & \sigma_{1}(a_{1})\frac{x(\sigma_{1}(a_{1}))}{1-x(\sigma_{1}(a_{1}))}\frac{1-K(a_{1},\overline{a}_{2},1)}{K(a_{1},\overline{a}_{2},1)}\\
= & \pi_{1}(a_{1})\frac{z_{1}}{1-z_{1}}\frac{1-K(a_{1},\overline{a}_{2},1)}{K(a_{1},\overline{a}_{2},1)},
\end{align*}
which is independent of $\sigma_{1}$. Hence, $\widetilde{\sigma}_{1}(a_{1})$ and $\widehat{\sigma}_{1}(a_{1})$, the equilibrium probabilities that a strategic player 1 will mimic $a_{1}$, are the same.
\end{proof}

\section{Evidence: MLB and NHL salary arbitration}\label{sec:data}
\subsection{Major League Baseball}
We consider the process of negotiation preceding Major League Baseball salary arbitration.\footnote{We thank Nicholas Butler, Mengdongxue Han, and Ethan Pritchard for manually collecting the data.} We consider this application because we can obtain from publicly available reports of (i) the initial offers of the two parties (player and team), (ii) the time the challenge opportunity becomes credible (filing for arbitration and scheduling court date), (iii) the time the negotiation ends (signing  the contract), and (iv) the outcome (the terms of the contract). In contrast, the initial proposals, the duration, and final outcome of the arbitration over economic disputes are often confidential.

From mid-January to mid-February each year, players with a defined amount of service time (i.e., number of years playing at the MLB level) will enter into the salary arbitration process with their teams where the player and the team will present their case to have the player's salary set by a neutral third party arbitration panel for the upcoming season by final-offer arbitration.

A team has the contractual rights to a player until that player has six years of service time and becomes a free agent. During the first three years of service a player will typically make around the major league minimum salary. Players with between three and six years of service time and high-caliber players with two years of service time become eligible for salary arbitration if they do not already have a contract with their team for the next season by mid-January.\footnote{A high-caliber second-year player\threeemdashes{}the so-called Super 2\threeemdashes{}is a player who has between two and three years of service time, and has at least 86 days of service time during the second year and ranks in the top 22 percent of players who fall into that classification. A Super 2 player will have three years as a pre-arbitration eligible player and four arbitration years while a player who doesn't earn Super 2 status will have three years of salary arbitration following their four pre-arbitration years \citep{Sievert2018}.}

A player eligible for salary arbitration has to file by a prespecified date mid-January. Once the player files, the player and team will exchange salary offers by January 16. Because only the player can file a salary arbitration, the player is thought to be the side that has the challenge opportunity. At this point, the player and team can still have the opportunity to come to an agreement on a specific figure for the upcoming season prior to the hearing. If the player and team are unable to come to an agreement prior to the scheduled hearing, the player's salary will be determined by the arbitration panel. These hearings occur around mid-February.

During the salary arbitration hearing, both the player and the team will present their case to the arbitration panel. Following the hearing, the panel will choose between the player's and the team's salary offer. The information the two sides can use during the hearing to present their case includes the player's contribution to the team during the past season (e.g., the player's on-field performance and other qualities such as leadership and fan appeal), the length and consistency of the player's career contributions, the player's past compensation, the existence of any physical or mental defects, the team's recent performance (e.g., the team's record, improvement and attendance) and comparative baseball salaries.

The panel gives the most weight to each side's presentation of comparable baseball salaries. Here, the player and team can only compare the contracts of players whose service time does not exceed one annual service group above the player's service group. For example, a starting pitcher who enters the second year of salary arbitration would be compared other starting pitchers who are entering their second and third year of salary arbitration.

Information the panel cannot consider during the hearing includes the financial position of the team or player, testimonials or press comments regarding the team's or player's performance, prior contractual negotiations between the team and player, any costs associated with the salary arbitration process (i.e., attorney's fees), and salaries in other sports or occupations.

\begin{figure}[t]
\centering
\includegraphics[width=0.5\textwidth]{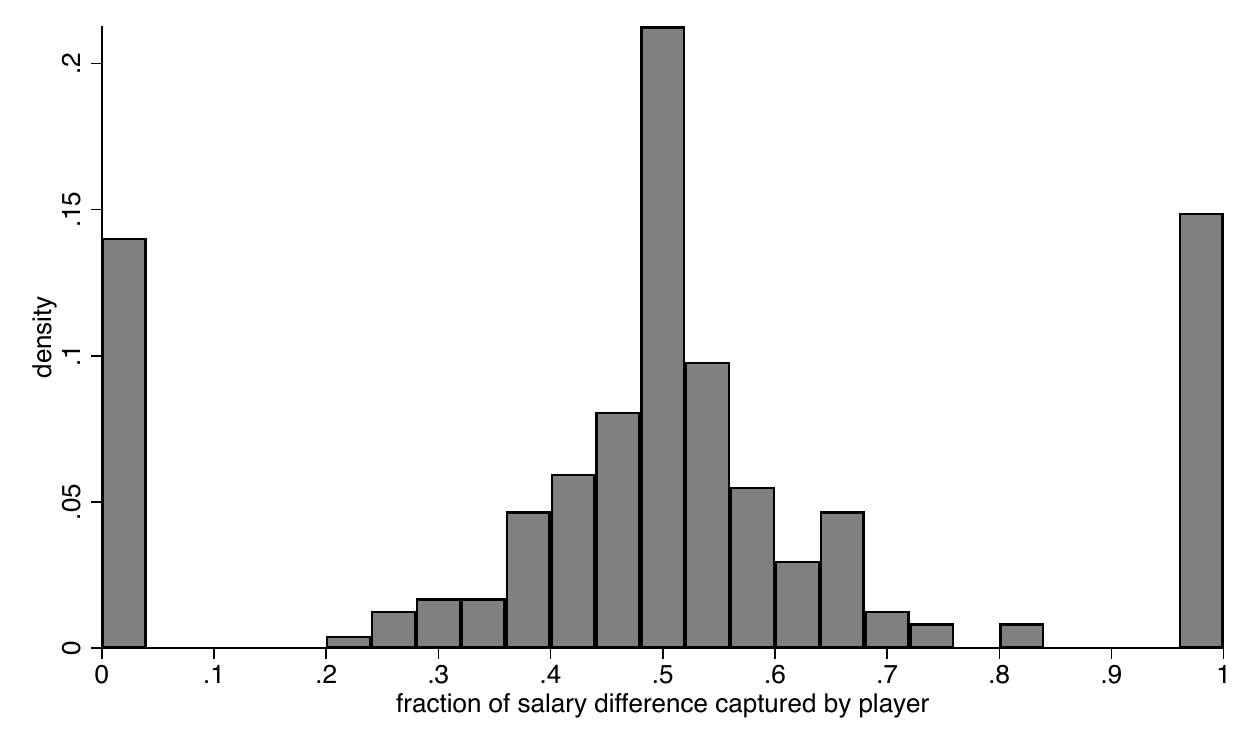}
\caption{\label{fig:outcome}Distribution of fraction of salary difference captured by player.}
\end{figure}

We collect all 292 cases in which the player has filed a salary arbitration from 2011 to 2020.\footnote{Salary arbitration has been in effect since the 1970s.} On average, these players have 3.6 years of service time, players' initial offers are 4.75 million USD, and teams' initial offers are 3.66 million USD, so their disputes are on average a little above a million dollars. On average, players' initial offers are 34.6\% higher than teams' initial offers, and the final settled amounts are 16.7\% higher than teams' initial offers. Overall 22.9\% of the cases\threeemdashes{}13.8\% in 2011-2016 and 64.3\% in 2017 and 2018\threeemdashes{}were decided by the final arbitration. Of the 67 cases decided by arbitration, 33 are won by the player and 34 are won by the team.

Figure \ref{fig:outcome} shows the distribution of the outcome of the bargaining measured by the fraction of salary difference captured by the player. The outcome is fairly symmetrically distributed around .5, suggesting that the outcome does not systematically favor one side or another on aggregate and that negotiation is important.

\begin{figure}[ht]
\centering
\begin{subfigure}{0.48\textwidth}
\includegraphics[width=\textwidth]{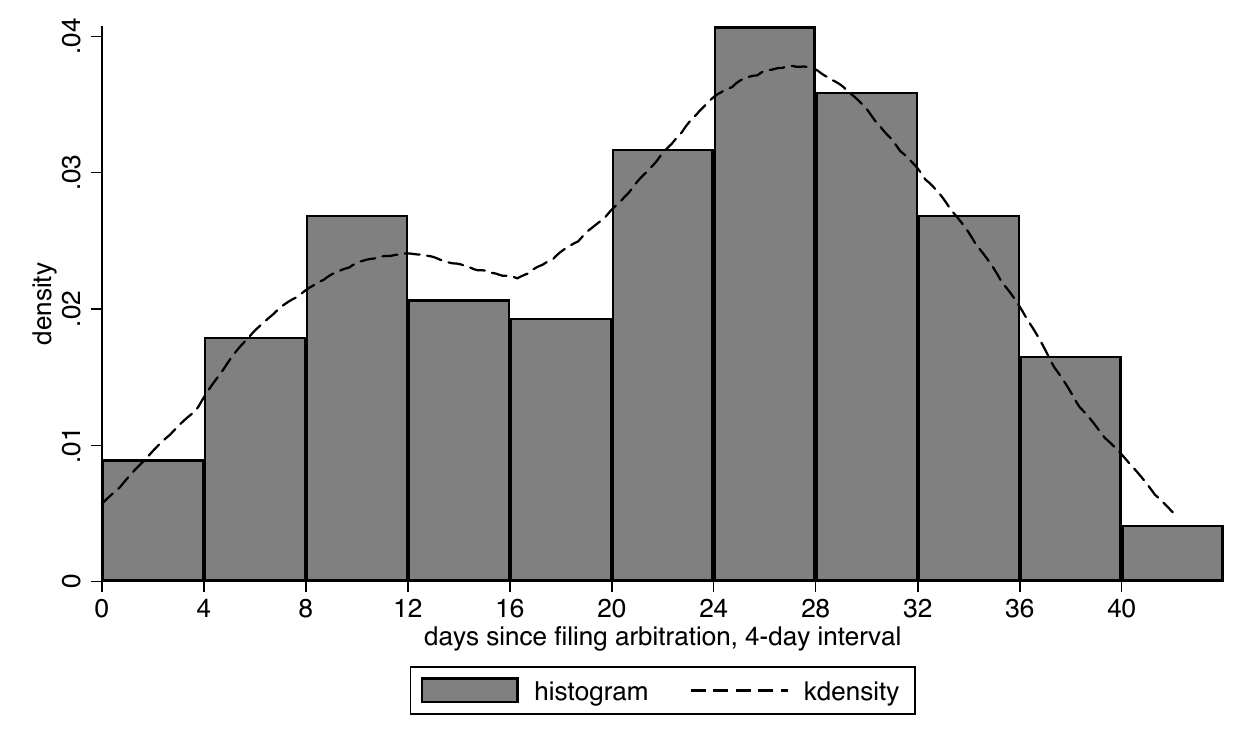}
\caption{\label{fig:density}Distribution of days to resolution.}
\end{subfigure}
~
\begin{subfigure}{0.48\textwidth}
\includegraphics[width=\textwidth]{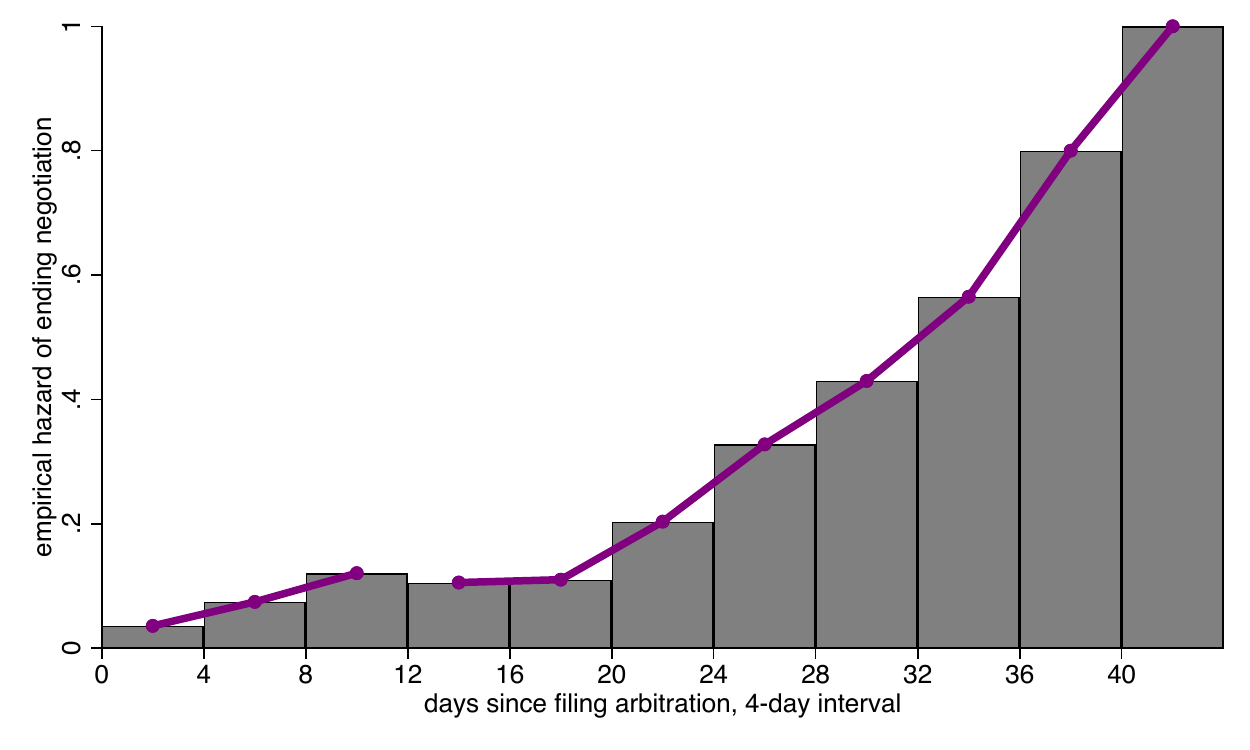}
\caption{\label{fig:hazard}Hazard rate of resolution, 4-day interval.}
\end{subfigure}
\caption{Distribution of days to resolution and hazard rates of resolution in MLB}
\end{figure}

\begin{figure}[ht]
\centering
\begin{subfigure}{0.48\textwidth}
\includegraphics[width=\textwidth]{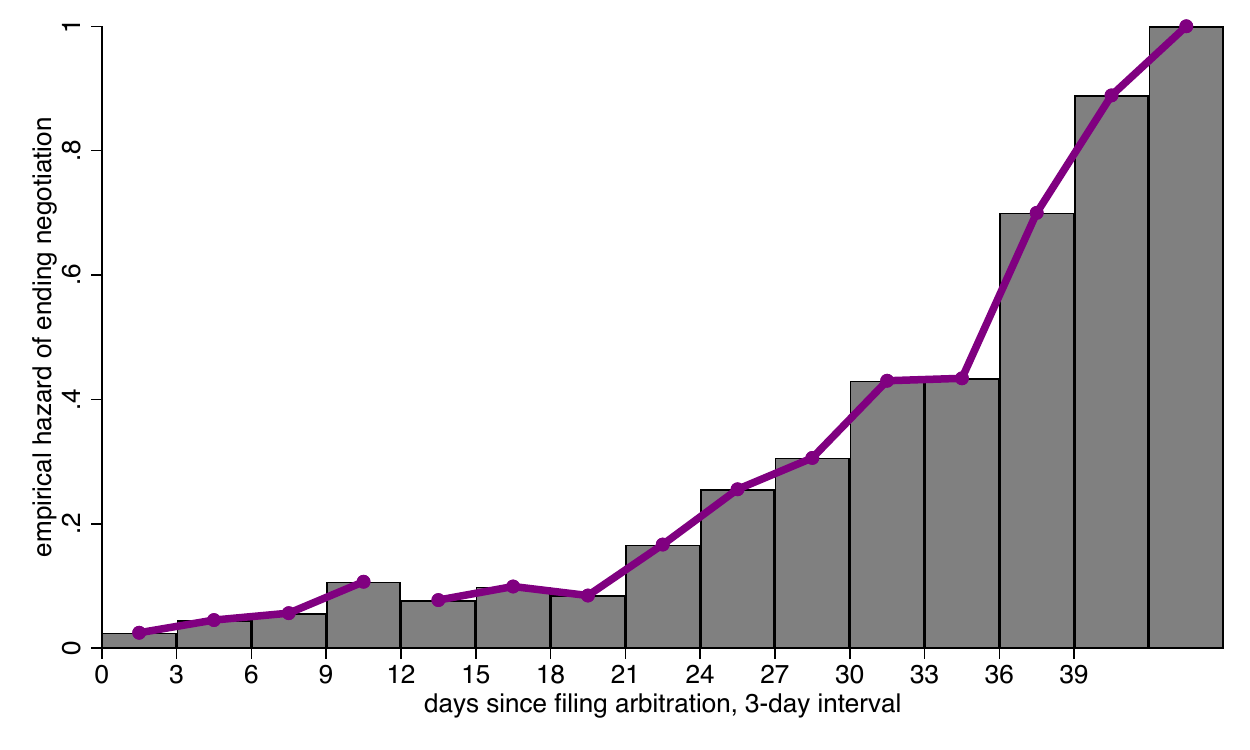}
\caption{\label{fig:hazardrobust3}Hazard rate of resolution, 3-day interval.}
\end{subfigure}
\begin{subfigure}{0.48\textwidth}
\includegraphics[width=\textwidth]{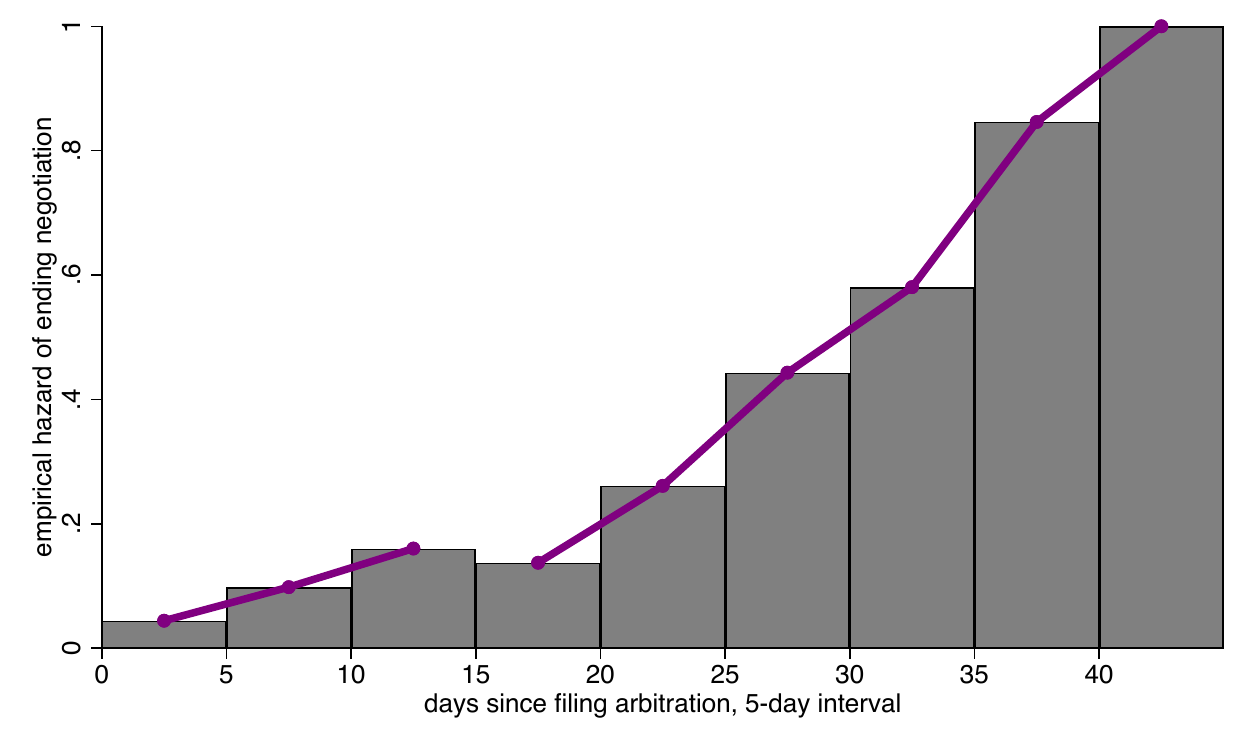}
\caption{\label{fig:hazardrobust5}Hazard rate of resolution, 5-day interval.}
\end{subfigure}
\caption{\label{fig:hazardrobust}Hazard rates of resolution by 3-day and 5-day intervals in MLB}
\end{figure}

The time it takes to reach an agreement ranges from 0 day to 39 days. Figure \ref{fig:density} illustrates the histogram and the kernel density of days to reach an agreement. We can see a dip in the negotiation after two weeks from frequency and kernel density.

A unique prediction of our model is the existence of a discontinuity in empirical hazard rates in reaching an agreement. Figure \ref{fig:hazard} illustrates the empirical hazard rate of the end of the negotiation. We can see a dip in the negotiation after approximately two weeks, from frequency, kernel density, as well as hazard rate. We can think of the first two weeks as the time interval for the player to challenge the team. Besides the dip, the hazard rates of end of the game are increasing in time. Alternative specifications\threeemdashes{}(i) specifying business days rather than calendar days, (ii) varying the number of days in a time interval from 3 to 5, (iii) considering only the negotiations that did not end with arbitration, (iv) excluding years 2017 and 2018 with abnormally high rates of arbitration\threeemdashes{}show the dip in hazard rate around 10 days to two weeks.  As a robustness check, Figures \ref{fig:hazardrobust} shows the empirical hazard rate resolution with 3 days and 5 days pooled. The discontinuous drop in hazard rates remains with these alternative specifications.

Ideally, more detailed data are available: (i) actual salary figures for extensions and (ii) the scheduled hearing dates even for the cases that did not go to hearing. We can also try to investigate when and why negotiation breaks in the cases decided by arbitration by comparing cases that avoided filing arbitration.

Salary arbitration was the product of collective bargaining agreement, and was a procedure insisted by the players' union. However, as our theoretical results show, it is unclear whether the ability to take the case to arbitration court benefits players. Arguably, more often than not, players who do not have solid evidence are hurt by the introduction of the arbitration procedure.

\begin{figure}[t!]
\begin{center}
\includegraphics[width=0.5\textwidth]{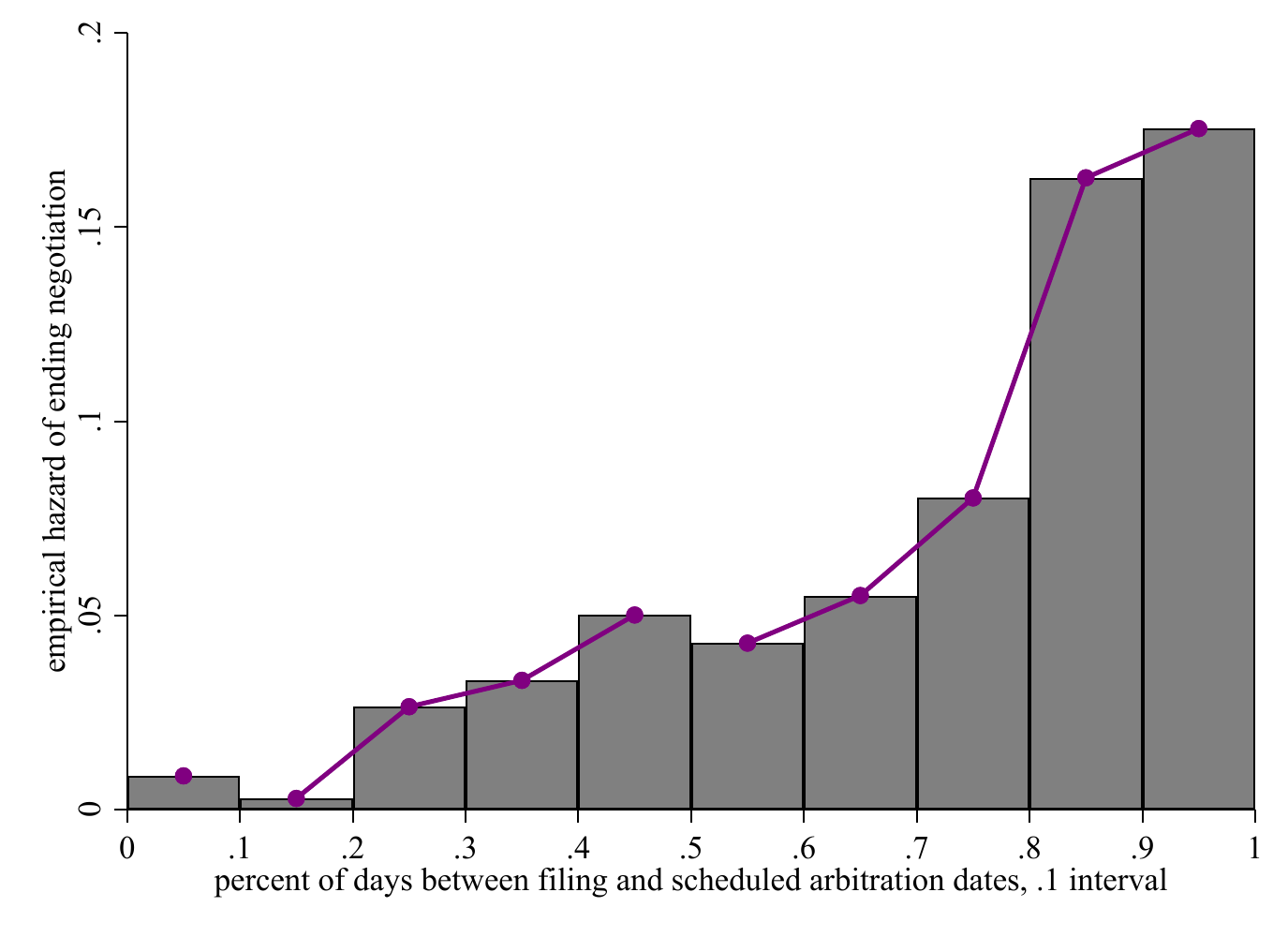}
\caption{\label{fig:nhl} Hazard rate of resolution in NHL arbitration cases from 1995 to 2020.}
\end{center}
\end{figure}

\subsection{National Hockey League}
The National Hockey League adopted a similar arbitration system in 1995 after its lockout and canceled season, the first canceled season in all four major league sports in North America. Since the process is designed to incentivize settlements and has proven that it works exceptionally well in forcing the parties to find a solution to the disputes, the process has been highly successful. The NHL has had far more success in creating a process and a culture that encourages settlements at a much higher percentage than that of MLB. We collect the filing dates, scheduled arbitration dates, and results from these negotiations from 1995 to 2020; the scheduled arbitration dates are publicly available for majority of cases because National Hockey League Players' Association publishes them prior to arbitration. We show that the empirical hazard rates also exhibit piecewise monotonicity in the midst of the negotiation phase, in addition to having peaks at times 0 and 1, as illustrated by Figure \ref{fig:nhl}.

\end{document}